\documentclass[12pt,a4paper]{article}
\usepackage[T1]{fontenc}
\usepackage[top=3cm,left=2.5cm,right=2.5cm,bottom=3cm]{geometry}
\usepackage{ragged2e}
\usepackage{lmodern}
\usepackage{multicol}
\usepackage{amsmath}
\usepackage{amsthm}
\usepackage{amsfonts}
\usepackage{thmtools}
\usepackage{mathtools}
\usepackage{amssymb}
\usepackage{bm}
\usepackage{dsfont}
\usepackage{enumitem}
\usepackage{graphicx}
\usepackage[dvipsnames]{xcolor}
\usepackage[bottom]{footmisc}
\usepackage{hyperref}
\usepackage[nodisplayskipstretch]{setspace}
\usepackage{subcaption}
\usepackage{float}

\allowdisplaybreaks

\renewcommand{\P}{\mathbb{P}}
\newcommand{\E}{\mathbb{E}}

\newcommand{\Int}{\mathbb{Z}}
\newcommand{\Real}{\mathbb{R}}
\newcommand{\Complex}{\mathbb{C}}

\newcommand{\1}[1]{\mathbb{I}\{#1\}}

\newcommand{\vertiii}[1]{{
    \left\vert\kern-0.25ex\left\vert\kern-0.25ex\left\vert #1 
    \right\vert\kern-0.25ex\right\vert\kern-0.25ex\right\vert
}}

\makeatletter
\renewcommand*\env@matrix[1][*\c@MaxMatrixCols c]{%
  \hskip -\arraycolsep
  \let\@ifnextchar\new@ifnextchar
  \array{#1}}
\makeatother

\makeatletter
\DeclareFontFamily{U}{tipa}{}
\DeclareFontShape{U}{tipa}{m}{n}{<->tipa10}{}
\newcommand{\arc@char}{{\usefont{U}{tipa}{m}{n}\symbol{62}}}%

\newcommand{\arc}[1]{\mathpalette\arc@arc{#1}}

\newcommand{\arc@arc}[2]{%
  \sbox0{$\m@th#1#2$}%
  \vbox{
    \hbox{\resizebox{\wd0}{\height}{\arc@char}}
    \nointerlineskip
    \box0
  }%
}
\makeatother

\definecolor{MATblue}{HTML}{0072BD}
\definecolor{MATgreen}{HTML}{77AC30}
\definecolor{MATorange}{HTML}{D95319}
\definecolor{MATred}{HTML}{A2142F}
\definecolor{MATyellow}{HTML}{EDB120}
\definecolor{MATcyan}{HTML}{4DBEEE}
\definecolor{MATpurple}{HTML}{7E2F8E}
\hypersetup{
	colorlinks=true,
	linkcolor=MATblue,
	citecolor=MATblue,
	urlcolor=MATblue
}

\theoremstyle{plain}
\declaretheorem[name=Definition,numberwithin=section]{definition}
\theoremstyle{plain}
\declaretheorem[name=Theorem,numberwithin=section]{theorem}
\declaretheorem[name=Corollary,numberwithin=section]{corollary}
\declaretheorem[name=Proposition,numberwithin=section]{proposition}
\declaretheorem[name=Lemma,numberwithin=section]{lemma}
\makeatletter
\renewenvironment{proof}[1][\proofname]{\par
  \pushQED{\qed}%
  \normalfont \topsep6\p@\@plus6\p@\relax
  \trivlist
  \item[\hskip\labelsep\bfseries
    #1\@addpunct{.}]\ignorespaces
}{%
  \popQED\endtrivlist\@endpefalse
}
\makeatother
\theoremstyle{definition}
\declaretheorem[name=Assumption]{assumption}

\declaretheorem[name=Remark,numberwithin=section]{remark}
\declaretheorem[name=Example,numberwithin=section]{example}

\newtheorem{assumptionalt}{Assumption}[assumption]

\newenvironment{assumptionp}[1]{
  \renewcommand\theassumptionalt{#1}
  \assumptionalt
}{\endassumptionalt}

\usepackage[authoryear,round,comma]{natbib}
\begin{document}


\onehalfspacing
	
\title{
    Impulse Response Analysis of Structural Nonlinear Time Series Models
}
\author{Giovanni Ballarin%
	\thanks{
        E-mail: \texttt{\href{mailto:giovanni.ballarin@unisg.ch}{giovanni.ballarin@unisg.ch}}.
        I thank Otilia Boldea, Timo Dimitriadis, Juan Carlos Escanciano, Lyudmila Grigoryeva, Klodiana Istrefi, Marina Khismatullina, So Jin Lee, Yuiching Li, Sarah Mouabbi, Andrey Ramirez, Christoph Rothe, Carsten Trenkler and Mengshan Xu, as well as the participants of the Econometrics Seminar at the University of Mannheim, the 2023 ENTER Jamboree, the 10th HKMEtrics Workshop, the GSS Weekly Seminar at Tilburg University, the Internal Econometrics Seminar at Vrije Universiteit Amsterdam, the Brown Bag Seminar at the University of St. Gallen and the 2024 ESAM for their comments, suggestions and feedback.
        A significant part of this work was developed at the University of Mannheim thanks to the support of the Center for Doctoral Studies in Economics and the Chair of Empirical Economics.
    }\\ %
	University of St. Gallen
}
\date{
    \today
}

\makeatletter
\let\thetitle\@title
\let\theauthor\@author
\let\thedate\@date
\makeatother
	
\maketitle
	
\renewcommand{\abstractname}{\vspace{-\baselineskip}} 
\vspace{-2em}
\begin{abstract}
	\textbf{Abstract:}
	This paper proposes a semiparametric sieve approach to estimate impulse response functions of
	nonlinear time series within a general class of structural autoregressive models. We prove that a two-step procedure can flexibly accommodate nonlinear specifications while avoiding the need to choose fixed parametric forms. Sieve impulse responses are proven to be consistent by deriving uniform estimation guarantees, and an iterative algorithm
	makes it straightforward to compute them in practice. With simulations, we show that the proposed
	semiparametric approach proves effective against misspecification while suffering only from minor efficiency
	losses. In a U.S. monetary policy application, the pointwise sieve GDP response associated
	with an interest rate increase is larger than that of a linear model. Finally, in an analysis of interest rate uncertainty shocks, sieve responses indicate more substantial
	contractionary effects on production and inflation.
\end{abstract}
	

	
\pagebreak

\doublespacing
\renewcommand*{\arraystretch}{0.7}

\section{Introduction}

Linearity is a foundational assumption in structural time series modeling.
For example, large classes of macroeconomic models in modern New Keynesian theory can be linearized, justifying the use of the linear time series toolbox from a theoretical point of view. The seminal work of \cite{simsMacroeconomicsReality1980} on vector autoregressive (VAR) models brought the study of dynamic economic relationships into focus within the macro-econometric literature, for which the estimation and analysis of impulse response functions (IRFs) is key \citep{hamilton1994TimeSeriesAnalysis,lutkepohlNewIntroductionMultiple2005,kilianStructuralVectorAutoregressive2017}.
The local projection (LP) approach of \cite{jordaEstimationInferenceImpulse2005} has also gained popularity as a flexible and easy-to-implement alternative.

Linear econometrics models are, however, limited in the kind of effects that they can describe. 
In nonlinear DGPs, linear VARs as well as standard LP methods can only reconstruct the best linear impulse responses approximation \citep{plagborg-mollerLocalProjectionsVARs2021}.
And even though asymmetries in monetary policy and non-proportional shock effects are now commonly studied, most works still rely on parametric specifications. For example, \cite{tenreyro2016pushing} study both sign and size effects of monetary policy (MP) shocks using censoring and cubic transformations, respectively. 
\cite{caggianoEconomicPolicyUncertainty2017,pellegrinoUncertaintyMonetaryPolicy2021} and \cite{caggianoUncertaintyShocksGreat2021} use multiplicative interacted VAR models to estimate the effects of uncertainty and MP shocks.
From a macro-finance perspective, \cite{forniNonlinearTransmissionFinancial2023,forniAsymmetricEffectsNews2023} study the economic effects of financial shocks following a quadratic VMA specification  \citep{debortoliAsymmetricEffectsMonetary2020}.
\cite{gambettiBadNewsGood2022} study news shocks asymmetries by imposing that news changes enter an autoregressive model through a threshold map. Parametric nonlinear specifications are also common prescriptions in time-varying models \citep{auerbachMeasuringOutputResponses2012,caggianoEstimatingFiscalMultipliers2015} and state-dependent models \citep{ramey2018government,goncalvesStatedependentLocalProjections2024}.

In this paper, we aim to design a semiparametric, structural nonlinear time series modeling and estimation framework with explicit theoretical properties.
Our primary contribution is the extension and combination of the block-recursive structural framework of \cite{goncalvesImpulseResponseAnalysis2021} with the uniform sieve estimation theory of \cite{chenOptimalUniformConvergence2015} within a general physical dependence setup \citep{wuNonlinearSystemTheory2005}.
Under classical nonparametric assumptions, we show that a two-step semiparametric series estimation procedure can consistently recover the structural model in a uniform sense. 
In order to be able to relax the assumptions on the linearity and fixed parametric model specifications, we restrict our study to the case of compactly supported, weakly dependent data.
We emphasize that, even in this constrained setting, to the best of our knowledge, our work is the first to offer a formal combination of these approaches.
We provide explicit guarantees for semiparametrically-estimated nonlinear IRFs: Nonlinear impulse response function estimates are asymptotically consistent and, thanks to an iterative algorithm, can also be straightforwardly computed.  

To illustrate the validity of our proposed methodology, we first evaluate its performance with several simulations. With realistic sample sizes, the efficiency costs of our semiparametric procedure are small compared to correctly specified parametric responses. A second set of simulations provides a simple setup where the nonlinear parametric model is mildly misspecified. Still, the large-sample bias is considerable, while for semiparametric estimates it is negligible. 
We then evaluate how the IRFs computed using the new method compare with those from two empirical exercises studied in the literature. 
In a small, quarterly model of the U.S. macroeconomy based on \cite{tenreyro2016pushing}, we find that point estimates of linear and parametric nonlinear IRFs may underestimate in intensity the GDP responses by up to $13\%$ and $16\%$, respectively, after a large exogenous monetary policy shock. Moreover, sieve responses achieve maximum impact a year before their linear counterparts. 
Then, we evaluate the effects of interest rate uncertainty on US output, prices, and unemployment following \cite{istrefiSubjectiveInterestRate2018}. In this exercise, the impact on industrial production of a one-deviation increase in uncertainty is $54\%$ stronger according to semiparametric IRFs than the comparable linear specification. These findings suggest that structural responses based on linear specifications can significantly underestimate the effects of shocks.

\paragraph{Literature Review.}
Let us mention some key references directly related to our discussion.
On the one hand, \cite{jordaEstimationInferenceImpulse2005} already proposed a ``{flexible local projection}'' approach based on the Volterra expansion. 
The flexible LP proposal is effectively equivalent to adding polynomial terms to a linear regression, meaning it is a semiparametric method, and it should be analyzed as such. Yet, the Volterra expansion is not formally justified, nor is its truncation, which is key in studying its properties \citep{SirotkoSibirskaya2020volterraBootstrap,Movahedifar2023closedloopVolterra}.
On the other hand, smoothed LP methods, see e.g. \cite{plagborg-moller2016essays,barnichonImpulseResponseEstimation2019}, can address exclusively concerns of regularity in the shape of estimated impulse responses, but not any potential underlying nonlinearities in the DGP.
Recently, \cite{goncalvesStatedependentLocalProjections2024} outlined a general, nonparametric LP estimation procedure for nonlinear IRFs, which was later studied in \cite{goncalvesNonparametricLocalProjections2024} under high-level conditions on the functional form of the IRF itself.
\cite{gourieroux2023nonlinear} also devise a framework for nonparametric kernel estimation and inference of IRFs via local projections, although they mostly work in the one-dimensional, single lag case.
Finally, following the Generalized IRF (GIRF) approach \citep{Koop1996,potterNonlinearImpulseResponse2000,gourierouxNonlinearInnovationsImpulse2005,terasvirtaModellingNonlinearEconomic2010}, \cite{kanazawaRadialBasisFunctions2020} proposed to use radial basis function neural networks to estimate nonlinear reduced-form GIRFs for the U.S. economy. While GIRFs can be essentially characterized as impulse responses with more sophisticated conditioning sets, they are lacking in that they do not inherently address the problem of {structural} identification \citep{kilianStructuralVectorAutoregressive2017}.

\paragraph{}
The remainder of this paper is organized as follows. Section~\ref{section:framework} provides the general framework for the structural model. 
Section~\ref{section:estimation} describes the two-step semiparametric estimation strategy, and Section~\ref{section:impulse_response_analysis} discusses nonlinear impulse response function computation, validity, and consistency.
In Section~\ref{section:simulations} we give a brief overview of simulation results, while Section~\ref{section:applications} contains the empirical analyses. Finally, Section~\ref{section:conclusion} concludes. All proofs and additional content can be found in the Appendix.
Concerning notation: scalar and vector random variables are denoted in capital or Greek letters, e.g. $Y_t$ or $\epsilon_t$, while realizations are shown in lowercase Latin letters, e.g. $y_t$. For a process $\{Y_t\}_{t\in \Int}$, we write $Y_{t:s} = (Y_t, Y_{t+1}, \ldots, Y_{s-1}, Y_s)$, as well as $Y_{*:t} = (\ldots, Y_{t-2}, Y_{t-1}, Y_t)$ for the left-infinite history and $Y_{t:*} = (Y_t, Y_{t+1}, Y_{t+2}, \ldots)$ for its right-infinite history. The same notation is also used for random variable realizations.
For a matrix $A \in \Real^{d \times d}$ where $d \geq 1$, $\lVert A \rVert$ is the spectral norm, $\lVert A \lVert_\infty$ is the supremum norm and $\lVert A \rVert_r$ for $0 < r < \infty$ is the $r$-operator norm. For a random vector or matrix, we will use $\lVert \,\cdot\, \rVert_{L^r}$ to denote the associated $L^r$ norm.

%
%

\section{Model Framework}\label{section:framework}

In this section, we introduce the general nonlinear time series model, which is a generalization of the one developed in \cite{goncalvesImpulseResponseAnalysis2021}. 
In terms of structural shocks identification, the idea is straightforward: A scalar series, $X_t$, is chosen to be the \textit{structural variable} identifying shocks, and it explicitly determines the dynamic effects on the remaining data, vector $Y_t$. This will enable the derivation of economically meaningful (structural) impulse responses due to an exogenous shock impacting $X_t$.

\subsection{General Model}\label{section:model_setup}

This paper focuses on the family of nonlinear autoregressive models of the form
\begin{equation}\label{eq_main:general_nonlin_model}
	\begin{split}
		X_t & = \mu_1 + A_{12}(L) Y_{t-1} + A_{11}(L) X_{t-1} + u_{1t},  \\ 
		Y_t & = \mu_2 + G_2(Y_{t-1}, \ldots, Y_{t-p}, X_t, X_{t-1}, \ldots, X_{t-p}) + u_{2t} .
	\end{split}
\end{equation}
where $X_t \in \mathcal{X} \subseteq \Real$ and $Y_t \in \mathcal{Y} \subseteq \Real^{d_Y}$ are scalar and $d_Y$-dimensional time series, respectively, $u_t = (u_{1t}, u_{2t}')' \in \mathcal{U} \subseteq \Real^d$ are innovations, $d = 1 + d_Y$, $G_2 : \Real^{1 + pd} \to \Real$ is a generic nonlinear map, and $A_{12}(L)$ and $A_{11}(L)$ are lag polynomials \citep{lutkepohlNewIntroductionMultiple2005}. We let $Z_t := (X_t, Y_t')' \in \Real^{d}$ be the full data vector.
Let us provide some examples for the model classes nested by  \eqref{eq_main:general_nonlin_model}.

\begin{example}[Linear VAR]
	In the simplest case, $G_2(Y_{t-1}, \ldots, Y_{t-p}, \allowbreak X_t, X_{t-1}, \ldots, \allowbreak X_{t-p}) = A_{22}(L) Y_{t-1} + A_{21}(L) X_{t}$, and we recover the class of linear vector autoregressive models. 
\end{example}


\begin{example}[Additively separable model]
	When $G_2(Y_{t-1}, \ldots, Y_{t-p}, \allowbreak X_t, X_{t-1}, \ldots, \allowbreak X_{t-p}) = \sum_{i=1}^p G_{i,22}(Y_{t-i}) + \sum_{j=0}^p G_{j,21}(X_{t-j})$, model \eqref{eq_main:general_nonlin_model} is additively separable \citep{fan2003nonlinear}.
\end{example}

\begin{example}[Nonlinear impact model]\label{example:nonlin_impact}
	A parsimonious semiparametric class, which may be informally termed the ``{nonlinear impact model class}'', involves specification
	\begin{equation*}
		\begin{split}
			Y_t & = \mu_2 + A_{22}(L) Y_{t-1} + \sum_{j=0}^p G_{j,21}(X_{t-j}) + u_{2t} ,
		\end{split}
	\end{equation*}
	see e.g. \cite{goncalvesImpulseResponseAnalysis2021}. 
	An equivalent representation for $Y_t$ is 
	\begin{equation*}
		Y_t = \mu_2 + A_{22}(L) Y_{t-1} + A_{21}(L) X_{t-1} + \sum_{j=0}^p \arc{G}_{j,21}(X_{t-j}) +  u_{2t} ,
	\end{equation*}
	where now to identify nonlinear functions $\arc{G}_{j,21} : \Real \to \Real^{d_Y}$, $0 \leq j \leq p$, we require that constant and linear factors be not included at indices $j \geq 1$. To make this more compact, write $Z_t = \mu + A(L) Z_{t-1} + \arc{G}(L) X_t + u_t$, where
	\begin{equation*}
		A(L) := 
		\begin{bmatrix}
			A_{11}(L) & A_{12}(L) \\
			A_{21}(L) & A_{22}(L)
		\end{bmatrix}
		\quad\textnormal{and}\quad
		\arc{G}(L) := 
		\begin{bmatrix}
			0 \\
			\arc{G}_{0,21} + \arc{G}_{1,21} L + \ldots + \arc{G}_{p,21} L^p
		\end{bmatrix} ,
	\end{equation*}
	with the minor abuse of notation that $\arc{G}_2(L) := \arc{G}_{0,21} + \allowbreak \ldots + \arc{G}_{p,21} L^p$ is now intended as a \textit{functional} lag polynomial, meaning $\arc{G}_2(L) X_t \equiv \sum_{j=0}^p \arc{G}_{j,21}(X_{t-j})$.\footnote{The choice to use a functional matrix notation is due to the ease of writing multivariate additive nonlinear models such as \eqref{eq_main:structural_model} in a manner consistent with standard formalisms of linear VAR models, following again e.g. \cite{lutkepohlNewIntroductionMultiple2005}.}  
	
	
\end{example}


\subsection{Structural Framework}

Model \eqref{eq_main:general_nonlin_model} involves only reduced-form innovations $u_{1t}$ and $u_{2t}$, and additional assumptions are necessary to provide a structural interpretation. Many such assumptions have been devised in the macroeconomic literature, but few can be directly applied to nonlinear models \citep{kilianStructuralVectorAutoregressive2017}. We follow the block-recursive identification strategy outlined in \cite{goncalvesImpulseResponseAnalysis2021} and originally due to \cite{kilian2011responses}. 

From \eqref{eq_main:general_nonlin_model} we derive   
\begin{equation*}
	\begin{split}
		X_t & = \mu_1 + A_{12}(L) Y_{t-1} + A_{11}(L) X_{t-1} + u_{1t},  \\ 
		Y_t & = \mu_2 + A_{22}(L) Y_{t-1} + A_{21}(L) X_{t-1} + \arc{G}_2(Y_{t-1:t-p}, X_{t:t-p}) + u_{2t} ,
	\end{split}
\end{equation*}
where, without loss of generality, we have assumed (as in Example~\ref{example:nonlin_impact}) that we can separate the linear and nonlinear ($\arc{G}_2$) components from $G_2$. In general, it can be the case that $\mu_2 = 0$, $A_{22}(L) = 0$ or $A_{21}(L) = 0$ if e.g. $G_2$ is strictly nonlinear.
In vector form: 
\begin{equation}\label{eq_main:pseudo_reduced_form_single_eq}
	Z_t = \mu + A(L) Z_{t-1} + \arc{G}(Z_{t:t-p}) + u_t ,
	\quad\textnormal{where}\quad
	\arc{G}(Z_{t:t-p}) :=
	\begin{bmatrix}
		0 \\
		\arc{G}_2(Y_{t-1:t-p}, X_{t:t-p})
	\end{bmatrix}
	.
\end{equation}
We can now formalize the structural specification of our model.

\begin{assumption}\label{assumption:structural_model}
	There exist (i) a vector $B_0^{21} \in \Real^{d_Y}$ and a matrix $B_0^{22} \in \Real^{d_Y \times d_Y}$ such that 
	\begin{equation*}
		\begin{bmatrix}
			1 & 0 \\
			B_0^{21} & B_0^{22} 
		\end{bmatrix} 
		=:
		B_0^{-1} 
	\end{equation*}
	is invertible and has unit diagonal, and (ii) mutually independent innovations sequences $\{\epsilon_{1t}\}_{t \in \Int}$, $\epsilon_{1t} \in \mathcal{E}_1 \subseteq \Real$, and $\{\epsilon_{2t}\}_{t \in \Int}$, $\epsilon_{2t} \in \mathcal{E}_2 \subseteq \Real^{d_Y}$, such that
	\begin{equation*}
		\begin{bmatrix}
			\epsilon_{1t} \\
			\epsilon_{2t}
		\end{bmatrix}
		\:\overset{\text{i.i.d.}}{\sim}\:
		\left(
		\begin{bmatrix}
			0 \\
			0 
		\end{bmatrix} ,
		\begin{bmatrix}
			\Sigma_1 & 0 \\
			0 & \Sigma_2
		\end{bmatrix}
		\right) ,
	\end{equation*}
	where $\Sigma_1 > 0$ and $\Sigma_2$ is a diagonal positive definite matrix so that 
	\begin{equation}\label{eq_main:pseudo_reduced_form_structural_model}
		\begin{split}
			X_t & = \mu_1 + A_{12}(L) Y_{t-1} + A_{11}(L) X_{t-1} + \epsilon_{1t},  \\ 
			Y_t & = \mu_2 + A_{22}(L) Y_{t-1} + A_{21}(L) X_{t-1} + \arc{G}_2(Y_{t-1:t-p}, X_{t:t-p}) + B_0^{21} \epsilon_{1t} + B_0^{22} \epsilon_{2t} ,
		\end{split}
	\end{equation}
	where $u_{1t} \equiv \epsilon_{1t}$, $u_{2t} := B_0^{21} \epsilon_{1t} + B_0^{22} \epsilon_{2t}$ and thus $u_t = B_0^{-1}\epsilon_t$ for $\epsilon_t = (\epsilon_{1t}, \epsilon_{2t}')' \in \mathcal{E} \subseteq \Real^d$. 
\end{assumption}

\begin{remark}
	Assumption~\ref{assumption:structural_model} follows \cite{goncalvesImpulseResponseAnalysis2021} closely. 
	By design, one does not need to identify the model fully, meaning that fewer assumptions on $Z_t$ and $\epsilon_t$ are needed to estimate the individual structural effects of $\epsilon_{1t}$ on $Y_t$. This comes at the price of not being able to simultaneously study structural effects for shocks impacting $\epsilon_{2t}$.
\end{remark}

Note that inverting $B_0^{-1}$ gives
\begin{equation*}
	B_0 = \begin{bmatrix}
		1 & 0 \\
		-B_{0,12} & B_{0,22} 
	\end{bmatrix} ,
\end{equation*}
and, multiplying both sides of \eqref{eq_main:pseudo_reduced_form_single_eq} by $B_0$, we find 
\begin{equation}\label{eq_main:structural_model}
	B_0 Z_t = b + B(L) Z_{t-1} + \arc{F}(Z_{t:t-p}) + \epsilon_t ,
\end{equation}
where $b = (b_1, b_2')' \in \Real^d$ and $\arc{F}(Z_{t:t-p}) = (0, \arc{F}_2(Z_{t:t-p}))'$ for $\arc{F}_2 : \Real^{1+pd_Y} \to \Real^{d_Y}$, $F_2 = B_{0,22} \arc{G}_2$.
In practice, to estimate the model's coefficients, we will leverage \eqref{eq_main:pseudo_reduced_form_structural_model}. This latter form was termed the \textit{pseudo-reduced form} by \cite{goncalvesImpulseResponseAnalysis2021}. 

Observe that $\arc{G}_2(Y_{t-1:t-p}, X_{t:t-p})$ is correlated with $u_{2t}$ through $B_0^{21} \epsilon_{1t}$. As $X_t$ depends linearly on $\epsilon_{1t}$, if $B_0^{21} \not= 0$ and $\arc{G}_2(Y_{t-1:t-p}, X_{t:t-p})$ is not independent of $X_t$, there is an endogeneity problem. 
%
\cite{goncalvesImpulseResponseAnalysis2021} address the issue by proposing a two-step estimation procedure wherein one proxies for $\epsilon_{1t}$ with residual $\widehat{\epsilon}_{1t}$. As we prove in Section~\ref{section:estimation} below, this approach also generally allows for consistent semiparametric estimation.

\begin{remark}{(Moving Average Identification).}
	\cite{forniNonlinearTransmissionFinancial2023,forniAsymmetricEffectsNews2023} work with an alternative nonlinear structural identification framework to the block-recursive form. Their approach follows \cite{debortoliAsymmetricEffectsMonetary2020} and is based on a vector MA representation. Under appropriate assumptions, the structural model studied by \cite{forniNonlinearTransmissionFinancial2023} is
	\begin{equation}\label{eq_main:debortoli_forni_structural_form}
		Z_t = \mu + A(L) Z_t + Q_0 F(\epsilon_{1t}) + B_0 \epsilon_t ,
	\end{equation}
	where $\epsilon_t$ are independent innovations with zero mean and identity covariance, and $\epsilon_{1t}$ identifies the shocks of interest. $Q(L)$ and $B(L)$ are both linear lag polynomials, and $F(x) = x^2$ in their baseline specification.
	For \eqref{eq_main:debortoli_forni_structural_form} to overlap with \eqref{eq_main:pseudo_reduced_form_structural_model} one must impose that (i) $X_t$ is exogenous and independently distributed and (ii) only $\epsilon_{1t}$ has nonlinear effects.
	We emphasize that, if the innovation sequence $\epsilon_{1t}$ is assumed to be observable, applying our results to the framework of \cite{debortoliAsymmetricEffectsMonetary2020} is straightforward.
\end{remark}

\subsection{Structural Nonlinear Impulse Responses}\label{section:nonlin_irfs}

Starting from pseudo-reduced equations \eqref{eq_main:pseudo_reduced_form_structural_model}, we begin by assuming that the linear autoregressive component is stable.

\begin{assumption}\label{assumption:roots_linear_part_model}
	The roots of $\det(I_{d} - A(L) L) = 0$ are outside the complex unit circle.
\end{assumption}

This standard stability assumption enables us to write impulse responses in a manner that can yield useful simplifications for additively separable models.\footnote{Stability of the linear VAR component is neither necessary nor sufficient for ensuring stability and stationarity of the entire nonlinear process, cf. Assumption~\ref{assumption:physical_dep} in Section~\ref{section:estimation} below.}
%
%
Then, letting $\Psi(L) = (I_{d} - A(L) L)^{-1}$, one can write
\begin{equation}
	Z_t = \eta + \Theta(L) \epsilon_t + \Gamma(Z_{t:*}) ,
\end{equation}
where
$\eta := \Psi(1) (\mu_1, \mu_2')'$,
$\Theta(L) := \Psi(L) B_0^{-1}$ and
$\Gamma(Z_{t:*}) := \Psi(L)(0, \arc{G}_2(Y_{t-1:t-p}, X_{t:t-p})')'$.
%
%
We emphasize that the nonlinear term $\Gamma(Z_{t:*})$ generally depends on the entire past of the process $Z_t$, as $\Psi(L)$ is an infinite-order MA polynomial.
To define impulse responses, we partition the polynomial $\Theta(L)$ according to
$\Theta(L) := [\Theta_{\cdot 1}(L) \:\vert\: \Theta_{\cdot 2}(L)]$ ,
%
%
where $\Theta_{\cdot 1}(L)$ represents the first column of matrices in $\Theta(L)$, and $\Theta_{\cdot 2}(L)$ the remaining $d_Y$ columns.

Given impulse $\delta \in \Real$ at time $t$, define the shocked innovation process as $\epsilon_{1 s}(\delta) = \epsilon_s$ for $s \not= t$ and $\epsilon_{1 t}(\delta) = \epsilon_{1 t} + \delta$, as well as the shocked structural variable as $Z_s(\delta) = Z_s$ for $s < t$ and $Z_s(\delta) = X_s(\epsilon_{s:t+1}, \epsilon_t + \delta, \epsilon_{t-1:*})$ for $s \geq t$. Further, for a given horizon $h \geq 0$, let
\begin{align*}
	Z_{t+h} & := \eta + \Theta_{\cdot 1}(L) \epsilon_{1 t+h} + \Theta_{\cdot 2}(L) \epsilon_{2 t+h} + \Gamma(Z_{t:*}) , \\[2pt]
	Z_{t+h}(\delta) & := \eta + \Theta_{\cdot 1}(L) \epsilon_{1 t+h}(\delta) + \Theta_{\cdot 2}(L) \epsilon_{2 t+h} + \Gamma(Z_{t:*}(\delta)) , 
\end{align*}
be the time-$t$ baseline and shocked series, respectively. 
Then,
\begin{equation}\label{eq_main:def_irf}
	\textnormal{IRF}_{h}(\delta) = \E\left[ Z_{t+h}(\delta) - Z_{t+h}  \right] 
\end{equation}
is the unconditional impulse response at horizon $h$ due to shock $\delta$. The difference between series is
$ Z_{t+h}(\delta) - Z_{t+h} 
= \Theta_{h,\cdot 1} \delta + \Gamma(Z_{t:*}(\delta)) - \Gamma(Z_{t:*}) $, 
%
%
hence
\begin{equation}\label{eq_main:irf_h}
	\textnormal{IRF}_{h}(\delta) 
	= \Theta_{h,\cdot 1} \delta + \E\left[ \Gamma(Z_{t:*}(\delta)) - \Gamma(Z_{t:*}) \right] .
\end{equation}

\begin{remark}
	In additively separable models, $\Gamma(Z_{t:*})$ is also additively separable over lags of $Z_t$. Accordingly, the baseline and shock series have an additive form,
	as terms with time indices $s < t$ remain unaffected by the shock.
	Therefore, \eqref{eq_main:irf_h} reduces to
	\begin{equation}\label{eq_main:irf_h_separable}
		\textnormal{IRF}_{h}(\delta) 
		= \Theta_{h,\cdot 1} \delta + \E\left[ \Gamma_0(Z_{t+h}(\delta)) - \Gamma_0(Z_{t+h}) \right] + \ldots + \E\left[ \Gamma_h(Z_{t}(\delta)) - \Gamma_h(Z_{t}) \right] .
	\end{equation}
	Coefficients $\Gamma_j$ are still functional, and cannot be collected across $X_{t+j}(\delta)$ and $X_{t+j}$.
\end{remark}

Closed-form computation of nonlinear IRFs is highly non-trivial. Even in the separable case \eqref{eq_main:irf_h_separable}, while one can linearly separate expectations in the impulse response formula, terms $\E\left[ \Gamma_j(Z_{t+j}(\delta)) - \Gamma_j(Z_{t+j}) \right]$ for $0 \leq j \leq h$ cannot be meaningfully simplified further. Moreover, these expectations involve nonlinear functions of lags of $Z_t$ and are impractical to derive explicitly. 
%
%
To avoid working with $\Theta(L)$ and $\Gamma(L)$, we now present an iterative algorithm which allows one to easily and efficiently compute nonlinear IRFs.\footnote{The algorithm we propose is a natural counterpart to the one in Proposition 3.1 of \cite{goncalvesImpulseResponseAnalysis2021}, wherein they suggest to estimate the MA form coefficients recursively. Our approach instead relies on directly iterating forward the model's equations, which is more computationally straightforward.}

\begin{proposition}\label{prop:irf_iterate_algorithm}
	For any $h = 0, 1, \ldots, H$, with $H \geq 1$ fixed, if impulse response $\textnormal{IRF}_{h}(\delta)$ is finite and well-defined, it can be computed with the following steps:
	\begin{description}
		\item[($\text{i}$)] For $j = 0$, let $X_t(\delta) = X_t + \delta$ and
		$Y_{t}(\delta) = \mu_2 + G_2(Y_{t-1}, \ldots, Y_{t-p}, X_t(\delta), X_{t-1}, \ldots, X_{t-p}) + B_0^{21} (\epsilon_{1t} + \delta) + \xi_{2t}$, where $\xi_{2t} = B_0^{22} \epsilon_{2t}$.
		%
		%
		\item[($\text{ii}$)] For $j = 1, \ldots, h$, let
		\begin{equation*}
			\begin{split}
				X_{t+j}(\delta) & = \mu_1 + A_{12}(L) Y_{t+j-1}(\delta) + A_{11}(L) X_{t+j-1}(\delta) + \epsilon_{1t+j} ,  \\ 
				Y_{t+j}(\delta) & = \mu_2 + G_2(Y_{t-1}(\delta), \ldots, Y_{t-p}(\delta), X_t(\delta), X_{t-1}(\delta), \ldots, X_{t-p}(\delta)) + B_0^{21} \epsilon_{1t+j} + \xi_{2t+j} .
			\end{split}
		\end{equation*}
		where $X_{t}(\delta)$ and $Y_{t}(\delta)$ are the shocked sequences determined by forward iteration after time $t$, equaling baseline sequences $X_{t}$ and $Y_{t}$ at lags before $t$, respectively. 
	\end{description}
	Setting $Z_{t+j}(\delta) = ( X_t(\delta), Y_t(\delta) )'$, it holds $\textnormal{IRF}_h(\delta) = \E[ Z_{t+j}(\delta) - Z_{t+j} ]$.
\end{proposition}

Proposition~\ref{prop:irf_iterate_algorithm} follows directly from the definition of the unconditional impulse response \eqref{eq_main:def_irf} combined with a direct forward iteration of \eqref{eq_main:pseudo_reduced_form_structural_model}, sidestepping the explicit MA($\infty$) formulation in \eqref{eq_main:irf_h}.  This approach dispenses from the need to simulate innovations $\{\epsilon_{t+j}\}_{j=1}^{h-1}$, as the joint distribution of $\{X_{t+h-1}, X_{t+j-1}, \ldots, X_{t}\}$ contains all relevant path information. 

When the model is estimated from data, for residuals $\widehat{\epsilon}_{1t}$ and $\widehat{\xi}_{2t}$ it trivially holds
\begin{equation*}
	\begin{split}
		X_{t} & = \widehat{\mu}_1 +  \widehat{A}_{12}(L) Y_{t-1} +  \widehat{A}_{11}(L) X_{t-1} + \widehat{\epsilon}_{1t},  \\ 
		Y_{t} & = \widehat{\mu}_2 + \widehat{G}_2(Y_{t-1}, \ldots, Y_{t-p}, X_t, X_{t-1}, \ldots, X_{t-p}) + \widehat{B}_0^{21} \widehat{\epsilon}_{1t} + \widehat{\xi}_{2t} .
	\end{split}
\end{equation*}
In practice, this means that one can numerically construct the shocked sequence as
\begin{equation*}
	\begin{split}
		\widehat{X}_{t+j}(\delta) 
		& = \widehat{\mu}_1 +  \widehat{A}_{12}(L) \widehat{Y}_{t+j-1}(\delta) +  \widehat{A}_{11}(L) \widehat{X}_{t+j-1}(\delta) + \widehat{\epsilon}_{1t+j},  \\ 
		\widehat{Y}_{t+j}(\delta) 
		& = \widehat{\mu}_2 + \widehat{G}_2(\widehat{Y}_{t-1}(\delta), \ldots, \widehat{Y}_{t-p}(\delta), \widehat{X}_t(\delta), \widehat{X}_{t-1}(\delta), \ldots, \widehat{X}_{t-p}(\delta)) + \widehat{B}_0^{21} \widehat{\epsilon}_{1t+j} + \widehat{\xi}_{2t+j} ,
	\end{split}
\end{equation*}
for $j = 1, \ldots, h$ where $\widehat{X}_t(\delta) = X_t + \delta$, $\widehat{X}_{t-s} = X_{t-s}$ for all $s \geq 1$, and similarly for $\widehat{Y}_t(\delta)$.

\section{Estimation}\label{section:estimation}

To discuss estimation, we will rewrite the equations in \eqref{eq_main:pseudo_reduced_form_structural_model} with some minor reordering as
%
%
%
\begin{equation}\label{eq_main:regression_model}
	\begin{split}
		X_t & = \Pi_1' W_{1t} + \epsilon_{1t} , \\
		Y_t & = \Pi_2' W_{2t} + \xi_{2t} ,
	\end{split}
\end{equation}
where 
$\xi_{2t} = B_0^{22} \epsilon_{2t}$,
$\Pi_1 := ( \eta_1, A_{1,11}, \cdots, A_{p,11}, A_{1,12}', \cdots, A_{p,12}' )' \in \Real^{1 + p d}$,
\begin{equation*}
	\Pi_{2} 
	:= 
	\begin{bmatrix}[c|c|c]
		& G_{1,2}(\cdot) & \\
		\mu_2 & \cdots & B_0^{21}\\
		& G_{d_Y,2}(\cdot) &
	\end{bmatrix}' : \Real^{2 + p d} \to \Real^{d_Y} ,
\end{equation*}
$W_{1t} := ( 1, X_{t-1}, \ldots, X_{t-p}, Y_{t-1}', \ldots, Y_{t-p}' )' \in \Real^{1 + p d}$, and
$W_{2t} := (1, X_t, \allowbreak X_{t-1}, \ldots, X_{t-p}, \allowbreak Y_{t-1}', \ldots, \allowbreak Y_{t-p}', \allowbreak \epsilon_{1t} )' \in \Real^{2 + p d}$.
With a slight abuse of notation, similar that of Example~\ref{example:nonlin_impact}, we have written the functional terms in $\Pi_2$ as a ``vector product'', $G_2 \cdot (X_{t:t-p}', Y_{t-1:t-p}')' \equiv G_2(X_{t:t-p}, Y_{t-1:t-p})$, where $G_2$ is a vector of functions,
%
%
one for each component of $Y_t$.

Whenever $\Pi_1 \not= 0$, $W_{2t}$ is an infeasible vector of regressors due to term ${\epsilon}_{1t}$. To estimate $\Pi_2$, one can use $\widehat{W}_{2t} = (1, X_t, \allowbreak X_{t-1}, \ldots, X_{t-p}, \allowbreak Y_{t-1}', \ldots, Y_{t-p}', \allowbreak \widehat{\epsilon}_{1t})'$ instead, which contains generated regressors in the form of residual $\widehat{\epsilon}_{1t}$.
A valid two-step estimation procedure \citep{goncalvesImpulseResponseAnalysis2021} is: 
(I) Regress $X_t$ on $W_{1t}$ to get estimate $\widehat{\Pi}_1$ and residuals $\widehat{\epsilon}_{1t} = X_t - \widehat{\Pi}_1' W_{1t}$;
(II) Semiparametrically regress $Y_t$ on $\widehat{W}_{2t}$ to get estimate $\widehat{\Pi}_2$. 

There are many ways to implement step (II), given that the literature on non- and semiparametric regression is mature. 
We rely on the sieve framework of \cite{chenOptimalUniformConvergence2015} as the workhorse to derive the main theoretical results. The sieve framework is quite rich, encompassing e.g. neural networks \citep{chenImprovedRatesAsymptotic1999,shenAsymptoticPropertiesNeural2023}.

\subsection{Semiparametric Series Estimation}

The semiparametric regression step we require is more readily analyzed by working on each component of $Y_t$. For $i \in \{1, \ldots, d_Y\}$, consider 
\begin{equation}\label{eq_main:regression_eq_i}
	Y_{t,i} = \mu_{2,i} + G_{2,i}(Y_{t-1}, \ldots, Y_{t-p}, X_t, X_{t-1}, \ldots, X_{t-p}) + B^{21}_{0,i} \epsilon_{1t} + \xi_{2t,i} .
\end{equation}
Let then $\pi_{2,i} := [ \mu_{2,i}, \: G_{2,i}, \: B^{21}_{0,i} ]'$. The regression equation for $\pi_{2,i}$ is thus $Y_{i} = \pi_{2,i}' W_{2} + \xi_{2i}$, where $Y_{i} = (Y_{1,i}, \ldots, Y_{n,i})'$ and $\xi_{2i} = (\xi_{2t,1}, \ldots, \xi_{2t,n})'$. The estimation target is the conditional expectation $\pi_{2,i}(w) = \E[ Y_{t,i} \:\vert\: W_{2t} = w ]$ under the assumption $\E[ \xi_{2t,i} \:\vert\: W_{2t} ] = 0$.

Assume that $G_{2,i} \in \Lambda$, where $\Lambda$ is a sufficiently regular function class to be specified in the following.
Given a collection $b_{1\kappa}, \ldots, b_{\kappa\kappa}$ of $\kappa \geq 1$ basis functions belonging to sieve $\mathcal{B}_\kappa$, define $b^\kappa(\cdot) := \left( b_{1\kappa}(\cdot), \ldots, b_{\kappa\kappa}(\cdot) \right)'$ and
$
B_\kappa := \big( 
b^\kappa(Y_{0:1-p}, X_{1:1-p}), 
\ldots, 
b^\kappa(Y_{n-1:n-p}, X_{n:n-p}) 
\big)' 
$. 
For univariate functions, one can directly apply spline, wavelet and Fourier sieves; in the multivariate case, tensor-product sieves are straightforward generalizations \citep{chenOptimalUniformConvergence2015}.
To construct the final semiparametric sieve for $\pi_{2,i}$, indicated by $\mathcal{B}_\pi$, let $b_{\pi,1K}, \ldots, b_{\pi,KK}$ be the sieve basis in $\Real \times \mathcal{B}_\kappa \times \Real$ for $\kappa \geq 1$ and $K = 2 + \kappa$ given by
$b_{\pi,1 K}(W_{2t}) = 1$, 
$b_{\pi,\ell K}(W_{2t}) = b_{\ell \kappa}(Y_{t-1:t-p}, X_{t:t-p})$, for $2 \leq \ell \leq \kappa+1$, and 
$b_{\pi,K K}(W_{2t}) = \epsilon_{1t}$.
%
%
Note that $K$, the overall size of the sieve, grows linearly in $\kappa$, which itself controls the effective dimension of the nonparametric component of the sieve, $b_{\pi,2 K}, \ldots, b_{\pi,(\kappa+1) K}$. 
%
Introducing $b^K_\pi(w) := ( b_{\pi,1K}(w), \ldots, b_{\pi,KK}(w) )'$ and $B_\pi := ( b^K_\pi(W_{21}), \ldots, b^K_\pi(W_{2n}) )'$, the generally \textit{infeasible} least squares series estimator $\widehat{\pi}_{2,i}^*(w)$ is given by $\widehat{\pi}^*_{2,i}(w) = b^K_\pi(w)' ({B}_\pi' {B}_\pi)^{-1} {B}_K' Y_i$.
Similarly, the feasible series regression matrix $\widehat{B}_\pi := ( b^K_\pi(\widehat{W}_{21}), \ldots, b^K_\pi(\widehat{W}_{2n}) )'$ yields the \textit{feasible} least squares series estimator, $\widehat{\pi}_{2,i}(w) = b^K_\pi(w)' (\widehat{B}_\pi' \widehat{B}_\pi)^{-1} \widehat{B}_K' Y_i$.

To further streamline notation, wherever it does not lead to confusion, we will let $\pi_2$ be a generic coefficient vector belonging to $\{\pi_{2,i}\}_{i=1}^p$, as well as define $\widehat{\pi}_{2}$, $Y$ and $u_2$ associated to the same regression equation.

\subsection{Distributional and Sieve Assumptions}

To derive asymptotic consistency results, we begin by stating conditions on the basic probability structure of the model.

\begin{assumption}\label{assumption:stationarity}
	$\{Z_{t}\}_{t \in \Int}$ is a strictly stationary and ergodic time series.
\end{assumption}
\begin{assumption}\label{assumption:compactness}
	$X_{t} \in \mathcal{X} \subset \Real$, $Y_{t} \in \mathcal{Y} \subset \Real^{d_Y}$ and $\epsilon_t \in \mathcal{E} \subset \Real^{d}$ for all $t \in \Int$, where $\mathcal{X}$, $\mathcal{Y}$ and $\mathcal{E}$ are compact, convex sets with nonempty interior.
\end{assumption}

Assumption~\ref{assumption:stationarity} follows both \cite{goncalvesImpulseResponseAnalysis2021} and \cite{chenOptimalUniformConvergence2015}. 
Note that, as $W_{2t}$ depends only on $X_{t:t-p}$, $Y_{t-1:t-p}$ and $\epsilon_{1t}$, the entries of $\xi_{2t}$ in \eqref{eq_main:regression_model} are independent of $W_{2t}$, so that $\E[ u_{2it} \:\vert\: W_{2t} ] = 0$.

Assumption~\ref{assumption:compactness} implies that $X_t$, $Y_t$, as well as $\epsilon_t$ are bounded random variables. In (semi-)nonparametric estimation, imposing that $X_t$ is bounded almost surely is a standard assumption. Since lags of $Y_t$ and innovations $\epsilon_t$ contribute linearly to all components of $Z_t$, it follows that they too must be bounded.
In practice, Assumption~\ref{assumption:compactness} is not particularly restrictive, as many credibly stationary economic series often have reasonable implicit (e.g., inflation) or explicit bounds (e.g., employment rate).
The analysis of impulse responses on compact domains is, however, non-trivial. In Section~\ref{section:relaxed_shocks} below, we provide an IRF shock relaxation framework that can accommodate this setting.  

\begin{remark}\label{remark:compactness_assumption}
	Bounded support assumptions are uncommon in time series econometrics, as boundedness is not necessary in the analysis of linear models \citep{hamilton1994state,lutkepohlNewIntroductionMultiple2005,kilianStructuralVectorAutoregressive2017,stock2016dynamic}. 
	Unbounded regressors are significantly more complex to handle when working in the nonparametric setting.
	\cite{chenOptimalUniformConvergence2015} do work in weighted sup-norms, but their uniform results are stated only under a compact domain assumption.
	Avoiding Assumption~\ref{assumption:compactness} can be achieved by changing the model's equations -- e.g., the lags of $Y_t$ only affect $X_t$ via bounded functions -- but this further restricts the model. 
	Establishing a general (uniform) theory of nonparametric regressions with unbounded data domains, on the other hand, is a complex question. For kernel, partitioning and nearest-neighbor methods and i.i.d. data, a handful of papers develop results in $L^1$ and $L^2$ norms, see \cite{kohlerRatesConvergencePartitioning2006,kohlerOptimalGlobalRates2009} and \cite{kohlerOptimalGlobalRates2013}. For wavelet estimators in the i.i.d. regression setting, \cite{zhouUniformConvergenceRates2022} provided the first sup-norm result in Besov spaces with suboptimal rates.
    \cite{hansenUNIFORMCONVERGENCERATES2008} is, to the best of our knowledge, the only work providing convergence rates for local constant and local linear regression estimators in a dependent data setting without bounded support restrictions. Yet, even in the nonparametric LP setup of \cite{goncalvesNonparametricLocalProjections2024}, the authors argue that it is not clear if these results allow for IRF estimation guarantees over $\Real$.
	Construction of a comprehensive nonparametric framework to handle non-independent, unbounded data should thus be considered an important objective of future research.
\end{remark}

%
%


Without loss of generality, let $\mathcal{Y} = [0,1]^{d_Y}$ and $\mathcal{X} = [0,1]$. 

\begin{assumption}\label{assumption:regressor_density}
	The unconditional densities of $Y_t$ and $X_{t}$ are uniformly bounded away from zero and infinity over $\mathcal{Y}$ and $\mathcal{X}$, respectively.
\end{assumption}

\begin{assumption}\label{assumption:function_class}
	For all $1 \leq i \leq d_Y$ the restriction of $G_{2,i}$ to $\mathcal{Y}^{\, p} \times \mathcal{X}^{1+p} \equiv [0,1]^{1+pd}$ belongs to the Hölder class $\Lambda^s([0,1]^{1+pd})$ of smoothness $s \geq 1$.
\end{assumption}

Assumptions \ref{assumption:regressor_density} and \ref{assumption:function_class} are classical in the nonparametric regression literature. 
%
Let then $\mathcal{W}_2 \subset \Real^d$ be the domain of $W_{2t}$. By assumption, $\mathcal{W}_2$ is compact and convex and is given by the direct product
$
\mathcal{W}_2 = \{1\} \times \mathcal{Y}^{\, p} \times \mathcal{X}^{1+p} \times \mathcal{E}_1 
$, 
where $\mathcal{E}_1$ is the domain of structural innovations $\epsilon_{1t}$ i.e. $\mathcal{E} \equiv \mathcal{E}_1 \times \mathcal{E}_2$. 

\begin{assumption}\label{assumption:sieve_regularity}
	Define $\zeta_{K,n} := \sup_{w \in \mathcal{W}_2} \lVert b^K_\pi(w) \rVert$ and $\lambda_{K,n} := [ \lambda_{\min}(\E[\, b^K_\pi(W_{2t}) b^K_\pi(W_{2t})' \,]) ]^{-1/2}$.
	It holds: (i) there exist $\omega_1, \omega_2 \geq 0$ s.t. $ \sup_{w \in \mathcal{W}_2} \lVert \nabla b^K_\pi(w) \rVert \lesssim n^{\omega_1} K^{\omega_2}$; 
	(ii) there exist $\overline{\omega}_1  \geq 0$, $ \overline{\omega}_2 > 0$ s.t. $ \zeta_{K,n} \lesssim n^{\overline{\omega}_1} K^{\overline{\omega}_2} $; 
	(iii) $\lambda_{\min}(\E[\, b^K(W_{2t}) b^K(W_{2t})' \,]) > 0$ for all $K$ and $n$.
\end{assumption}


Assumption \ref{assumption:sieve_regularity} provides mild regularity conditions on the families of sieves that can be used for the series estimator. More generally, letting $\mathcal{W}_2$ be compact and rectangular makes Assumptions~\ref{assumption:sieve_regularity}(i)-(ii) hold for commonly used basis functions \citep{chenOptimalUniformConvergence2015}. 
The approximation properties of these sieves are well understood \citep{chenChapter76Large2007}.\footnote{See also \cite{chenPenalizedSieve2013,belloniNewAsymptoticTheory2015} for additional discussion and examples of sieve families.}
In particular, Assumption \ref{assumption:sieve_regularity}(i) holds with $\omega_1 = 0$ since the domain is fixed over the sample size. 
What is also needed is that the nonparametric components of the sieve given by $b_{\pi,1K}, \ldots, b_{\pi,KK}$ are able to approximate $G_{2,i}$ with an error that decays sufficiently fast with $K$. 
Lastly, Assumption~\ref{assumption:sieve_regularity}(iii) is a mild assumption on the conditioning of the semiparametric sieve.

\begin{assumption}\label{assumption:sieve_type}
	Sieve $\mathcal{B}_\kappa$ belongs to $\text{BSpl}(\kappa, \mathcal{W}_2, r)$ or $\text{Wav}(\kappa, \mathcal{W}_2, r)$, the tensor B-spline and tensor wavelet sieve, respectively, of degree $r$ over $\mathcal{W}_2$, 
	with $r \geq \max\{ s, 1 \}$.
\end{assumption}


We define
$\widetilde{b}^K_\pi(w) := \E[\, {b}^K_\pi(W_{2t}) {b}^K_\pi(W_{2t})' \,]^{-1/2}\, {b}^K_\pi(w)$ and
$\widetilde{B}_\pi := \big( \widetilde{b}^K_\pi(W_{21}), \allowbreak \ldots, \allowbreak \widetilde{b}^K_\pi(W_{2n}) \big)'$
%
%
to be the orthonormalized vector of basis functions and the orthonormalized regression matrix, respectively.
To derive uniform converges rates under dependence, we require that the Gram matrix of orthonormalized sieve converges to the identity matrix. 

\begin{assumption}\label{assumption:series_gram_matrix_convergence}
	It holds that $\lVert (\widetilde{B}_\pi' \widetilde{B}_\pi / n) - I_K \rVert = o_P(1)$.
\end{assumption}

\cite{chenOptimalUniformConvergence2015} introduced Assumption~\ref{assumption:series_gram_matrix_convergence} as a key ingredient for their proofs, while also showing that it holds whenever $\{W_{2t}\}_{t\in\Int}$ is either an exponential or algebraic $\beta$-mixing process. 
Unfortunately, mixing conditions are difficult to verify or test with respect to model specification, as they rely on bounding the worst-case ``independence gap'' between probability events (see Appendix~\ref{appendix:dependence}). 
We extend their approach to the case of geometrically decaying physical dependence, a metric proposed by \cite{wuNonlinearSystemTheory2005}. This is a setting where many estimation and inference results have been derived, see for example \cite{wuKernelEstimationTime2010,wuAsymptoticTheoryStationary2011a,chenSelfnormalizedCramertypeModerate2016} and references within.

\begin{assumptionp}{\ref*{assumption:series_gram_matrix_convergence}$\,'$}\label{assumption:physical_dep}
	Let $\{Z_t\}_{t \in \Int}$ be such that we can write $Z_{t+h} = \Phi^{(h)}(Z_t,\allowbreak \epsilon_{t+1:t+h})$ for some nonlinear maps $\Phi^{(h)}$ and innovations $\{\epsilon_t\}_{t \in \Int}$ over all $h \geq 1$. Then, for $r \geq 2$, there exists constants $a_1 > 0$, $a_2 > 0$ and $\tau \in (0,1]$ such that it holds
	\begin{equation*}
		\sup_t \big\lVert\, Z_{t+h} - \Phi^{(h)}(Z_t, \epsilon_{t+1:t+h}) \,\big\rVert_{L^r}
		\leq
		a_1 \exp(- a_2 \, h^\tau) .
	\end{equation*}
\end{assumptionp}

Assumption~\ref{assumption:series_gram_matrix_convergence} is subsumed by Assumption~\ref{assumption:physical_dep}. Using a physical dependence measure, we argue that it is also possible to swap mixing conditions with more explicit, primitive conditions derived exclusively in terms of model specification \eqref{eq_main:general_nonlin_model}. In particular, for specific semiparametric model specifications, it is possible to verify Assumption~\ref{assumption:physical_dep} directly by leveraging stability/contractivity theory of dynamic systems.
We refer the reader to Appendix~\ref{appendix:dependence} for an additional, in-depth discussion of dependence and physical conditions.

\subsection{Uniform Convergence and Consistency}

We can now state our main result, which shows that the two-step estimation procedure for \eqref{eq_main:regression_model} provides consistent estimates. 

\begin{theorem}\label{theorem:twostep_estimator_consistency}
	Let $\{Z_t\}_{t \in \Int}$ be determined by structural model \eqref{eq_main:structural_model}. Under Assumptions \ref{assumption:structural_model}, \ref{assumption:stationarity}, \ref{assumption:compactness},  \ref{assumption:regressor_density}, \ref{assumption:function_class}, \ref{assumption:sieve_regularity}, \ref{assumption:sieve_type} and \ref{assumption:physical_dep}, let $\widehat{\Pi}_1$ and $\widehat{\Pi}_2$ be the least squares and two-step semiparametric series estimators for $\Pi_1$ and $\Pi_2$, respectively. Then,
	$ \lVert \widehat{\Pi}_1 - \Pi_1 \rVert_\infty = O_P(n^{-1/2\,}) $
	and
	\begin{equation*}
		\lVert \widehat{\Pi}_2 - \Pi_2 \rVert_\infty 
		\leq
		O_P\left( \zeta_{K,n} \lambda_{K,n} \, \frac{K}{\sqrt{n}} \right) + \lVert \widehat{\Pi}^*_2 - \Pi_2 \rVert_\infty ,
	\end{equation*}
	where $\widehat{\Pi}^*_2$ is the infeasible series estimator involving $\epsilon_{1t}$.
    %
\end{theorem}


Sup-norm bounds for $\lVert \widehat{\Pi}^*_2 - \Pi_2 \rVert_\infty$ may be obtained from Lemma~2.3 and Lemma~2.4 in \cite{chenOptimalUniformConvergence2015}. 
Assuming $s \geq 1$ and $d = 1$, such as in the setting of the additively separable model in Example~\ref{example:nonlin_impact} and in our empirical applications, it is possible to show that, if the optimal nonparametric rate for $K$ is used and the additive sieve inherits the conditioning of the underlying sieve bases, then $\widehat{\Pi}_2$ is sup-norm consistent. 

\begin{corollary}\label{corollary:twostep_estimator_op1}
    Under the same assumptions as Theorem~\ref{theorem:twostep_estimator_consistency}, further assume that $s \geq 1$, $G_{2,i}$, $1 \leq i \leq d_Y$, in \eqref{eq_main:regression_eq_i} is additively separable in all its components and $\lambda_{K,n} \lesssim 1$. Then for the choice $K \asymp (n / \log(n))^{1/(2 s + 1)}$ it holds that
    \begin{equation*}
        \lVert \widehat{\Pi}_2 - \Pi_2 \rVert_\infty
        =
        O_P \left(
            n^{- \frac{s-1}{2s+1}}
            \log(n)^{-\frac{3}{2(2s+1)}}
        \right)
    \end{equation*}
    and, in particular, $\lVert \widehat{\Pi}_2 - \Pi_2 \rVert_\infty = o_P(1)$.
\end{corollary}

The requirement $\lambda_{K,n} \lesssim 1$ for additively separable sieves is mild given the known properties of B-spline and wavelet sieves, although nontrivial. 
Since one cannot exploit sparsity in the case of non-locally supported bases, as is the case with linearly separable sieves, we assume $\lambda_{K,n}$ is upper bounded by a constant to streamline the analysis of the empirical sieve projection operator (see also the discussion in \citealp{huangLocalAsymptoticsPolynomial2003a}, Section 7).

%

\begin{remark}
    Several methods can be used to select $K$ in practice: Cross-validation, generalized cross-validation, Mallow's criterion, and others \citep{li2009nonparametric}. In the case of piecewise splines, once size is selected, knots can be chosen to be the $K$ uniform quantiles of the data. In our simulations and applications, for simplicity, we select sieve sizes manually, while knots are located following empirical quantiles.
\end{remark}

\section{Impulse Response Analysis}\label{section:impulse_response_analysis}

After discussing the estimation of the structural model's coefficients, we can now address the derivation of nonlinear impulse responses.
To ensure compatibility with bounded support assumptions, we introduce an extension of the classical IRF definition, termed \textit{relaxed impulse response function}, which differs only in the form of the shock applied to the model. 
We then show that nonlinear relaxed IRFs can be consistently estimated, and uniformly so for shocks picked within a compact range.

\subsection{Relaxed Shocks}\label{section:relaxed_shocks}

Under Assumptions~\ref{assumption:compactness} and \ref{assumption:regressor_density}, the standard construction of impulse responses following Section~\ref{section:nonlin_irfs} is, unfortunately, improper. This is immediately seen by noticing that, at impact, $X_t(\delta) = X_{t} + \delta$, meaning that $\P( X_t(\delta) \not\in \mathcal{X} ) > 0$ since there is a translation of size $\delta$ in the support of $X_t$. 
To address this problem, we introduce an extension to the standard additive shock that is used to define impulse responses.

We begin by defining mean-shift shocks, that is, shocks such that the distribution of time $t$ innovations is shifted to have mean $\delta$, while retaining compact support almost surely.

\begin{definition}
	Let $\mathcal{E}_1 \subseteq \Real$ and $\P(\epsilon_{1t} \in \mathcal{E}_1) = 1$. A {mean-shift structural shock} $\epsilon_{1t}(\delta)$ is an appropriately chosen transformation of $\epsilon_{1t}$ such that $\P(\epsilon_{1t}(\delta) \in \mathcal{E}_1) = 1$ and $\E[\epsilon_{1t}(\delta)] = \delta$.
\end{definition}

With a mean-shift shock, at impact it holds $X_{t}(\delta) = X_{t} + (\epsilon_{1t}(\delta) - \epsilon_{1t})$. In the standard setting, where $\E[\epsilon_t] = 0$ and $\mathcal{E}_1 \equiv \Real$, $\epsilon_{1t}(\delta) = \epsilon_{1t} + \delta$ is clearly valid. More generally, however, imposing $\E[\epsilon_{1t}(\delta)] = \delta$ requires the distribution of $\epsilon_{1t}$ to be known.
If instead one is willing to assume only that $\E[\epsilon_{1t}(\delta)] \approx \delta$, it is possible to sidestep this need by introducing a \textit{shock relaxation function}.

\begin{definition}
	Assume $\mathcal{E}_1 = [a, b]$. A shock relaxation function is a map $\rho : \mathcal{E}_1 \to [0, 1]$ such that $\rho(e) = 0$ for all $e \in \Real \,\setminus\, \mathcal{E}_1$, $\rho(e) \geq 0$ for all $e \in \mathcal{E}_1$ and there exists $e_0 \in \mathcal{E}_1$ for which $\rho(e_0) = 1$. Moreover, for a given shock $\delta \in \Real$,  
	\begin{itemize}
		\item[(i)] If $\delta > 0$, $\rho$ is said to be right-compatible with $\delta$ if $e + \rho(e)\delta \leq b$ for all $e \in \mathcal{E}_1$.
		\item[(ii)] If $\delta < 0$, $\rho$ is said to be left-compatible with $\delta$ if $e + \rho(e)\delta \geq a$ for all $e \in \mathcal{E}_1$.
		\item [(iii)] $\rho$ is compatible with shock magnitude $|\delta| > 0$ if it is both right- and left-compatible. 
	\end{itemize}
\end{definition}

By setting $\epsilon_{1t}(\delta) = \epsilon_{1t} + \delta \rho(\epsilon_{1t})$ for a $\rho$ compatible with $\delta$, it follows that $ X_{t}(\delta) = X_{t} + \delta \rho(\epsilon_{1t})$ and $\lvert \E[\epsilon_{1t}(\delta)] \rvert = \lvert \delta \E[\rho(\epsilon_{1t})] \rvert \leq \lvert \delta \rvert$ since $\E[\rho(\epsilon_{1t})] \in [0, 1]$ by definition of $\rho$. If $\rho$ is a bump function, a relaxed shock is a structural shock that has been mitigated proportionally to the density of innovations at the edges of $\mathcal{E}_1$ and the squareness of $\rho$. 

\begin{remark}
    When studying impulse responses, a researcher is primarily interested in shock $\delta$ itself, not in $\rho$, and the latter plays the role of a tuning parameter.
    From a practical perspective, given a choice of $\delta$ (or a range $\mathcal{D}$) of interest, the researcher should explicitly select $\rho$ to \textit{minimize distortions} implied by using $\epsilon_{1t} + \delta \rho(\epsilon_{1t})$ instead of a pure shift $\epsilon_{1t} + \delta$.
    If $\delta$ is sufficiently small and $\epsilon_{1t}$ is sufficiently concentrated, negligible distortions can be achieved. 
    Importantly, one can empirically check the impact of $\rho$ on the nonparametric IRFs by comparing relaxed and non-relaxed response estimates. 
    For example, in Appendix~\ref{appendix:robustness}, we provide robustness checks showing that, for both applied examples in Section~\ref{section:applications}, our chosen relaxation functions introduce negligible distortions.
\end{remark}

It is important to emphasize that shock relaxation is a generalization of standard shock designs. Indeed, when $\mathcal{X} = \Real$ and $\mathcal{E}_1 = \Real$, $\rho = 1$ is a relaxation function compatible with all $\delta \in \Real$.
Nonetheless, we may also wonder how much information on the nonlinear term $G_2$ we can recover at the ``boundary'' of a finite sample.
If $X_t$ is unbounded but well-concentrated, even under strong smoothness conditions and strictly positive density, little can be learned about the \textit{local} structure of regression functions in regions of low density.\footnote{In our regression setting, for example, Theorem 1 in \cite{kohlerOptimalGlobalRates2009} on $L_2$ kernel regression error, assuming $\E[\lvert X_t \rvert^{\beta}] \leq M < \infty$ for some constant $\beta > 2s$, would require the bandwidth to grow over $\mathcal{X}$ faster than $\lvert X_t \rvert$. This question is also linked to issues in kernel density estimation over sets with boundary, see e.g. \cite{karunamuni2005on,malec2014nonparametric,berry2017density} and references therein.}  

\begin{remark}\label{remark:shock_relax_choice}
	%
	In this paper, and more specifically in Sections~\ref{section:simulations} and \ref{section:applications}, we choose $\rho$ to be a symmetric exponential bump function, $\rho \in \{ x \mapsto \1{x \leq c} \exp(1 + (|x/c|^\alpha - 1)^{-1} ) \:|\: \alpha > 0 \}$ for some constant $c > 0$. 
	This $\mathcal{C}^\infty$ bump class is widely studied in both functional \citep{mitrovic1997fundamentals} and Fourier analysis \citep{stein2011fourier}.\footnote{For generic shock distributions, one can for also consider the class $\{ x \mapsto \1{a \leq x \leq b} \exp(1 + (|2(x-b)/(b-a) + 1|^\alpha - 1)^{-1} ) \:|\: \alpha > 0 \}$ of exponential bump functions with domain $[a,b] \subset \Real$.}
	We aim to set $\alpha$ to be as large as possible to minimize distortions from a linear shift, while retaining compatibility with $\delta \in \mathcal{D}$, where $\mathcal{D}$ is a set of shocks of empirical interest. 
\end{remark}



\subsection{Relaxed Impulse Response Consistency}

We will now study relaxed impulse responses in the setting of additively separable models. Additive separability is a common assumption in applied work, as we shall impose it in the empirical analyses of Section~\ref{section:simulations} and \ref{section:applications}. Further, collecting nonlinear terms over lags significantly streamlines notation and analysis, and aligns with the setup of Corollary~\ref{corollary:twostep_estimator_op1}. It would be straightforward, if tedious, to extend our derivations below to the more general setting of Theorem~\ref{theorem:twostep_estimator_consistency}.

Given $\delta \in \Real$ and compatible shock relaxation function $\rho$, let $\widetilde{\delta}_t := \delta \rho({\epsilon}_{1t})$.
Starting from a path $X_{t+j:t}$ and \eqref{eq_main:irf_h_separable}, the relaxed shock path is
\begin{equation*}
	X_{t+j}(\widetilde{\delta}_t) 
	= X_{t+j} + \Theta_{j,11} \widetilde{\delta}_t + \sum_{k=1}^j  \left[ \Gamma_{k,11} X_{t+j-k}(\widetilde{\delta}_t) - \Gamma_{k,11} X_{t+j-k} \right]
	=: \gamma_{j}(X_{t+j:t}; {\widetilde{\delta}}_t) .
\end{equation*}
The relaxed-shock impulse response is thus given by
\begin{align*}
	{\widetilde{\textnormal{IRF}}}_{h}(\delta) 
	& 
	:= \E[Z_{t+j}(\widetilde{\delta}_t) - Z_{t+j}] 
	= \Theta_{h,\cdot 1} \delta \, \E\left[ \rho(\epsilon_{1t}) \right] +  \sum_{k=1}^j \E\left[  \Gamma_{k} X_{t+j-k}(\widetilde{\delta}_t) - \Gamma_{k} X_{t+j-k} \right] .
\end{align*}
For $1 \leq \ell \leq d$, we define ${V}_{j,\ell}(\delta)$ to be the sample analog of the horizon $j$ nonlinear effect on the $\ell$th variable,
\begin{equation*}
	{V}_{j,\ell}(\delta) 
	:= 
	\frac{1}{n-j} \sum_{t=1}^{n-j} \left[ {\Gamma}_{j,\ell} {\gamma}_{j}(X_{t+j:t}; {\widetilde{\delta}}_t) - {\Gamma}_{j,\ell} X_{t+j} \right] 
	= 
	\frac{1}{n-j} \sum_{t=1}^{n-j}  v_{j,\ell}(X_{t+j:t}; {\widetilde{\delta}}_t) ,
\end{equation*}
where ${\Gamma}_{j,\ell}$ is the $\ell$th component of functional vector ${\Gamma}_{j}$. 
As $\epsilon_{1t}$ is not universally observable, we introduce its residual counterpart, $\widehat{\widetilde{\delta}}_t = \delta \rho(\widehat{\epsilon}_{1t})$.
The associated plug-in sample estimates are
$\widehat{V}_{j,\ell}(\delta) 
= 
({n-j})^{-1} \sum_{t=1}^{n-j} \widehat{v}_{j,\ell}\big( X_{t+j:t}; \widehat{\widetilde{\delta}}_t \big)$,
$\widehat{v}_{j,\ell}(X_{t+j:t}; \widehat{\widetilde{\delta}}_t)
=
\widehat{\Gamma}_{j,\ell} \widehat{\gamma}_{j}(X_{t+j:t}; \widehat{\widetilde{\delta}}_t) - \widehat{\Gamma}_{j,\ell} X_{t+j}$,
%
%
%
and 
\begin{equation*}
	\widehat{\widetilde{\textnormal{IRF}}}_{h,\ell}(\delta)
	= 
	\widehat{\Theta}_{h,\cdot 1} \delta \left( \frac{1}{n} \sum_{t=1}^{n} \rho(\widehat{\epsilon}_{1t}) \right) 
	+
	\sum_{j=0}^h \widehat{V}_{j,\ell}(\delta) .
\end{equation*}

Our next theorem proves the consistency of the relaxed impulse responses estimator based on semiparametric series estimates. We leverage the sup-norm bounds of Theorem~\ref{theorem:twostep_estimator_consistency} to derive a result that is uniform in $\delta$ over a compact interval $[-\mathcal{D}, \mathcal{D}]$, $\mathcal{D} > 0$. This allows us to make valid comparisons between IRFs due to shocks of different sizes.

\begin{theorem}\label{theorem:consistent_irf}
	Let $\widehat{\widetilde{\textnormal{IRF}}}_{h,\ell}(\delta)$ be the semiparametric estimate for the horizon $h$ relaxed shock IRF of variable $\ell$ based on relaxation function $\rho$ with compatibility range $[-\mathcal{D}, \mathcal{D}]$. Under the assumptions in Theorem \ref{theorem:twostep_estimator_consistency} and Assumption~\ref{assumption:roots_linear_part_model},
	\begin{equation*}
		\sup_{\delta \in [-\mathcal{D}, \mathcal{D}]} \left\lvert \widehat{\widetilde{\textnormal{IRF}}}_{h,\ell}(\delta) 
		-
		\widetilde{\textnormal{IRF}}_{h,\ell}(\delta)
		\right\rvert
		=
		o_P(1)
	\end{equation*}
	for any fixed integers $0 \leq h < \infty$ and $1 \leq \ell \leq d$.
\end{theorem}

\begin{remark}
	By construction of $\widehat{\widetilde{\textnormal{IRF}}}_{h,\ell}(\delta)$, Proposition~\ref{prop:irf_iterate_algorithm} remains valid when computing ${\widetilde{\textnormal{IRF}}}_{h}(\delta)$ instead of $\textnormal{IRF}_{h}(\delta)$. The only adjustment to be made is that in step (i) one must set $X_{t}(\delta) = X_{t} + \delta \rho(\epsilon_{1t})$ and iterate forward accordingly. Assumptions \ref{assumption:structural_model}, \ref{assumption:stationarity}, and \ref{assumption:physical_dep} ensure that the IRFs of interest are well-defined.
\end{remark}

\begin{remark}
	Our definition of a compatible relaxation function is \textit{static}, as it considers only the impact effect of a shock. Nonetheless, $X_t(\delta) \in \mathcal{X}$ for all $t$ must hold to properly define ${\widetilde{\textnormal{IRF}}}_{h}(\delta)$. In theory, given $\delta$, one can always either expand $\mathcal{X}$ or strengthen $\rho$ so that compatibility is enforced at all horizons $1 \leq h \leq H$. In simulations, the choice of domains and relaxation functions can be done transparently. When working with empirical data, unless $X_t$ is exogenous or strictly autoregressive, more care has to be taken to check that there is no dynamic domain violation. In Section~\ref{subsection:app_istrefi}, where $X_t$ is an endogenous series, we discuss such a robustness check.
\end{remark}

\section{Simulations}\label{section:simulations}

To analyze the performance of the two-step semiparametric estimation strategy discussed above, we begin by considering the two simulation setups employed by \cite{goncalvesImpulseResponseAnalysis2021}. We compare the bias and MSE of the estimated relaxed shocked impulse response functions for different methods. The population responses we consider in this section are also constructed using the same shock relaxation scheme; therefore, both relaxed IRF estimators are correctly specified. Appendix~\ref{appendix:robustness} includes a robustness analysis wherein non-relaxed nonlinear population IRFs are targeted.
We also provide simulations under a misspecified design, which highlight how, in larger samples, the nonparametric sieve estimator consistently recovers impulse responses, whereas a least-squares estimator constructed with a pre-specified nonlinear transform may not.\footnote{Population impulse responses are estimated with $10^5$ replications, while MSE and bias of both semiparametric and parametric IRFs are computed with $10^4$ Monte Carlo replications. In all setups, a cubic B-spline sieve is used.}

\paragraph*{Benchmarks.}

Like in \cite{goncalvesImpulseResponseAnalysis2021}, we consider two simulation setups: A bivariate design with identified shocks (DGPs 1-3) and a three-variable design with partial block-recursive identification (DGPs 4-6).
In both, we set a sample size of $n = 240$, which is realistic for most macroeconomic data settings: this is approximately equivalent to 20 years of monthly data or 60 years of quarterly data \citep{goncalvesImpulseResponseAnalysis2021}. We discuss here only the bivariate simulation design with shock $\delta = +1$, and refer the reader to Appendix~\ref{appendix:sim_details} for the block-recursive setup.

\begin{figure}[t!]
	\centering
	\includegraphics[width=\textwidth]{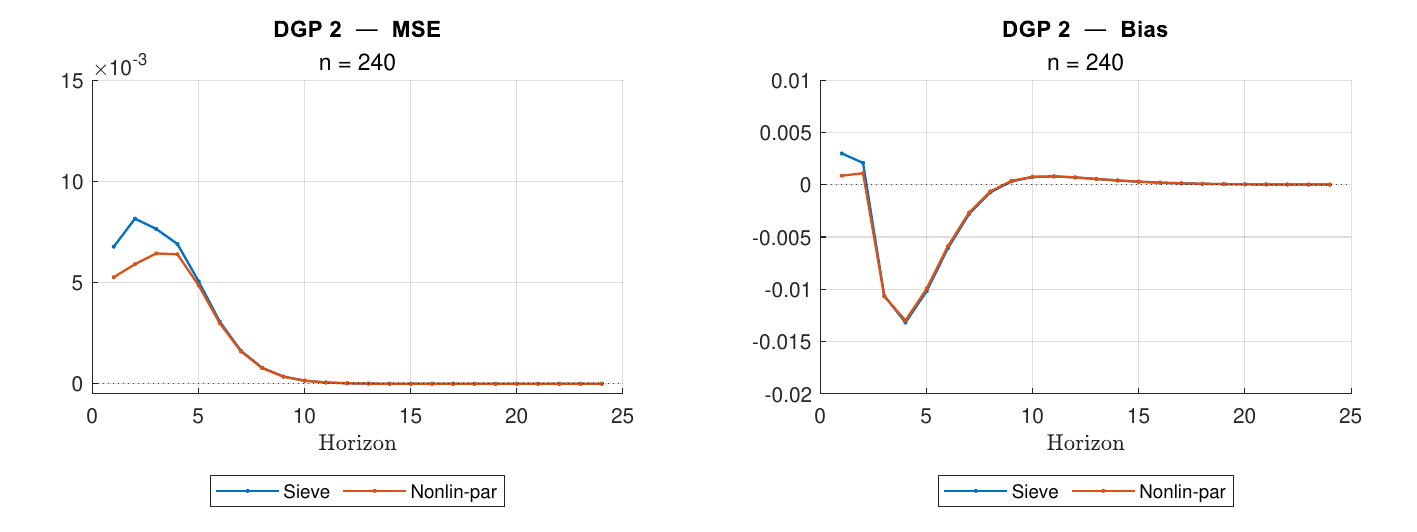}
	\caption{Simulation results for DGP 2 with $\delta = +1$.}
	\label{fig:mse_bias_DGP_2}
\end{figure}

We set either $X_t = \epsilon_{1t}$ (DGP 1), $X_t = 0.5 X_{t-1} + \epsilon_{1t}$ (DGP 2) or $X_t = 0.5 X_{t-1} + 0.2 Y_{t-1} + \epsilon_{1t}$ (DGP 3), and 
\begin{equation*}
	Y_t = 0.5 Y_{t-1} + 0.5 X_{t} + 0.3 X_{t-1} - 0.4 \max(0, X_t) + 0.3 \max(0, X_{t-1}) + \epsilon_{2t} .
\end{equation*}
Innovations $\epsilon_{1t}$ and $\epsilon_{2t}$ are drawn as independent, truncated standard Gaussian variables over $[-3, 3]$. The shock relaxation function is $\rho(z) = \mathbb{I}\{|z| \leq 3 \}\exp\left( 1 + ( \lvert {z}/{3} \rvert^4 - 1 )^{-1} \right)$, cf. Remark~\ref{remark:shock_relax_choice}.
In Figure~\ref{fig:mse_bias_DGP_2} we show MSE and bias curves for the IRF on $Y_t$ in DGP 2, where $X_t$ is an exogenous AR(1) process. One can see that the sieve IRF leads only to a minor increase in mean squared error at short horizons compared to directly estimating the parameter of the true specification. This marginal increase in MSE is consistent across DGPs 1 through 3.

These simulations show that there is negligible loss of efficiency in terms of either MSE or bias when implementing the fully flexible semiparametric estimates at realistic sample sizes. We confirm these results when studying DGPs 4-6, where estimation of the structural matrix $B_0$ is included in the regression problem. Detailed results can be found in Appendix~\ref{appendix:sim_details}.

\paragraph*{Misspecified Model.}

To assess the robustness of the proposed semiparametric approach versus the parametric nonlinear model, we consider a modified process (DGP 7):
\begin{equation}\label{sim_eg:DGP_7}
	\begin{split}
		X_t & = 0.8 X_{t-1} + \epsilon_{1t} , \\
		Y_t & = 0.5 Y_{t-1} + 0.9 \varphi(X_t) + 0.5 \varphi(X_{t-1}) + \epsilon_{2t} ,
	\end{split}
\end{equation}
where $\varphi(x) := (x - 1)(0.5 + \tanh(x - 1)/2)$. In this design, we assume that the researcher's prior is $\varphi(x) = \max(0, x)$, as in the benchmark simulations. To emphasize the difference in estimated IRFs, in this setup we focus on $|\delta| = 2$ and $n = 2400$; innovations $\epsilon_{1t}$ and $\epsilon_{2t}$ are drawn from a standard Gaussian distribution truncated over $[-5, 5]$, and $\rho(z) = \mathbb{I}\{|z| \leq 5 \}\exp( 1 + ( \lvert {z}/{5} \rvert^{3.9} - 1 )^{-1} )$.
As Figure~\ref{fig:mse_bias_DGP_2_plus} shows, positive-shock parametric nonlinear IRF estimates are severely biased, while semiparametric sieve IRFs have comparatively negligible error: This yields an up to 4 times reduction of overall MSE at short horizons. Appendix~\ref{appendix:sim_details} provides additional simulation results showing that the same improvements hold when $\delta = - 2$. There, we also discuss the setting where $\varphi(x)$ is replaced with map $\widetilde{\varphi}(x) = \varphi(x+1)$, which agrees closely with $\max(0, x)$. In this last setting, we find that parametric nonlinear regression dominates in MSE and bias terms. As one might expect, therefore, parametric modeling is reliable only in cases where a sufficiently good model prior is available.

\begin{figure}[t!]
	\centering
	\includegraphics[width=\textwidth]{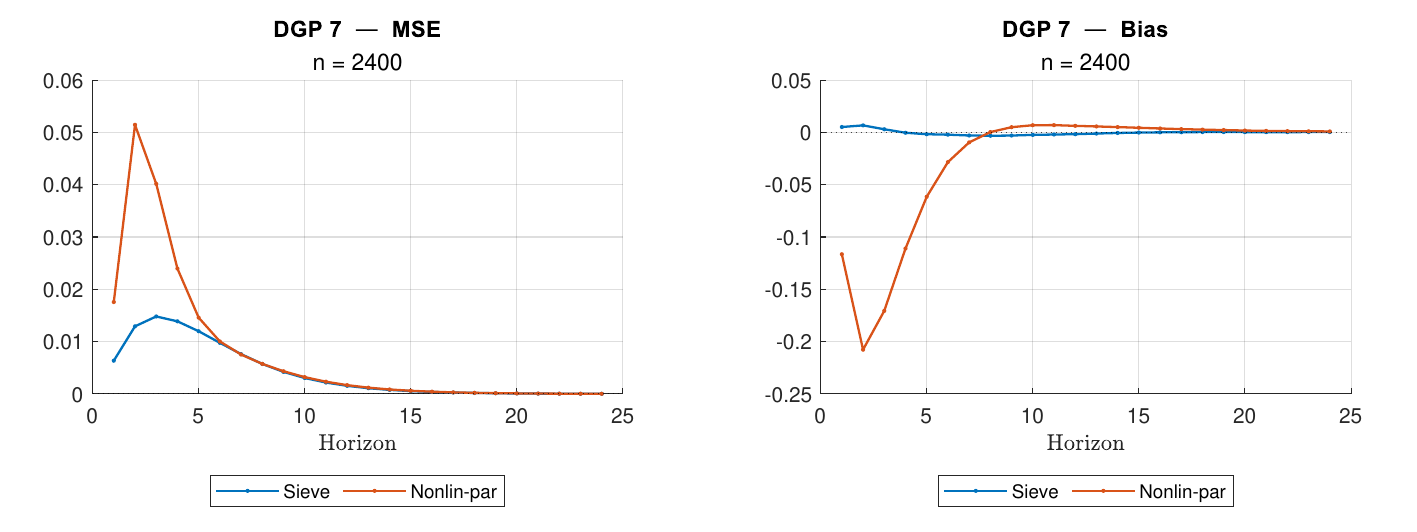}
	\caption{Simulation results for DGP 7 with shock $\delta = +2$.}
	\label{fig:mse_bias_DGP_2_plus}
\end{figure}

\section{Empirical Applications}\label{section:applications}

In this section, we showcase the practical utility of the proposed semiparametric sieve estimator by considering two applied exercises. 
In line with previous work applying nonlinear IRF methods to macroeconomic data, such as e.g. \cite{kilian2011responses,goncalvesImpulseResponseAnalysis2021,goncalvesStatedependentLocalProjections2024} and \cite{goncalvesNonparametricLocalProjections2024}, our discussion is focused on point impulse response estimates.
In both cases, we will consider sieve IRFs constructed with shock relaxation: Appendix~\ref{appendix:robustness} shows that our analysis remains valid also when evaluating non-relaxed semiparametric responses.

\subsection{Monetary Policy Shocks}\label{subsection:app_goncalves}

We first consider a four-variable model identical to the one analyzed by \cite{goncalvesImpulseResponseAnalysis2021}, and based on \cite{tenreyro2016pushing}. Let $Z_t = (X_t, \textnormal{FFR}_t, \textnormal{GDP}_t, \textnormal{PCE}_t)'$, where $X_t$ is the series of narrative U.S. monetary policy shocks, $\textnormal{FFR}_t$ is the federal funds rate, $\textnormal{GDP}_t$ is log-real GDP and $\textnormal{PCE}_t$ is PCE inflation.\footnote{In \cite{goncalvesImpulseResponseAnalysis2021} p.~122, it is mentioned that CPI inflation is included in the model, but both in the replication package made available by one the authors (\url{https://sites.google.com/site/lkilian2019/research/code}) from which we source the data, and in \cite{tenreyro2016pushing}, PCE inflation is used instead. Moreover, the authors say that both the FFR and PCE enter the model in first differences, yet, in their code, these variables are kept in levels. We thus consider a model in levels to allow for a proper comparison between estimation methods, although the series are highly persistent.} 
As a pre-processing step, GDP is transformed to log GDP and then linearly detrended. The data is available quarterly and spans from 1969:Q1 to 2007:Q4.
As in \cite{tenreyro2016pushing}, we use a model with one lag, $p=1$. The narrative shock $X_t$ is considered to be an i.i.d. sequence, i.e. $X_t = \epsilon_{1t}$, therefore we assume no dependence on lagged variables when implementing the pseudo-reduced form \eqref{eq_main:pseudo_reduced_form_structural_model}. Like in \cite{goncalvesImpulseResponseAnalysis2021}, we consider positive and negative shocks of size $|\delta| = 1$ and choose
$\rho(z) = \mathbb{I}\{ |z| \leq 4 \} \exp( 1 + ( \lvert {z}/{4} \rvert^{6} - 1 )^{-1} )$
%
%
to be the shock relaxation function. Figure \ref{fig:plot_app_gonc2021_meanshift} in the Online Appendix provides a check for the compatibility of $\rho$ given the sample distribution of $X_t$. Knots for sieve estimation are located at $\{-1, 0, 1\}$. The model is block-recursive, and U.S. monetary policy shocks are identified without the need to impose additional assumptions on the remaining shocks. 
\cite{goncalvesImpulseResponseAnalysis2021}, like \cite{tenreyro2016pushing}, use two nonlinear transformations, $F(x) = \max(0, x)$ and $F(x) = x^3$, to try to gauge how negative versus positive and large versus small shocks, respectively, affect the U.S. macroeconomy. They find that the two maps yield very similar responses, so we focus on comparing the IRFs estimated via sieve regression with the ones obtained by setting $F(x) = \max(0, x)$, as well as linear IRFs.
\begin{figure}[t!]
	\centering
	\begin{subfigure}[b]{\textwidth}
		\centering
		\includegraphics[width=\textwidth]{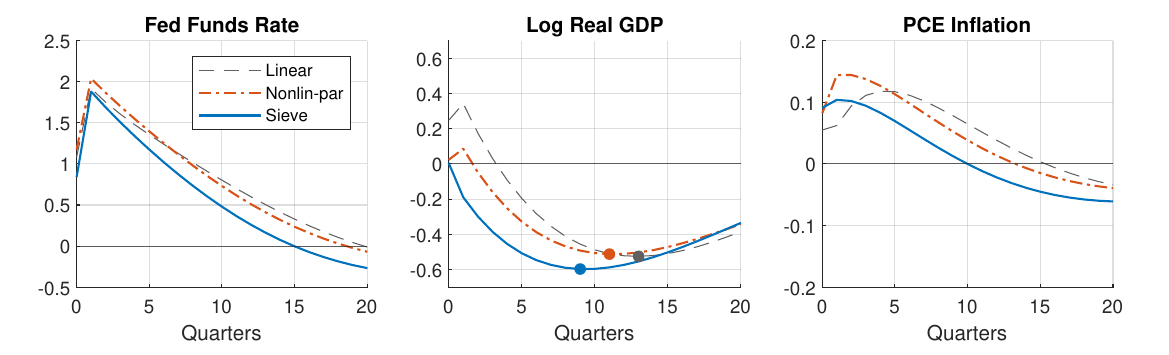}
		\caption{$\delta = +1$}
	\end{subfigure}
	\\[15pt]
	\begin{subfigure}[b]{\textwidth}
		\centering
		\includegraphics[width=\textwidth]{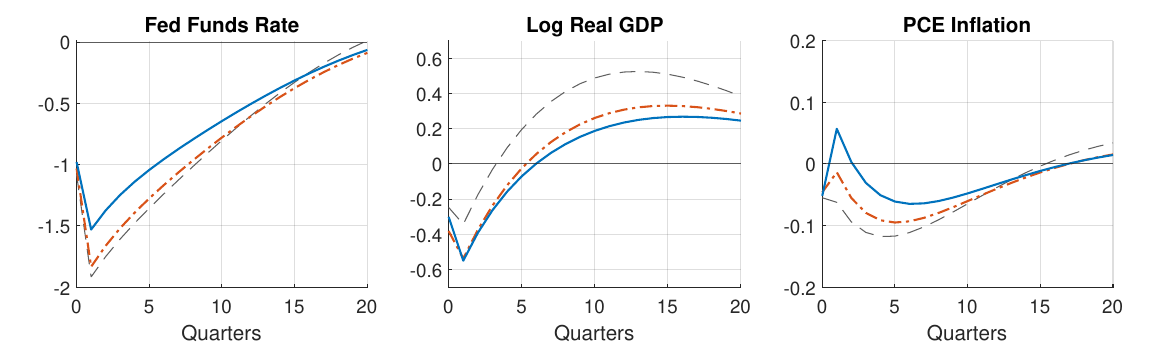}
		\caption{$\delta = -1$}
	\end{subfigure}
	\\[5pt]
	\caption{Effect of an unexpected U.S. monetary policy shock on federal funds rate, GDP and inflation. Linear (gray, dashed), parametric nonlinear with $F(x) = \max(0, x)$ (red, point-dashed) and sieve (blue, solid) structural impulse responses. For $\delta = +1$, the lowest point of the GDP response is marked with a dot. Note that $\delta = 1 \approx 1.7 \times \sigma_{\epsilon,1}$.}
	\label{fig:app_gonc2021_irfs}
\end{figure}
Figure \ref{fig:app_gonc2021_irfs} plots estimated impulse responses to both positive and negative monetary policy shocks. 
The impact on the federal funds rate is consistent across all three procedures. The semiparametric nonlinear response for GDP, unlike in the case of linear and parametric nonlinear IRFs, is nearly zero at impact and has a monotonic decrease until around 10 quarters ahead. The change in shape is meaningful, as the procedure of \cite{goncalvesImpulseResponseAnalysis2021} still yields a small short-term upward jump in GDP when a monetary tightening shock hits. Moreover, after the positive shock, the sieve GDP responses reaches its lowest value 4 and 2 quarters before the linear and parametric nonlinear responses, while its size is 13\% and 16\% larger, respectively.\footnote{The strength of this effect changes across different shock sizes, as Figure \ref{fig:plot_app_gonc2021_scale} in Appendix \ref{appendix:additional_plots} proves. As shock sizes get smaller, nonlinear IRFs, both parametric and sieve, show decreasing negative effects.} Finally, the sieve PCE response is positive for a shorter interval, but looks to be more persistent once it turns negative, also 10 months after impact.

When the shock is expansionary, one sees that the semiparametric FFR response is marginally mitigated compared to the alternative estimates. An important puzzle is due to the negative impact on GDP: Both types of nonlinear responses show a drop in output in the first 5 quarters. Such a quick change seems unrealistic, as one does not expect inflation to suddenly reverse sign, but, as \cite{goncalvesImpulseResponseAnalysis2021} also remark, the overall impact on inflation of both shocks is small when compared to the change in federal funds rate.

\subsection{Uncertainty Shocks}\label{subsection:app_istrefi}

Traditional central bank policymaking is heavily guided by the principle that a central bank can and should influence expectations. Therefore, controlling the (perceived) level of ambiguity in current and future commitments is key. \cite{istrefiSubjectiveInterestRate2018} provide an analysis of the impact of unforeseen changes in the level of subjective interest rate uncertainty on the macroeconomy.  
For the sake of simplicity, our evaluation will focus only on their 3-month-ahead uncertainty measure for short-term interest rate maturities (3M3M) and the U.S. economy.
Like in \cite{istrefiSubjectiveInterestRate2018}, let $Z_t = (X_t, \textnormal{IP}_t, \textnormal{CPI}_t, \textnormal{PPI}_t, \textnormal{RT}_t, \textnormal{UR}_t)'$ be a vector where $X_t$ is the chosen uncertainty measure, $\textnormal{IP}_t$ is the (log) industrial production index, $\textnormal{CPI}_t$ is the CPI inflation rate, $\textnormal{PPI}_t$ is the producer price inflation rate, $\textnormal{RT}_t$ is (log) retail sales and $\textnormal{UR}_t$ is the unemployment rate. The nonlinear model specification is given by
%
%
$Z_t = \mu + A_1 Z_{t-1} + A_2 Z_{t-1} + F_1(X_{t-1}) + F_2(X_{t-2}) + D W_t + u_t ,$
where $W_t$ includes a linear time trend and oil price $\textnormal{OIL}_t$.\footnote{Inclusion of linear exogenous variables in the semiparametric theoretical framework in Section~\ref{section:estimation} is straightforward as long as one can assume that they are stationary and weakly dependent. The choice of using $p=2$ is identical to that of the original authors, based on BIC.} 
The data has a monthly frequency and spans the period between May 1993 and July 2015.\footnote{We utilize the original data employed by the authors, who kindly shared it upon request. However, we rescale retail sales ($\textnormal{RT}_t$) so that the level in January 2000 equals 100.} 
Note here that nonlinear functions $F_1$ and $F_2$ are assumed not to affect $X_t$, which is the structural variable. The linear VAR specification of \cite{istrefiSubjectiveInterestRate2018} is recovered by simply assuming $F_1 = F_2 = 0$ prior to estimation. Since they use recursive identification and order the uncertainty measure first, this model too is block-recursive.
We consider a positive shock with intensity $\delta = \sigma_{\epsilon,1}$, where $\sigma_{\epsilon,1}$ is the standard deviation of structural innovations. In this empirical exercise, the relaxation function is  
$\rho(z) = \mathbb{I}\left\{ |z| \leq 1/4 \right\} \exp( 1 + ( \lvert 4 x \rvert^{8} - 1 )^{-1} ) $
%
%
and we set $\{0.1, 0.3\}$ to be the cubic spline knots. 
As 3M3M is a non-negative measure of uncertainty, some care must be taken to make sure that the shocked paths for $X_t$ do not reach negative values. Figure \ref{fig:plot_app_istrefi2018_meanshift} in Appendix \ref{appendix:additional_plots} shows that the relaxation function is compatible, and also that the shocked nonlinear paths of $X_t$ with impulse $\delta$ and $\delta'$ all do not cross below zero.

\begin{figure}[t!]
	\centering
	\includegraphics[width=\textwidth]{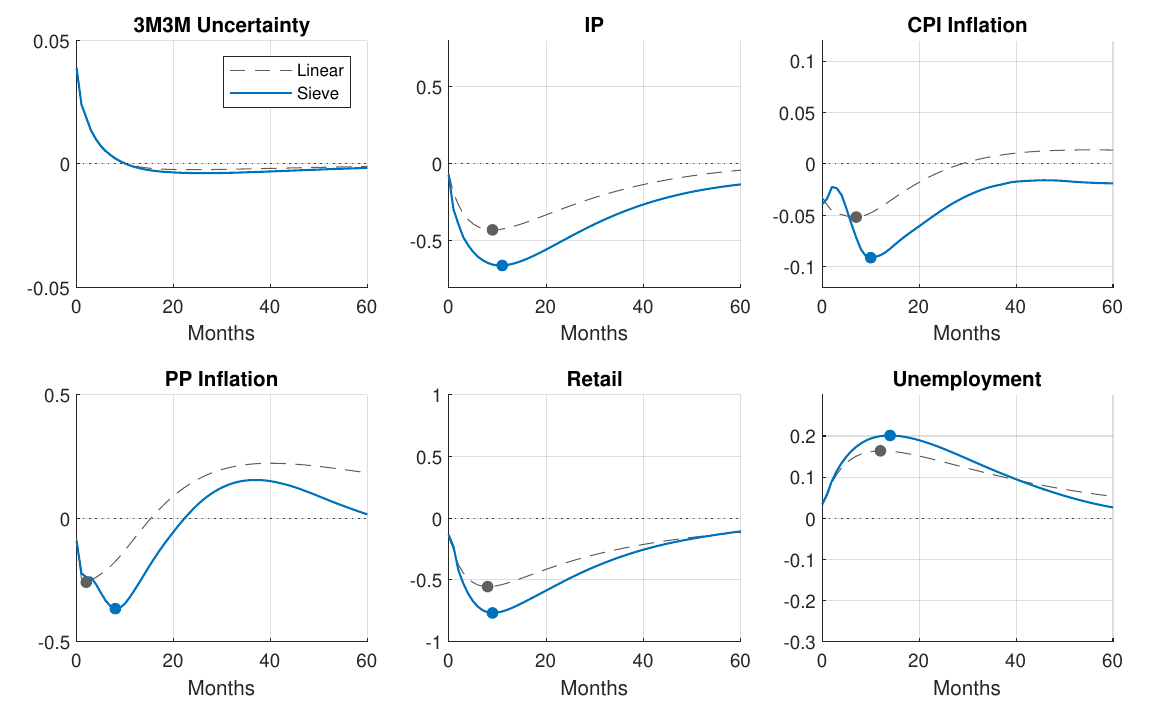}
	\\[5pt]
	\caption{Effect of an unexpected, one-standard-deviation uncertainty shock to U.S. macroeconomic variables. Linear (gray, dashed) and sieve (blue, solid) structural impulse responses. The extreme points of the responses are marked with a dot.}
	\label{fig:app_istrefi2018_irfs}
\end{figure}

Figure \ref{fig:app_istrefi2018_irfs} presents both the linear and nonlinear structural impulse responses obtained. Importantly, even though \cite{istrefiSubjectiveInterestRate2018} estimate a Bayesian VAR model and here we consider a frequentist vector autoregressive benchmark, the shape of the IRFs is retained, cf. the median response in the top row of their Figure 4. When uncertainty increases, industrial production drops, and the size and extent of this decrease are intensified in the nonlinear responses. The sieve IP response reaches a value that is $54\%$ lower than that of the respective linear IRF.\footnote{Figure \ref{fig:plot_app_istrefi2018_scale_IP} in Appendix \ref{appendix:additional_plots} confirms that this difference is consistent over a range of shock sizes, too.} A similar behavior holds true for retail sales ($38\%$ lower) and unemployment ($23\%$ higher), proving that this shock is more profoundly contractionary than suggested by the linear VAR model. 
Further, CPI and PP inflation both display short-term fluctuations, which strengthen the short- and medium-term impact of the shock. CPI and PP nonlinear inflation responses are $76\%$ and $41\%$ stronger than their linear counterpart, respectively. These differences show that linear IRFs might be both under-estimating the short-term intensity and misrepresenting the long-term persistence of inflation reactions.  
Given the strength of nonlinear IRFs, this discrepancy may also suggest that the 3M3M uncertainty measure partially captures the financial channel, too.
Hence, we believe our analysis provides some evidence that the linear VAR used by \cite{istrefiSubjectiveInterestRate2018} may miss some key impulse response features.\footnote{See also Figure \ref{fig:plot_app_istrefi2018_regfuns}, which plots the estimated functions of the endogenous variables.}

\section{Conclusion}\label{section:conclusion}

This paper studies the application of semiparametric series estimation to the problem of structural impulse response analysis for time series. After first discussing partial block-recursive identification, we have shown that, for models with moderate physical dependence, series estimation can be employed and structural IRFs are consistently estimated. Simulations showcase that this approach is both valid in moderate samples and has the added benefit of being robust to misspecification of the nonlinear model components. Finally, two empirical applications showcase the potential insights gained by departing from either linear or parametric nonlinear specifications when estimating structural responses.

A key aspect that we have not touched upon is inference in the form of confidence intervals. The development of inferential theory appears feasible in light of the uniform inference results obtained by e.g. \cite{belloniNewAsymptoticTheory2015} in the i.i.d. setting and \cite{liUniformNonparametricInference2020} for time series data, and it is an important direction for future research.
Studying other sieve spaces, such as neural networks \citep{chenImprovedRatesAsymptotic1999,farrellDeepNeuralNetworks2021a} or shape-preserving sieves \citep{chenChapter76Large2007}, would also be highly desirable. Finally, in the spirit of \cite{kangINFERENCENONPARAMETRICSERIES2021}, deriving new inference results that are uniform in the selection of series terms is important, as, in practice, the sieve should be tuned in a data-driven way.


\pagebreak
\bibliography{nonlin_irf_bib}

\newpage
\appendix
\setcounter{page}{1}   
\renewcommand{\thefootnote}{$\ast$} 
\onehalfspacing


\begin{center}%
    {\Large%
	APPENDIX
    }\\[0.5em]
\end{center}
\vspace{-1em}

\setcounter{footnote}{0} 
\renewcommand*{\thefootnote}{\arabic{footnote}}

\setcounter{figure}{0}
\renewcommand\thefigure{\thesection.\arabic{figure}} 



\section{Proofs}\label{appendix:proofs}

\paragraph{Functional Notation.} In this appendix, to streamline and simplify notation, we will often write $G X$, where $G(\cdot)$ is a function and $X$ is a (random) variable, to mean $G(X)$, as mentioned in Section~\ref{section:estimation}.

\subsection{Theorem~\ref{theorem:twostep_estimator_consistency}}

Before delving into the proof of Theorem \ref{theorem:twostep_estimator_consistency}, note that we can decompose $\widehat{\Pi}_2 - {\Pi}_2$ as
\begin{equation*}
	\widehat{\Pi}_2 - {\Pi}_2 = (\widehat{\Pi}_2 - \widehat{\Pi}^*_2) + (\widehat{\Pi}^*_2 - \widetilde{\Pi}_2) + (\widetilde{\Pi}_2 - {\Pi}_2) ,
\end{equation*}
where $\widetilde{\Pi}_2$ is the projection of $\Pi_2$ onto the linear space spanned by the sieve.
The last two terms can be handled directly with the theory developed by \cite{chenOptimalUniformConvergence2015}. Specifically, their Lemma 2.3 controls the second term (variance term), while Lemma 2.4 handles the third term (bias term). This means here we can focus on the first term, which is due to using generated regressors $\widehat{\epsilon}_{1t}$ in the second step.

Since $\widehat{\Pi}_2$ can be decomposed in $d_Y$ rows of semiparametric coefficients, we further reduce to the scalar case. Let $\pi_2$ be any row of $\Pi_2$ and, with a slight abuse of notation, $Y$ the vector of observations of the component of $Y_t$ of the same row, so that one may write
\begin{align*}
	\widehat{\pi}_2(x) - \widehat{\pi}^*_2(x)
	& = 
	\widetilde{b}^K_\pi(x) \big( \widehat{\widetilde{B}}_\pi' \widehat{\widetilde{B}}_\pi \big)^- \big( \widehat{\widetilde{B}}_\pi - {\widetilde{B}}_\pi \big)'Y 
	+
	\widetilde{b}^K_\pi(x) \left[ \big( \widehat{\widetilde{B}}_\pi' \widehat{\widetilde{B}}_\pi \big)^- - \big( {\widetilde{B}}_\pi' {\widetilde{B}}_\pi \big)^- \right] {\widetilde{B}}_\pi'Y \\
	& = I + II
\end{align*}
where $\widetilde{b}^K_\pi(x) = \Gamma_{B,2}^{-1/2}\allowbreak {b}^K_\pi(x)$ is the orthonormalized sieve according to $\Gamma_{B,2} := \E[ {b}^K_\pi(W_{2t})\allowbreak {b}^K_\pi(W_{2t})' ]$, ${\widetilde{B}}_\pi$ is the \textit{infeasible} orthonormalized design matrix (involving $\epsilon_{1t}$) and $\widehat{\widetilde{B}}_\pi$ is \textit{feasible} orthonormalized design matrix (involving $\widehat{\epsilon}_{1t}$). In particular, note that
\begin{equation*}
	\widehat{B}_\pi = B_\pi + R_n, 
	\quad \text{where} \quad 
	R_n := \begin{bmatrix}
		0 & & 0 & \widehat{\epsilon}_{11} - \epsilon_{11} \\
		\vdots & \cdots & \vdots & \vdots \\
		0 & & 0 & \widehat{\epsilon}_{1n} - \epsilon_{1n}
	\end{bmatrix}
	\in \Real^{n \times K} ,
\end{equation*}
which implies $\widehat{\widetilde{B}}_\pi - {\widetilde{B}}_\pi = R_n \,\Gamma_{B,2}^{-1/2} =: \widetilde{R}_n$. 

The next Lemma provides a bound for the difference $( \widehat{\widetilde{B}}_\pi' \widehat{\widetilde{B}}_\pi / n ) - ( \widetilde{B}_\pi' \widetilde{B}_\pi / n )$ that will be useful in the proof of Theorem~\ref{theorem:twostep_estimator_consistency} below.

\begin{lemma}\label{lemma:bound_difference_BBpi}
	Under the setup of Theorem~2.1 in \cite{chenOptimalUniformConvergence2015}, it holds
	\begin{equation*}
		\big\lVert ( \widehat{\widetilde{B}}_\pi' \widehat{\widetilde{B}}_\pi / n ) - ( \widetilde{B}_\pi' \widetilde{B}_\pi / n ) \big\rVert = O_P(\sqrt{K / n}) .
	\end{equation*}
\end{lemma}

\begin{proof}
	Using the expansion
	$\widehat{\widetilde{B}}_\pi' \widehat{\widetilde{B}}_\pi = \widetilde{B}_\pi' \widetilde{B}_\pi + (\widetilde{B}_\pi'\widetilde{R}_n + \widetilde{R}_n'\widetilde{B}_\pi) + \widetilde{R}_n'\widetilde{R}_n$,
	one immediately finds that
	$
	\big\lVert ( \widehat{\widetilde{B}}_\pi' \widehat{\widetilde{B}}_\pi / n ) - ( \widetilde{B}_\pi' \widetilde{B}_\pi / n ) \big\rVert \leq 2 \big\lVert \widetilde{B}_\pi'\widetilde{R}_n / n \big\rVert + \big\lVert \widetilde{R}_n'\widetilde{R}_n / n \big\rVert .
	$
	The second right-hand side factor satisfies $\big\lVert \widetilde{R}_n'\widetilde{R}_n / n \big\rVert \leq \lambda_{K,n}^2 \big\lVert {R}_n'{R}_n / n \big\rVert$. Moreover, 
	\begin{align*}
		\big\lVert R_n'R_n / n \big\rVert
		& = \left\lVert \frac{1}{n} \sum_{t=1}^n (\widehat{\epsilon}_{1t} - \epsilon_{1t})^2 \right\rVert
		= \left\lVert \frac{1}{n} \sum_{t=1}^n (\Pi_1 - \widehat{\Pi}_1)' W_{1t} W_{1t}' (\Pi_1 - \widehat{\Pi}_1) \right\rVert \\
		& \leq \big\lVert \Pi_1 - \widehat{\Pi}_1 \big\rVert^2 \, \big\lVert W_1' W_1 / n \big\rVert 
		= O_P(n^{-1}) ,
	\end{align*}
	since $\lVert W_1' W_1 / n \rVert = O_P(1)$. Under Assumption \ref{assumption:sieve_type}, $\lambda_{K,n}^2 / n = o_P(\sqrt{K / n})$ since B-splines and wavelets satisfy $\lambda_{K,n} \lesssim 1$. Consequently, $\big\lVert \widetilde{R}_n'\widetilde{R}_n / n \big\rVert = o_P(\sqrt{K / n})$. 
	Factor $\lVert \widetilde{B}_\pi'R_n / n \rVert$ is also straightforward, but depends on sieve dimension $K$,
	\begin{align*}
		\big\lVert \widetilde{B}_\pi'R_n / n \big\rVert 
		& \leq \left\lVert \frac{1}{n} \sum_{t=1}^n \widetilde{b}_\pi^K(W_{2t}) (\widehat{\epsilon}_{1t} - \epsilon_{1t}) \right\rVert 
		= \left\lVert \frac{1}{n} \sum_{t=1}^n \widetilde{b}_\pi^K(W_{2t}) W_{1t}' (\Pi_1 - \widehat{\Pi}_1) \right\rVert \\
		& \leq \big\lVert \Pi_1 - \widehat{\Pi}_1 \big\rVert \, \big\lVert \widetilde{B}_\pi' W_1 / n \big\rVert 
		= O_P(\sqrt{K / n}) ,
	\end{align*}
	since $\lVert \widetilde{B}_\pi' W_1 / n \rVert = O_P(\sqrt{K})$ as the column dimension of $W_1$ is fixed. The claim then follows by noting $O_P(\sqrt{K / n})$ is the dominating order of convergence.
\end{proof}

\begin{proof}[Proof of Theorem \ref{theorem:twostep_estimator_consistency}]
	Since $\widehat{\Pi}_1$ is a least squares estimator of a linear equation, the rate of convergence is the parametric rate $n^{-1/2}$. The first result is therefore immediate.
	For the second step, we use
	$
	\big\lVert \widehat{\Pi}_2 - {\Pi}_2 \big\rVert_\infty 
	\leq \big\lVert \widehat{\Pi}_2 - \widehat{\Pi}^*_2 \big\rVert_\infty + \big\lVert \widehat{\Pi}^*_2 - {\Pi}_2 \big\rVert_\infty ,
	$
	and bound explicitly the first right-hand side term. For a given component of the regression function,
	\begin{equation*}
		| \widehat{\pi}_2(x) - \widehat{\pi}^*_2(x) |
		\leq | I | + | II | .
	\end{equation*}
	We now control each term on the right side.
	\begin{itemize}
		\item[(1)] It holds 
		\begin{align*}
			| I | 
			& \leq \lVert \widetilde{b}^K_\pi(x) \rVert \, \big\lVert \big( \widehat{\widetilde{B}}_\pi' \widehat{\widetilde{B}}_\pi / n \big)^- \big\rVert \, \big\lVert \big( \widehat{\widetilde{B}}_\pi - {\widetilde{B}}_\pi \big)'Y / n \big\rVert \\
			& \leq \sup_{x \in \mathcal{W}_2}\lVert \widetilde{b}^K_\pi(x) \rVert \, \big\lVert \big( \widehat{\widetilde{B}}_\pi' \widehat{\widetilde{B}}_\pi / n \big)^- \big\rVert \, \big\lVert \big( \widehat{\widetilde{B}}_\pi - {\widetilde{B}}_\pi \big)'Y / n \big\rVert \\
			& \leq \zeta_{K,n} \lambda_{K,n} \, \big\lVert \big( \widehat{\widetilde{B}}_\pi' \widehat{\widetilde{B}}_\pi / n \big)^- \big\rVert \, \big\lVert \big( \widehat{\widetilde{B}}_\pi - {\widetilde{B}}_\pi \big)'Y / n \big\rVert .
		\end{align*}
		
		Let $\mathcal{A}_n$ denote the event on which $\big\lVert \widehat{\widetilde{B}}_\pi' \widehat{\widetilde{B}}_\pi / n - I_K \big\rVert \leq 1/2$, so that $\big\lVert \big( \widehat{\widetilde{B}}_\pi' \widehat{\widetilde{B}}_\pi / n \big)^- \big\rVert \leq 2$ on $\mathcal{A}_n$. 
		Notice that since $\lVert ( \widehat{\widetilde{B}}_\pi' \widehat{\widetilde{B}}_\pi / n ) - ( \widetilde{B}_\pi' \widetilde{B}_\pi / n ) \rVert = o_P(1)$ (Lemma \ref{lemma:bound_difference_BBpi}) and, by assumption, $\lVert \widetilde{B}_\pi' \widetilde{B}_\pi / n - I_K \rVert = o_P(1)$, then $\P(\mathcal{A}_n^\text{c}) = o(1)$.
		On $\mathcal{A}_n$, then
		\begin{equation*}
			| I | \lesssim \zeta_{K,n} \lambda_{K,n}^2 \, \big\lVert \big( \widehat{B}_\pi - {B}_\pi \big)'Y / n \big\rVert = \zeta_{K,n} \lambda_{K,n}^2 \, \big\lVert R_n' Y / n \big\rVert .
		\end{equation*}
		From $R_n' Y = \sum_{t=1}^n b_\pi^K(W_{2t}) (\widehat{\epsilon}_{1t} - \epsilon_{1t}) Y_t = (\Pi_1 - \widehat{\Pi}_1)' W_1' Y$ it follows that 
		$
		\big\lVert R_n' Y / n \big\rVert \leq \big\lVert \Pi_1 - \widehat{\Pi}_1 \big\rVert \, \big\lVert W_1' Y / n \big\rVert
		$
		on $\mathcal{A}_n$, meaning
		$
		| I | = O_P\left( \zeta_{K,n} \lambda_{K,n}^2 / \sqrt{n} \right)
		$
		as $\lVert W_1' Y / n \rVert = O_P(1)$ and $\P(\mathcal{A}_n^\text{c}) = o(1)$.
		\item[(2)] Again we proceed by uniformly bounding $II$ according to
		\begin{equation*}
			| II | \leq \zeta_{K,n} \lambda_{K,n} \, \big\lVert \big( \widehat{\widetilde{B}}_\pi' \widehat{\widetilde{B}}_\pi / n \big)^- - \big( {\widetilde{B}}_\pi' {\widetilde{B}}_\pi / n \big)^- \big\rVert \, \big\lVert {\widetilde{B}}_\pi' Y / n \big\rVert .
		\end{equation*}
		The last factor has order $\lVert {\widetilde{B}}_\pi' Y / n \rVert = O_P(\sqrt{K})$ since $\widetilde{B}_\pi$ is growing in row dimension with $K$. 
		For the middle term, introduce
		$
		\Delta_B := \widehat{\widetilde{B}}_\pi' \widehat{\widetilde{B}}_\pi / n - {\widetilde{B}}_\pi' {\widetilde{B}}_\pi / n
		$
		and event
		$
		\mathcal{B}_n := \left\{ \big\lVert \big( {\widetilde{B}}_\pi' {\widetilde{B}}_\pi / n \big)^- \, \Delta_B \big\rVert \leq 1/2 \right\} \cap \left\{ \big\lVert {\widetilde{B}}_\pi' {\widetilde{B}}_\pi / n - I_K \big\rVert \leq 1/2 \right\} .
		$
		On $\mathcal{B}_n$, we can apply the bound \citep{horn2012matrix}
		\begin{equation*}
			\big\lVert \big( \widehat{\widetilde{B}}_\pi' \widehat{\widetilde{B}}_\pi / n \big)^- - \big( {\widetilde{B}}_\pi' {\widetilde{B}}_\pi / n \big)^- \big\rVert
			\leq
			\frac
			{\lVert ( {\widetilde{B}}_\pi' {\widetilde{B}}_\pi / n )^- \rVert^2 \, \lVert \Delta_B \rVert}
			{1 - \lVert ( {\widetilde{B}}_\pi' {\widetilde{B}}_\pi / n )^- \, \Delta_B \rVert}
			\lesssim
			\big\lVert \widehat{\widetilde{B}}_\pi' \widehat{\widetilde{B}}_\pi / n - {\widetilde{B}}_\pi' {\widetilde{B}}_\pi / n \big\rVert .
		\end{equation*}
		Since $\big\lVert \widehat{\widetilde{B}}_\pi' \widehat{\widetilde{B}}_\pi / n - {\widetilde{B}}_\pi' {\widetilde{B}}_\pi / n \big\rVert = O_P(\sqrt{K / n})$ by Lemma \ref{lemma:bound_difference_BBpi}, we get
		\begin{equation*}
			| II | =  O_P\left( \zeta_{K,n} \lambda_{K,n} \frac{K}{\sqrt{n}} \right)
		\end{equation*}
		on $\mathcal{B}_n$. Finally, using $\P((A \cap B)^\text{c}) \leq \P(A^\text{c}) + \P(B^\text{c})$ we note that $\P(\mathcal{B}_n^\text{c}) =  o(1)$ so that the bound asymptotically holds irrespective of event $\mathcal{B}_n$.
	\end{itemize}
	Thus, we have shown that
	\begin{align*}
		| \widehat{\pi}_2(x) - \widehat{\pi}^*_2(x) |
		& \leq
		O_P\left( \zeta_{K,n} \lambda_{K,n}^2 \frac{1}{\sqrt{n}} \right)
		+
		O_P\left( \zeta_{K,n} \lambda_{K,n} \frac{K}{\sqrt{n}} \right) 
		= 
		O_P\left( \zeta_{K,n} \lambda_{K,n} \frac{K}{\sqrt{n}} \right)
	\end{align*}
	as clearly $\sqrt{n}^{-1} = o({K}/{\sqrt{n}})$ and, as discussed in the proof of Lemma \ref{lemma:bound_difference_BBpi}, $\lambda_{K,n}^2 / n = o_P(\sqrt{K / n})$. This bound is uniform in $x$ and holds for each of the (finite number of) components of $\widehat{\Pi}_2$; therefore, the proof is complete.
\end{proof}

\subsection{Corollary~\ref{corollary:twostep_estimator_op1}}

First, let us note that $\zeta_{K,n} \lesssim \sqrt{K}$. This holds since (i) for any constant and linear component, compact $\mathcal{Y}$, $\mathcal{X}$, and $\mathcal{E}$ imply a trivial upper bound; (ii) additive (tensor product) B-spline and wavelets sieve basis elements individually satisfy $\zeta_{K,n} \lesssim \sqrt{K}$ \citep{belloniNewAsymptoticTheory2015}.

To prove Corollary~\ref{corollary:twostep_estimator_op1}, we will apply Lemmas~2.3 and 2.4 of \cite{chenOptimalUniformConvergence2015} together with their generic bound on the sup-norm of the empirical sieve projection operator in Remark~2.5. While this bound is generic, it is sufficient for our purposes, even as we consider sieves that are \textit{not} constructed as a tensor product of known sieve bases. 

\begin{proof}[Proof of Corollary~\ref{corollary:twostep_estimator_op1}]
    In our setting the empirical projection operator associated to $\mathcal{B}_\pi$ is 
    $ 
        P_{K,n} := \widetilde{b}^K_n(w) (\widetilde{B}_\pi' \widetilde{B}_\pi)^{-} \widetilde{B}_\pi' \Pi_2
    $, 
    for which we introduce the uniform operator norm 
    \begin{equation*}
        \lVert P_{K,n} \rVert_\infty
        :=
        \sup_{f \in L^{\infty}_\pi(\mathcal{W}_2) \,:\, \lVert f \rVert_\infty = 1}
        \lVert P_{K,n} f \rVert_\infty ,
    \end{equation*}
    where $L^{\infty}_\pi(\mathcal{W}_2)$ is the class of uniformly bounded functions with the same semiparametric additive structure as $\Pi_2$.
    Under the same assumptions of Theorem~\ref{theorem:twostep_estimator_consistency}, the conditions required by \cite{chenOptimalUniformConvergence2015}, Lemmas~2.3 and 2.4, and Remark~2.5 are satisfied. 
    In particular, since we assume $\epsilon_t$ to have compact range, it holds that $\E[\vert \xi_{2t,i} \vert^{2+\delta}] < \infty$ for all $\delta > 0$ and $1 \leq i \leq d_Y$, uniformly in $t$.
    By Lemma~\ref{lemma:gmc_sample_gram_matrix_op1}, and the choice of $K$, $\zeta_{K,n} \lambda_{K,n} \sqrt{K \log(n) / n} = O(1)$ is fulfilled under Assumption~\ref{assumption:physical_dep} (geometric moment contractivity, see Appendix~\ref{appendix_subsec:physical_dep} below).
    Therefore we obtain
    \begin{equation*}
		\lVert \widehat{\Pi}^*_2 - \Pi_2 \rVert_\infty 
        \leq
        O_P\left( \zeta_{K,n} \lambda_{K,n} \, \sqrt{\frac{\log(n)}{n}} \right)
        + \sqrt{2} \zeta_{K,n} \lambda_{K,n} \inf_{\widetilde{\Pi}_2 \in \textnormal{clsp}(\mathcal{B}_\pi)} \lVert \Pi_2 - \widetilde{\Pi}_2 \rVert_\infty ,
	\end{equation*}
    where $\textnormal{clsp}(\mathcal{B}_\pi)$ is the closed linear span of sieve $\mathcal{B}_\pi$.
    
    Given that the dimension of the nonlinear autoregressive system $d_Y$ is fixed, we can consider directly $\lVert \pi_{2,i} - \widetilde{\pi}_{2,i} \rVert_\infty$. The additive structure of the model, fixed lag $p$, and compactness of $\mathcal{W}_2$ allow to decompose the approximation error into
    \begin{equation*}
        \lVert \pi_{2,i} - \widetilde{\pi}_{2,i} \rVert_\infty
        \lesssim
        \lVert \mu_{2,i} - \widetilde{\mu}_{2,i} \rVert
        + \sum_{j=1}^p \sum_{\ell=1}^{d_Y} \lVert G^Y_{2,ij\ell} - \widetilde{G}^Y_{2,ij\ell} \rVert_\infty
        + \sum_{m=0}^p \lVert G^X_{2,im} - \widetilde{G}^X_{2,im} \rVert_\infty,
    \end{equation*}
    where $G^Y_{2,\cdot}$ and $G^X_{2,\cdot}$ are used to indicate the one-dimensional nonlinear functions of $Y_t$ (and its components) and $X_t$, and their lags, respectively. The error of estimating the intercept and any strictly linear function is zero under correct specification of the semiparametric sieve, while for nonlinear components the approximation error of splines and wavelets series decays as $K^{-s}$ \citep{chenOptimalUniformConvergence2015,belloniNewAsymptoticTheory2015}. Hence,
    \begin{equation*}
        \lVert \pi_{2,i} - \widetilde{\pi}_{2,i} \rVert_\infty
        \lesssim
        K^{-s} 
    \end{equation*}
    and 
    \begin{equation*}
        \lVert \widehat{\Pi}^*_2 - \Pi_2 \rVert_\infty
        \lesssim
        O_P\left( \sqrt{\frac{K \log(n)}{n}} \right)
        + O\left( K^{1/2-s} \right) 
    \end{equation*}
    by using that $\zeta_{K,n} \lambda_{K,n} \lesssim \sqrt{K}$.

    Let $s = 1/2 + \nu$, where by assumption $\nu \geq 1/2$. Then, $K \asymp (n / \log(n))^{1/(2 s + 1)}$ implies
    \begin{align*}
        \sqrt{\frac{K \log(n)}{n}} 
            & \:\asymp\: \left( \frac{\log(n)}{n} \right)^{\frac{1}{3}} , \\
        K^{1/2-s} 
            & \:\asymp\: \left( \frac{\log(n)}{n} \right)^{\frac{\nu}{2s+1}}
    \end{align*}
    and, lastly, for the first term in the bound of Theorem~\ref{theorem:twostep_estimator_consistency},
    \begin{equation*}
        \zeta_{K,n} \lambda_{K,n} \, \frac{K}{\sqrt{n}}
        \:\lesssim\:
        \frac{K^{\frac{3}{2}}}{\sqrt{n}}
        \:\asymp\:
        n^{- \frac{\nu - 1/2}{2s+1}}
            \log(n)^{-\frac{3}{2(2s+1)}} .
    \end{equation*}
    To conclude, we simply note that, while $o(1)$, the order obtained in the last display is the slowest, hence dominating, rate.
\end{proof}

\begin{remark}
    The proof of Corollary~\ref{corollary:twostep_estimator_op1} can be easily extended to allow for more complex nonlinear relationships in the autoregressive model, such as, for example, a general nonlinear dependence on $X_t, \ldots, X_{t-p}$. Then, $\lVert \pi_{2,i} - \widetilde{\pi}_{2,i} \rVert_\infty$ will coverge at rate $K^{-s/d^{\dagger}}$, where $d^{\dagger} \geq 1$ is the maximum dimensionality among all nonlinear components in the equation of $Y_{t,i}$.
    We avoid this extension here to keep notational complexity low in the additive decomposition of error $\lVert \pi_{2,i} - \widetilde{\pi}_{2,i} \rVert_\infty$.
    Finally, note that a sharper upper bound on $\lVert P_{K,n} \rVert_\infty$ (e.g. $\lVert P_{K,n} \rVert_\infty \lesssim 1$ as proven for B-splines by \cite{huangLocalAsymptoticsPolynomial2003a} and wavelets by \cite{chenOptimalUniformConvergence2015}) would not result in an improved rate for $\widehat{\Pi}_2$, since $\zeta_{K,n} \lambda_{K,n} {K}/{\sqrt{n}}$ remains the dominating factor.
\end{remark}

\subsection{Theorem \ref{theorem:consistent_irf}}

Before proving impulse response consistency, we show that the functional moving average coefficient matrices $\Gamma_j$ can be consistently estimated with $\widehat{\Pi}_1$ and $\widehat{\Pi}_2$.

\begin{lemma}\label{lemma:Gamma_j_consistent}
	Under the assumptions of Theorem \ref{theorem:twostep_estimator_consistency} and for any fixed integer $j \geq 0$, it holds 
	\begin{equation*}
		\lVert \widehat{\Gamma}_j - \Gamma_j \rVert_\infty = o_P(1) .
	\end{equation*}
\end{lemma}

\begin{proof}
	By definition, recall that $\Gamma(L) = \Psi(L) G(L)$ where $\Psi = (I_d - A(L)L)^{-1}$. Since $\Psi(L)$ is an MA($\infty$) lag polynomial, we have that 
	$
	\Gamma(L) = ( \sum_{k=0}^\infty \Psi_k L^k ) ( G_{0} + G_{1} L \allowbreak + \ldots \allowbreak + G_p L^p ) ,
	$
	where $\Psi_0 = I_d$, $\{\Psi_k\}_{k=1}^\infty$ are purely real matrices and $G_0$ is a functional vector that may also contain linear components (i.e. allow linear functions of $X_t$). This means that $\Gamma_j$ is a convolution of real and functional matrices,
	$
	\Gamma_j = \sum_{k=1}^{\min\{j,\, p\}} \Psi_{j-k} G_{k} .
	$
	The linear coefficients of $A(L)$ can be consistently estimated by $\widehat{\Pi}_1$ and $\widehat{\Pi}_2$, and, thus, the plug-in estimate $\widehat{\Psi}_j$ is consistent for $\Psi_j$ \citep{lutkepohlNewIntroductionMultiple2005}. 
	Therefore,
	\begin{align*}
		\lVert \widehat{\Gamma}_j - \Gamma_j \rVert_\infty 
		& \leq \sum_{k=1}^{\min\{j,\, p\}} \left\lVert \Psi_{j-k} -  \widehat{\Psi}_{j-k} \right\rVert_\infty \left\lVert {G}_{k} \right\rVert_\infty + \left\lVert \widehat{\Psi}_{j-k} \right\rVert_\infty \left\lVert G_{k} - \widehat{G}_{k} \right\rVert_\infty \\
		& \leq \sum_{k=1}^{\min\{j,\, p\}} o_p(1) C_{G,k} + O_P(1) o_p(1)
		= o_p(1) ,
	\end{align*}
	where $C_{G,k}$ is a constant and $\lVert G_{k} - \widehat{G}_{k} \rVert_\infty = o_p(1)$ as a direct consequence of Theorem~\ref{theorem:twostep_estimator_consistency}. 
\end{proof}


Recall now that the sample estimate for the relaxed-shock impulse response is 
\begin{equation*}
	\widehat{\widetilde{\textnormal{IRF}}}_{h,\ell}(\delta)
	= 
	\widehat{\Theta}_{h,\cdot 1} \delta \, n^{-1} \sum_{t=1}^{n} \rho(\widehat{\epsilon}_{1t}) 
	+
	\sum_{j=0}^h \widehat{V}_{j,\ell}(\delta)
\end{equation*}
where
\begin{equation*}
	\widehat{V}_{j,\ell}(\delta) 
	= 
	\frac{1}{n-j} \sum_{t=1}^{n-j} \widehat{v}_{j,\ell}\big( X_{t+j:t}; \widehat{\widetilde{\delta}}_t \big)
	=
	\frac{1}{n-j} \sum_{t=1}^{n-j} \left[ \widehat{\Gamma}_j \widehat{\gamma}_{j}(X_{t+j:t}; \widehat{\widetilde{\delta}}_t) - \widehat{\Gamma}_j X_{t+j} \right] .
\end{equation*}
Therefore, the estimated horizon $h$ impulse response of the $\ell$th variable is
\begin{equation*}
	\widehat{\widetilde{\textnormal{IRF}}}_{h,\ell}(\delta) 
	:= 
	\widehat{\Theta}_{h,\ell 1} \delta \, n^{-1} \sum_{t=1}^{n} \rho(\widehat{\epsilon}_{1t})
	+ 
	\sum_{j=0}^h \left[ \frac{1}{n-j} \sum_{t=1}^{n-j} \widehat{v}_{j,\ell}\big( X_{t+j:t}; \widehat{\widetilde{\delta}}_t \big) \right] .
\end{equation*}

\begin{lemma}\label{lemma:gamma_v_j_consistent}
	Under the assumptions of Theorem \ref{theorem:consistent_irf} , let $x_{j:0} = (x_j, \ldots, x_0) \in \mathcal{X}^j$ and $\varepsilon \in \mathcal{E}_1$ be nonrandom quantities. Let $\widetilde{\delta}$ be the relaxed shock determined by $\delta$, $\rho$ and $\varepsilon$. Then,
	\begin{itemize}
		\item[(i)] $\sup_{x_{j:0}, \varepsilon} \lvert \widehat{\gamma}_{j}(x_{j:0}; \widetilde{\delta}) - {\gamma}_{j}(x_{j:0}; \widetilde{\delta}) \rvert = o_P(1)$ ,
		\item[(ii)] $\sup_{x_{j:0}, \varepsilon} \lvert \widehat{v}_{j,\ell}\big( x_{j:0}; \widetilde{\delta} \big) - v_{j,\ell}\big( x_{j:0}; \widetilde{\delta} \big) \lvert = o_P(1)$ ,
	\end{itemize}
	for any fixed integers $j \geq 0$ and $\ell \in \{1, \ldots, d\}$.
\end{lemma}

\begin{proof}
	~
	\begin{itemize}
		\item[(i)] We have that 
		%
		\begin{align*}
			\lvert  \widehat{\gamma}_{j}(x_{j:0}; \delta) - {\gamma}_{j}(x_{j:0}; \delta) \rvert 
			& = \left\lvert \sum_{k=1}^j \left[ (\widehat{\Gamma}_{k,11} x_{j-k}(\widetilde{\delta}) - \widehat{\Gamma}_{k,11} x_{j-k} ) - (\Gamma_{k,11} x_{j-k}(\widetilde{\delta}) - \Gamma_{k,11} x_{j-k} ) \right] \right\rvert \\
			& \leq \sum_{k=1}^j \left\lvert \widehat{\Gamma}_{k,11} x_{j-k}(\widetilde{\delta}) - \Gamma_{k,11} x_{j-k}(\widetilde{\delta}) \right\rvert + \sum_{k=1}^j \left\lvert \widehat{\Gamma}_{k,11} x_{j-k} - \Gamma_{k,11} x_{j-k} \right\rvert .
		\end{align*}
		This yields
		$
		\sup_{x_{j:0}, \varepsilon} \: \lvert \widehat{\gamma}_{j}(x_{j:0}; \widetilde{\delta}) - {\gamma}_{j}(x_{j:0}; \widetilde{\delta}) \rvert \leq 2 j \: \sup_{x \in \mathcal{X}} \left\lvert \widehat{\Gamma}_{k,11} x - \Gamma_{k,11} x \right\rvert .
		$
		Since $j$ is finite and fixed and the uniform consistency bound of Lemma \ref{lemma:Gamma_j_consistent} holds, a fortiori $ \sup_{x \in \mathcal{X}} \left\lvert \widehat{\Gamma}_{k,11} x - \Gamma_{k,11} x \right\rvert = o_P(1) $.
		\item[(ii)] Similarly to above, 
		\begin{align*}
			\lvert \widehat{v}_{j,\ell}\big( x_{j:0}; \widetilde{\delta} \big) - {v}_{j,\ell}\big( x_{j:0}; \widetilde{\delta} \big) \lvert 
			& = \left\lvert \left( \widehat{\Gamma}_{j,\ell} \widehat{\gamma}_{j}(x_{j:0}; \widetilde{\delta}) - {\Gamma}_{j,\ell} {\gamma}_{j}(x_{j:0}; \widetilde{\delta}) \right) - \left( \widehat{\Gamma}_{j,\ell} x_{j} - {\Gamma}_{j,\ell} x_{j} \right) \right\rvert \\
			& \leq \lVert \widehat{\Gamma}_{j,\ell} - {\Gamma}_{j,\ell} \rVert_\infty + \lVert {\Gamma}_{j,\ell} \rVert_\infty \lvert \widehat{\gamma}_{j}(x_{j:0}; \delta) - {\gamma}_{j}(x_{j:0}; \delta) \rvert \\
			& \qquad + \lvert \widehat{\Gamma}_{j,\ell} x_{j} - {\Gamma}_{j,\ell} x_{j} \rvert \\
			& \leq 2 \lVert \widehat{\Gamma}_{j,\ell} - {\Gamma}_{j,\ell} \rVert_\infty + C_{\Gamma,j,l} \, \lvert \widehat{\gamma}_{j}(x_{j:0}; \delta) - {\gamma}_{j}(x_{j:0}; \delta) \rvert , \\
		\end{align*}
		where we have used that ${\gamma}_{j}(x_{j:0}; \widetilde{\delta}) \in \mathcal{X}$ to derive the first term in the second line.
		In the last line, $C_{\Gamma,j,l}$ is a constant such that
		$
		\lVert {\Gamma}_{j,\ell} \rVert_\infty \leq \sum_{k=1}^{\min\{j,\, p\}} \lVert \Psi_{j-k} \rVert_\infty \lVert G_{k} \rVert_\infty \leq C_{\Gamma,j,l} .
		$
		The claim then follows thanks to Lemma \ref{lemma:Gamma_j_consistent} and (i).
	\end{itemize}
\end{proof}

In what follows, define $\widehat{v}_{j,\ell}\big( X_{t+j:t}; {\widetilde{\delta}}_t \big)$ to be a version of ${v}_{j,\ell}$ that is constructed using coefficient estimates from $\{ \widehat{\Pi}_1, \widehat{\Pi}_2 \}$ but evaluated on the true innovations $\epsilon_t$.

\begin{proof}[Proof of Theorem \ref{theorem:consistent_irf}]
	If we introduce 
	\begin{equation*}
		\widetilde{\textnormal{IRF}}_{h,\ell}(\delta)^*
		:=
		\widehat{\Theta}_{h,\ell 1} \delta \, n^{-1} \sum_{t=1}^{n} \rho({\epsilon}_{1t})
		+ 
		\sum_{j=0}^h \left[ \frac{1}{n-j} \sum_{t=1}^{n-j} \widehat{v}_{j,\ell}\big( X_{t+j:t}; \widetilde{\delta}_t \big) \right]
		,
	\end{equation*}
	then clearly
	\begin{align*}
		\left\lvert \widehat{\widetilde{\textnormal{IRF}}}_{h,\ell}(\delta) - \widetilde{\textnormal{IRF}}_{h,\ell}(\delta) \right\rvert
		& \leq 
		\left\lvert \widehat{\widetilde{\textnormal{IRF}}}_{h,\ell}(\delta) - \widetilde{\textnormal{IRF}}^*_{h,\ell}(\delta) \right\rvert
		+ 
		\left\lvert \widetilde{\textnormal{IRF}}^*_{h,\ell}(\delta) - \widetilde{\textnormal{IRF}}_{h,\ell}(\delta) \right\rvert \\
		& = I + II .
	\end{align*}

	To control $II$, we can observe 
	\begin{align*}
		II 
		& \leq \left\lvert \widehat{\Theta}_{h,\ell 1} \delta \, n^{-1} \sum_{t=1}^{n} \rho({\epsilon}_{1t}) - {\Theta}_{h,\ell 1} \delta\E[\rho({\epsilon}_{1t})] \right\rvert \\
		& \qquad + \sum_{j=0}^h \left\lvert \frac{1}{n-j} \sum_{t=1}^{n-j} \widehat{v}_{j,\ell}\big( X_{t+j:t}; \widetilde{\delta}_t \big) - \E[{v}_{j,\ell}\big( X_{t+j:t}; \widetilde{\delta} \big)] \right\rvert \\
		& \leq \delta \left\lvert \widehat{\Theta}_{h,\ell 1} - {\Theta}_{h,\ell 1} \right\rvert \left\lvert n^{-1} \sum_{t=1}^{n} \rho({\epsilon}_{1t}) \right\rvert 
		+ \delta \left\lvert \widehat{\Theta}_{h,\ell 1} \right\rvert \left\lvert n^{-1} \sum_{t=1}^{n} \rho({\epsilon}_{1t}) - \E[\rho({\epsilon}_{1t})] \right\rvert \\
		& \qquad + \sum_{j=0}^h \left\lvert \frac{1}{n-j} \sum_{t=1}^{n-j} \widehat{v}_{j,\ell}\big( X_{t+j:t}; \widetilde{\delta}_t \big) - \E[{v}_{j,\ell}\big( X_{t+j:t}; \widetilde{\delta} \big)] \right\rvert \\
		& \leq \delta \left\lvert \widehat{\Theta}_{h,\ell 1} - {\Theta}_{h,\ell 1} \right\rvert \left\lvert n^{-1} \sum_{t=1}^{n} \rho({\epsilon}_{1t}) \right\rvert
		+ \delta \left\lvert \widehat{\Theta}_{h,\ell 1} \right\rvert \left\lvert n^{-1} \sum_{t=1}^{n} \rho({\epsilon}_{1t}) - \E[\rho({\epsilon}_{1t})] \right\rvert \\
		& \qquad + 
		\sum_{j=0}^h \left\lvert \frac{1}{n-j} \sum_{t=1}^{n-j} \widehat{v}_{j,\ell}\big( X_{t+j:t}; \widetilde{\delta}_t \big) - {v}_{j,\ell}\big( X_{t+j:t}; \widetilde{\delta}_t \big) \right\rvert \\
		& \qquad + 
		\sum_{j=0}^h \left\lvert \frac{1}{n-j} \sum_{t=1}^{n-j} {v}_{j,\ell}\big( X_{t+j:t}; \widetilde{\delta}_t \big) - \E[{v}_{j,\ell}\big( X_{t+j:t}; \widetilde{\delta} \big)] \right\rvert
		.
	\end{align*}
	The first two terms in the last bound are $o_P(1)$ since $\left\lvert \widehat{\Theta}_{h,\ell 1} - {\Theta}_{h,\ell 1} \right\rvert = o_P(1)$, as discussed in Lemma \ref{lemma:Gamma_j_consistent}, and $n^{-1} \sum_{t=1}^{n} \rho({\epsilon}_{1t}) \overset{p}{\to} \E[\rho({\epsilon}_{1t})]$ by a WLLN. 
	For the other terms in the last sum above, we similarly note that  
	\begin{equation*}
		\left\lvert \frac{1}{n-j} \sum_{t=1}^{n-j} \widehat{v}_{j,\ell}\big( X_{t+j:t}; \widetilde{\delta}_t \big) - {v}_{j,\ell}\big( X_{t+j:t}; \widetilde{\delta}_t \big) \right\rvert 
		= o_P(1)
	\end{equation*}
	from Lemma \ref{lemma:gamma_v_j_consistent}, while, thanks again to a WLLN, it holds
	\begin{equation*}
		\left\lvert \frac{1}{n-j} \sum_{t=1}^{n-j} {v}_{j,\ell}\big( X_{t+j:t}; \widetilde{\delta}_t \big) - \E[{v}_{j,\ell}\big( X_{t+j:t}; \widetilde{\delta} \big)] \right\rvert 
		= o_P(1) .
	\end{equation*}
	Since $h$ is fixed finite, this implies that $II = o_P(1)$.

	Considering now $I$, we can write 
	\begin{align*}
		I 
		& \leq \delta \left\lvert \widehat{\Theta}_{h,\ell 1} \right\rvert \left\lvert n^{-1} \sum_{t=1}^{n} \rho(\widehat{\epsilon}_{1t}) - \rho({\epsilon}_{1t}) \right\rvert 
		+ \sum_{j=0}^h \left\lvert \frac{1}{n-j} \sum_{t=1}^{n-j} \widehat{v}_{j,\ell}\big( X_{t+j:t}; \widehat{\widetilde{\delta}}_t \big) - \widehat{v}_{j,\ell}\big( X_{t+j:t}; \widetilde{\delta}_t \big) \right\rvert \\
		& = I' + I''
		.
	\end{align*}
	Since by assumption $\rho$ is a bump function, thus continuously differentiable over the range of $\epsilon_t$, by the mean value theorem
	\begin{equation*}
		\left\lvert n^{-1} \sum_{t=1}^{n} \rho(\widehat{\epsilon}_{1t}) - \rho({\epsilon}_{1t}) \right\rvert
		\leq
		n^{-1} \sum_{t=1}^{n} \lvert \rho'_t \rvert \, \big\lvert  \widehat{\epsilon}_{1t} - {\epsilon}_{1t} \big\rvert 
	\end{equation*}
	for a sequence $\{\rho'_t\}_{t=1}^n$ of evaluations of first-order derivative $\rho'$ at values $\overline{\epsilon_t}$ in the interval with endpoint $\epsilon_t$ and $\widehat{\epsilon}_t$. One can use $|\rho'_t| \leq C_{\rho'}$ with a finite positive constant $C_{\rho'}$, and by recalling that $\widehat{\epsilon}_{1t} - {\epsilon}_{1t} = (\Pi_1 - \widehat{\Pi}_1)' W_{1t}$ one, thus, gets
	\begin{equation*}
		\left\lvert n^{-1} \sum_{t=1}^{n} \rho(\widehat{\epsilon}_{1t}) - \rho({\epsilon}_{1t}) \right\rvert
		\leq
		C_{\rho'} \frac{1}{n} \sum_{t=1}^{n} \left\lvert (\Pi_1 - \widehat{\Pi}_1)' W_{1t} \right\rvert
		\leq
		C_{\rho'} \lVert \Pi_1 - \widehat{\Pi}_1 \rVert_2 \: \frac{1}{n} \sum_{t=1}^{n} \lVert W_{1t} \rVert_2 
		= o_P(1).
	\end{equation*}
	This proves that term $I'$ is itself $o_P(1)$. Finally, to control $I''$, we use that by construction estimator $\widehat{\Pi}_2$ is composed of sufficiently regular functional elements i.e. B-spline estimates of order 1 or greater. 
	Thanks again to the mean value theorem
	\begin{align*}
		\left\lvert \frac{1}{n-j} \sum_{t=1}^{n-j} \widehat{v}_{j,\ell}\big( X_{t+j:t}; \widehat{\widetilde{\delta}}_t \big) - \widehat{v}_{j,\ell}\big( X_{t+j:t}; \widetilde{\delta}_t \big) \right\rvert
		& \leq
		\frac{1}{n-j} \sum_{t=1}^{n-j} \left\lvert \widehat{v}_{j,\ell}\big( X_{t+j:t}; \widehat{\widetilde{\delta}}_t \big) - \widehat{v}_{j,\ell}\big( X_{t+j:t}; \widetilde{\delta}_t \big) \right\rvert \\
		& \leq C_{\widehat{v}',j,\ell} \: \frac{1}{n-j} \sum_{t=1}^{n-j} 
		\big\lvert  \widehat{\epsilon}_{1t} - {\epsilon}_{1t} \big\rvert 
	\end{align*}
	for any fixed $j$ and some $C_{\widehat{v}',j,\ell} > 0$.
	This holds since $\widehat{v}_{j,\ell}$ is uniformly continuous by construction. Note that we have assumed that the nonlinear part of $\Pi_2$ belongs to a Hölder class with smoothness $s > 1$ (for simplicity, assume here that $s$ is integer, otherwise a similar argument can be made). Then, even though $C_{\widehat{v}',j,\ell}$ depends on the sample, it is bounded above in probability for $n$ sufficiently large.
	Following the discussion of term $I'$, we deduce that the last line in the display above is $o_p(1)$. As $h$ is finite and independent of $n$, it follows that also $I''$ is of order $o_P(1)$.
	
	Finally, to obtain uniformity with respect to $\delta \in [-\mathcal{D},\mathcal{D}]$, simply note that bounds on $I$ and $II$ are explicit in $\delta$, therefore 
	\begin{equation*}
		\sup_{\delta \in [-\mathcal{D}, \mathcal{D}]} \left\lvert \widehat{\widetilde{\textnormal{IRF}}}_{h,\ell}(\delta) 
		-
		\widetilde{\textnormal{IRF}}_{h,\ell}(\delta)
		\right\rvert
		\leq
		\mathcal{D} \times o_P(1) ,
	\end{equation*}
	concluding the proof.
\end{proof}

\section{Dependence Conditions}\label{appendix:dependence}

\subsection{$\beta$-Mixing}

For the sake of completeness, we first recall the definition of $\beta$-mixing process, which is the dependence framework employed by \cite{chenOptimalUniformConvergence2015}. Let $(\Omega, \mathcal{Q}, \P)$ be the underlying probability space and define
\begin{equation*}
    \beta(\mathcal{A}, \mathcal{B}) := \frac{1}{2} \: \sup \sum_{(i,j) \in I \times J} \lvert  \P(A_i \cap B_i) - \P(A_i) \P(B_i) \rvert
\end{equation*}
where $\mathcal{A}, \mathcal{B}$ are two $\sigma$-algebras, $\{A_i\}_{i \in I} \subset \mathcal{A}$, $\{B_j\}_{j \in J} \subset \mathcal{B}$ and the supremum is taken over all finite partitions of $\Omega$. 
The $h$-th $\beta$-mixing coefficient of process $\{W_{2t}\}_{t \in \Int}$ is defined as
\begin{equation*}
    \beta(h) = \sup_t \beta \big( \sigma(\ldots, W_{2t-1}, W_{2t}), \ldots, \sigma(W_{2t+h}, W_{2t+h+1}, \ldots) \big) ,
\end{equation*}
and $W_{2t}$ is said to be \textit{geometric} or \textit{exponential $\beta$-mixing} if $\beta(h) \leq \gamma_1 \exp(-\gamma_2 h)$ for some $\gamma_1 > 0$ and $\gamma_2 > 0$. An important consideration to be made regarding mixing assumptions is that they are, in general, hard to study. Especially in nonlinear systems, assuming that $\beta(h)$ decays exponentially over $h$ imposes very high-level assumptions on the model. There are, however, many setups, both linear and nonlinear, in which it is known that $\beta$-mixing holds under primitive assumptions, see \cite{chenPenalizedSieve2013} for examples and relevant references.

\subsection{Model-Based Physical Dependence}

Consider a \textit{non-structural model} of the form
\begin{equation}\label{eq_sec3:general_nonlin_model}
    Z_t = G(Z_{t-1}, \epsilon_t) .
\end{equation}
This is a generalization of the semi-reduced model \eqref{eq_main:pseudo_reduced_form_single_eq} where linear and nonlinear components are absorbed into one functional term and $B_0$ is the identity matrix.\footnote{In this subsection, shock identification does not play a role and, as such, one can safely ignore $B_0$.} Indeed, note that models of the form $Z_t = G(Z_{t-1}, \ldots, Z_{t-p}, \epsilon_t)$ can be rewritten as \eqref{eq_sec3:general_nonlin_model} using a companion formulation.
If $\epsilon_t$ is stochastic, \eqref{eq_sec3:general_nonlin_model} defines a causal nonlinear stochastic process. More generally, it defines a nonlinear difference equation and an associated dynamical system driven by $\epsilon_t$. 
Throughout this subsection, we shall assume that $Z_t \in \mathcal{Z} \subseteq \Real^{d_Z}$ as well as $\epsilon_t \in \mathcal{E} \subseteq \Real^{d_Z}$.

Relying on the framework of \cite{potscher1997dynamic}, we now introduce explicit conditions that allow us to control dependence in nonlinear models by using the toolbox of physical dependence measures developed by \cite{wuNonlinearSystemTheory2005,wuAsymptoticTheoryStationary2011a}.
The aim is to use a dynamical system perspective to impose meaningful assumptions on nonlinear dynamic models. This makes it possible to give more primitive conditions under which one can estimate \eqref{eq_main:regression_model} in a semiparametric way.

\subsubsection{Stability}\label{appendix_subsec:stability}

An important concept for dynamical system theory is that of stability. Stability plays a key role in constructing a valid asymptotic theory, as it is well understood in linear models. It is also fundamental for developing the approximation theory of nonlinear stochastic systems. 

\begin{example}
	As a motivating example, first consider the linear system
	$Z_t = B Z_{t-1} + \epsilon_t$,
	where we may assume that $\{\epsilon_t\}_{t \in \Int}$, $\epsilon_t \in \Real^{d_Z}$, is a sequence of i.i.d. innovations.\footnote{One could alternatively think of the case of a deterministic input, setting $\epsilon_t \sim P_t(a_t)$, where $P_t(a_t)$ is a Dirac density on the deterministic sequence $\{a_t\}_{t \in \Int}$.} It is well-known that this system is stable if and only if the largest eigenvalue of $B$ is strictly less than one in absolute value \citep{lutkepohlNewIntroductionMultiple2005}. For a higher order linear system, $Z_t = B(L) Z_{t-1} + \epsilon_t$ where $B(L) = B_1 + B_2 L + \ldots + B_p L^{p-1}$, stability holds if and only if $| \lambda_{\max}(\mathbf{B}) | < 1$ with 
	\begin{equation*}
	   \textbf{B} := 
	   \begin{bmatrix}
	   B_1 & B_2 & \cdots & B_p \\
	   I_{d_Z} & 0 & \cdots & 0 \\
	   0 & I_{d_Z} & \cdots & 0 \\
	   \vdots & \vdots & \cdots & \vdots \\
	   0 & \cdots & I_{d_Z} & 0
	   \end{bmatrix} 
	\end{equation*}
	being the companion matrix associated with $B(L)$. 
\end{example}

Extending the notion of stability from linear to nonlinear systems requires some care. \cite{potscher1997dynamic} derived generic conditions allowing to formally extend stability to nonlinear models by first analyzing \textit{contractive} systems.

\begin{definition}[Contractive System]\label{definition:constractive_system}
    Let $Z_t \in \mathcal{Z} \subseteq \Real^{d_Z}$, $\epsilon_t \in \mathcal{E} \subseteq \Real^{d_Z}$, where $\{Z_t\}_{t \in \Int}$ is generated according to $Z_t = G(Z_{t-1}, \epsilon_t)$.
    The system is {contractive} if for all $(z, z') \in \mathcal{Z} \times \mathcal{Z}$ and $(e, e') \in \mathcal{E} \times \mathcal{E}$ the condition $\lVert G(z, \epsilon) - G(z', \epsilon') \rVert \leq C_Z \lVert z - z' \rVert + C_\epsilon \lVert e - e' \rVert$
    holds with Lipschitz constants $0 \leq C_Z < 1$ and $0 \leq C_\epsilon < \infty$.
\end{definition}

Sufficient conditions to establish contractivity are
\begin{equation}\label{eq_sec3:contract_y}
    \text{sup} \left\{
        \left\lVert\, \text{stack}_{i=1}^{d_Z} \, \left[ \frac{\partial G}{\partial Z} (z^i, e^i)\right]_i \,\right\rVert \,\bigg\vert\, z^i \in \mathcal{Z}, e^i \in \mathcal{E}
    \right\}
    < 1
\end{equation}
and 
\begin{equation}\label{eq_sec3:contract_eps}
    \left\lVert \frac{\partial G}{\partial \epsilon} \right\lVert < \infty ,
\end{equation}
where the stacking operator $\text{stack}_{i=1}^{d_Z} [\,\cdot\,]_i$ progressively stacks the rows, indexed by $i$, of its argument (which can be changing with $i$) into a matrix. Values $(z^i, e^i) \in \mathcal{Z} \times \mathcal{E}$ change with index $i$ as the above condition is derived using the mean value theorem. Therefore, it is necessary to consider a different set of values for each component of $Z_t$.

It is easy to see, as \cite{potscher1997dynamic} point out, that contractivity is often too strong a condition to be imposed. Indeed, even in the simple case of a scalar AR(2) model $Z_t = b_1 Z_{t-1} + b_2 Z_{t-2} + \epsilon_t$, regardless of the values of $b_1, b_2 \in \Real$ contractivity is violated. This is due to the fact that in a linear AR(2) model studying contractivity reduces to checking $\lVert \mathbf{B} \rVert  < 1$ instead of $| \lambda_{\max}(\mathbf{B}) | < 1$, and the former is a stronger condition than the latter.\footnote{See \cite{potscher1997dynamic}, pp.68-69.} One can weaken contractivity -- which must hold for $G$ as a map from $Z_{t-1}$ to $Z_t$ -- to the idea of \textit{eventual contractivity}. That is, intuitively, one can impose conditions on the dependence of $Z_{t+h}$ on $Z_t$ for $h > 1$ sufficiently large. To do this formally, we must first introduce the definition of system map iterates. 

\begin{definition}[System Map Iterates]
    Let $Z_t \in \mathcal{Z} \subseteq \Real^{d_Z}$, $\epsilon_t \in \mathcal{E} \subseteq \Real^{d_Z}$, where $\{Z_t\}_{t \in \Int}$ is generated from a sequence $\{\epsilon_t\}_{t \in \Int}$ according to $Z_t = G(Z_{t-1}, \epsilon_t)$.
    The $h$-order system map iterate is defined to be
    \begin{align*}
        G^{(h)}(Z_t, \epsilon_{t+1}, \epsilon_{t+2}, \ldots, \epsilon_{t+h})
        & := G(G(\cdots G(Z_t, \epsilon_{t+1}) \cdots, \epsilon_{t+h-1}), \epsilon_{t+h}) \\
        & ~= G(\cdot, \epsilon_{t+h}) \circ G(\cdot, \epsilon_{t+h-1}) \circ \cdots \circ G(Z_t, \epsilon_{t+1}) ,
    \end{align*}
    where $\circ$ signifies function composition and $G^{(0)}(Z_t) = Z_t$. 
\end{definition}

To shorten notation, in place of $G^{(h)}(Z_t, \epsilon_{t+1}, \epsilon_{t+2}, \ldots, \epsilon_{t+h})$ we shall use $G^{(h)}(Z_t, \epsilon_{t+1:t+h})$. Additionally, for $1 \leq j \leq h$, the partial derivative $\partial G^{(h^*)} / {\partial \epsilon_j}$ for some fixed $h^*$ is to be intended with respect to $\epsilon_{t+j}$, the $j$-th entry of the input sequence. This derivative does not depend on the time index since by assumption $G$ is time-invariant and so is $G^{(h)}$.

Taking again the linear autoregressive model as an example, 
\begin{equation*}
    Z_{t+h} = G^{(h)}(Z_t, \epsilon_{t+1:t+h}) = B_1^h Z_t + \sum_{i=0}^{h-1} B_1^{i} \epsilon_{t+h-i}
\end{equation*}
since $G(z, \epsilon) = B_1 z + \epsilon$. If $B_1$ determines a stable system, then $\lVert B_1^h \rVert \to 0$ as $h \to \infty$ since $G^h$ converges to zero, and therefore $\lVert B_1^h \rVert \leq C_Z < 1$ for $h$ sufficiently large. It is thus possible to use system map iterates to define stability for higher-order nonlinear systems.

\begin{definition}[Stable System]\label{definition:stable_system}
    Let $Z_t \in \mathcal{Z} \subseteq \Real^{d_Z}$, $\epsilon_t \in \mathcal{E} \subseteq \Real^{d_Z}$, where $\{Z_t\}_{t \in \Int}$ is generated according to the system $Z_t = G(Z_{t-1}, \epsilon_t)$.
    The system is {stable} if there exists $h^* \geq 1$ such that
    for all $(z, z') \in \mathcal{Z} \times \mathcal{Z}$ and $(e_1, e_2, \ldots e_{h^*},$ $ e'_1, e'_2, \ldots, e'_{h^*}) \in \bigtimes_{i=1}^{2 h^*} \mathcal{E}$ 
    \begin{equation*}
        \lVert G^{(h^*)}(z, e_{1:h^*}) - G^{(h^*)}(z', e'_{1:h^*}) \rVert \leq C_Z \lVert z - z' \rVert + C_\epsilon \lVert e_{1:h^*} - e'_{1:h^*} \rVert
    \end{equation*}
    holds with Lipschitz constants $0 \leq C_Z < 1$ and $0 \leq C_\epsilon < \infty$.
\end{definition}

It is important to remember that this definition encompasses systems with an arbitrary finite autoregressive structure, i.e., $Z_t = G(Z_{t-p+1}, \ldots, Z_{t-1}, \epsilon_t)$ for $p \geq 1$, thanks to the companion formulation of the process. An explicit stability condition, similar to that discussed above for contractivity, can be derived using the mean value theorem. Indeed, for a system to be stable, it is sufficient that, at iterate $h^*$,
\begin{equation}\label{eq_sec3:stable_y}
    \text{sup} \left\{
        \left\lVert\, \text{stack}_{i=1}^{d_Z} \, \left[ \frac{\partial G^{(h^*)}}{\partial Z} (z^i, e^i_{1:h^*})\right]_i \,\right\rVert \,\bigg\vert\, z^i \in \mathcal{Z}, e^i_{1:h^*} \in \bigtimes_{i=1}^{h^*} \mathcal{E}
    \right\}
    < 1
\end{equation}
and
\begin{equation}\label{eq_sec3:stable_eps}
    \text{sup} \left\{
        \left\lVert\, \frac{\partial G^{(h^*)}}{\partial \epsilon_j} (z, e_{1:h^*}) \,\right\rVert \,\bigg\vert\, z \in \mathcal{Z}, e_{1:h^*} \in \bigtimes_{i=1}^{h^*} \mathcal{E}
    \right\}
    < \infty ,
    \qquad j = 1, \ldots, h^* .
\end{equation}

\begin{remark}
	\cite{potscher1997dynamic} have used conditions \eqref{eq_sec3:contract_y}-\eqref{eq_sec3:contract_eps} and \eqref{eq_sec3:stable_y}-\eqref{eq_sec3:stable_eps} as the basis for uniform laws of large numbers and central limit theorems for $L^r$-approximable and near epoch dependent processes.
\end{remark}

\subsubsection{Physical Dependence}\label{appendix_subsec:physical_dep}

\cite{wuNonlinearSystemTheory2005} first proposed alternatives to mixing concepts by proposing dependence measures rooted in a dynamical system view of a stochastic process. Much work has been done to use such measures to derive approximation results and estimator properties; see, for example, \cite{wuKernelEstimationTime2010,wuAsymptoticTheoryStationary2011a,chenSelfnormalizedCramertypeModerate2016}, and references within.

\begin{definition}\label{def:physical_dependence_measure}
	Let $\{Z_t\}_{t \in \Int}$ be a process that can be written as $Z_{t+h} = G^{(h)}(Z_t, \epsilon_{t+1:t+h})$ for all $h \geq 1$, (nonlinear) maps $G^{(h)}$ and innovations $\{\epsilon_t\}_{t \in \Int}$. If for all $t \in \Int$ and chosen $r \geq 1$, $Z_t$ has finite $r$th moment, the {functional physical dependence measure} $\Delta_r$ is 
	\begin{equation*}
		\Delta_r(h) := \sup_t \left\lVert\, Z_{t+h} - G^{(h)}(Z'_t, \epsilon_{t+1:t+h}) \,\right\rVert_{L^r}
	\end{equation*}
	where $\{Z'_{t}\}_{t \in \Int}$ is an independent copy of $\{Z_{t}\}_{t \in \Int}$ based on innovation process $\{\epsilon'_{t}\}_{t \in \Int}$, itself an independent copy of $\{\epsilon_{t}\}_{t \in \Int}$.
\end{definition}

\cite{chenSelfnormalizedCramertypeModerate2016}, among others, show how one may replace the geometric $\beta$-mixing assumption with a physical dependence assumption.\footnote{We adapt here the definitions of \cite{chenSelfnormalizedCramertypeModerate2016} to work with a system of the form $Z_t = G(Z_{t-1}, \epsilon_t)$.}
We will consider the setting where models have dependence -- as measured by $\Delta_r(h)$ -- which decays exponentially with $h$.

\begin{definition}[Geometric Moment Contracting Process]\label{def:gmc_process}
	$\{Z_t\}_{t \in \Int}$ is {geometric moment contracting} (GMC) in $L^r$ norm if there exists $a_1 > 0$, $a_2 > 0$ and $\tau \in (0,1]$ such that
	\begin{equation*}
		\Delta_r(h) \leq a_1 \exp(- a_2 \, h^\tau) .
	\end{equation*}
\end{definition}

GMC conditions can be considered more general than $\beta$-mixing, as they encompass well-known counterexamples, e.g., the known counterexample provided by $Z_t = (Z_{t-1} + \epsilon_t)/2$ for $\epsilon_t$ i.i.d. Bernoulli r.v.s \citep{chenSelfnormalizedCramertypeModerate2016}. 

In the following proposition, we prove that, if contractivity or stability conditions as defined by \cite{potscher1997dynamic} hold for $G$ and $\{\epsilon_t\}_{t \in \Int}$ is an i.i.d. sequence, then the process $\{Z_t\}_{t \in \Int}$ is GMC -- according to Definition~\ref{def:gmc_process} -- under weak moment assumptions.
 
\begin{proposition}\label{prop:gmc_conditions_on_map}
    Assume that $\{\epsilon_t\}_{t \in \Int}$, $\epsilon_t \in \mathcal{E} \subseteq \Real^{d_Z}$ are i.i.d. and $\{Z_t\}_{t \in \Int}$ is generated according to $Z_t = G(Z_{t-1}, \epsilon_t)$, where $Z_t \in \mathcal{Z} \subseteq \Real^{d_Z}$ and $G$ is a measurable function. 
    \begin{itemize}
        \item[(a)] If contractivity conditions \eqref{eq_sec3:contract_y}-\eqref{eq_sec3:contract_eps} hold, $\sup_{t \in \Int} \lVert \epsilon_t \rVert_{L^r} < \infty$ for $r \geq 2$ and $\lVert G(\overline{z}, \overline{\epsilon}) \rVert < \infty$ for some $(\overline{z}, \overline{\epsilon}) \in \mathcal{Z} \times \mathcal{E}$, then $\{Z_t\}_{t \in \Int}$ is GMC with
        $
            \Delta_r(k) \leq a \exp(- \gamma h)
        $
        where $\gamma = - \log(C_Z)$ and $a = 2 \lVert Z_t \rVert_{L^r} < \infty$.
        \item[(b)] If stability conditions \eqref{eq_sec3:stable_y}-\eqref{eq_sec3:stable_eps} hold, $\sup_{t \in \Int} \lVert \epsilon_t \rVert_{L^r} < \infty$ for $r \geq 2$ and $\lVert \partial G / \partial Z \rVert \leq M_Z < \infty$, then $\{Z_t\}_{t \in \Int}$ is GMC with
        $
            \Delta_r(k) \leq \bar{a} \exp(- \gamma_{h^*}\, h)
        $
        where $\gamma_{h^*} = - \log(C_Z) / h^* $ and $\bar{a} = 2 \lVert Z_t \rVert_{L^r} \max\{M_Z^{h-1}, 1\} / C_Z < \infty$.
    \end{itemize}
\end{proposition}

Proposition \ref{prop:gmc_conditions_on_map} is important in that it links the GMC property to transparent conditions on the structure of the nonlinear model. It also allows to handle multivariate systems, while previous work has focused on scalar systems (cf. \cite{wuAsymptoticTheoryStationary2011a} and \cite{chenSelfnormalizedCramertypeModerate2016}).

Lastly, the following lemma shows that if $\{W_{2t}\}_{t\in\Int}$ is geometric moment contracting, Assumption \ref{assumption:series_gram_matrix_convergence} is fulfilled.\footnote{Compare also with Lemma 2.2 in \cite{chenOptimalUniformConvergence2015}.} 

\begin{lemma}\label{lemma:gmc_sample_gram_matrix_op1}
	If Assumption \ref{assumption:sieve_regularity}(iii) holds and $\{W_{2t}\}_{t \in \Int}$ is strictly stationary and GMC then one may choose an integer sequence $q = q(n) \leq n/2$ with $(n/q)^{r+1} q K^\rho \Delta_r(q) = o(1)$ for $\rho = 5/2 - (r/2 + 2/r) + \omega_2$ and $r > 2$ such that
	\begin{equation*}
		\lVert (\widetilde{B}_\pi' \widetilde{B}_\pi / n) - I_K \rVert = O_P\left( \zeta_{K,n} \lambda_{K,n} \sqrt{\frac{q \log K}{n}} \right) = o_P(1)
	\end{equation*}
	provided $\zeta_{K,n} \lambda_{K,n} \sqrt{(q \log K)/{n}} = o(1)$.
\end{lemma}

It can be seen that Lemma \ref{lemma:gmc_sample_gram_matrix_op1} holds by setting $\sqrt{K (\log(n))^2 / n} = o(1)$ and choosing $ q(n) = \gamma^{-1} \log(K^\rho n^{r+1}) $, where $\gamma$ is a GMC factor, see Proposition~\ref{prop:gmc_conditions_on_map} in the Online Appendix for details.  
Therefore, the rate is the same as the one derived by \cite{chenOptimalUniformConvergence2015} for exponentially $\beta$-mixing regressors. 
It is straightforward to prove that, if one assume $\{Z_t\}_{t \in \Int}$ fulfills GMC conditions, then $\{W_{2t}\}_{t \in \Int}$ is also a geometric moment contracting process, see Remark~\ref{remark:W2_regs_gmc} below. 
Accordingly, Lemma \ref{lemma:gmc_sample_gram_matrix_op1} applies and Assumption \ref{assumption:series_gram_matrix_convergence} is automatically verified.

\subsection{Supporting Proofs}

\paragraph{Matrix Norms.}
Let
\begin{equation*}
	\lVert A \rVert_r := \max \big\{ \lVert A x \rVert_r \,\big\vert\, \lVert x \rVert_r \leq 1 \big\}
\end{equation*}
be the $r$-operator norm of matrix $A \in \Complex^{d_1 \times d_2}$. The following Theorem establishes the equivalence between different operator norms as well as the compatibility constants.

\begin{theorem}[\cite{fengEquivalenceConstantsCertain2003}]
	Let $1 \leq p, q \leq \infty$. Then for all $A \in \Complex^{d_1 \times d_2}$,
	\begin{equation*}
		\lVert A \rVert_p \leq \lambda_{p,q}(d_1) \lambda_{q,p}(d_2) \lVert A \rVert_q 
		\quad\textnormal{where}\quad
		\lambda_{a,b}(d) := \begin{cases}
			1 & \text{if } a \geq b , \\
			d^{1/a-1/b} & \text{if } a < b .  \\
		\end{cases}
	\end{equation*}
	This norm inequality is sharp.
\end{theorem}

In particular, if $p > q$ then it holds 
$
(d_2)^{-(1/q - 1/p)} \lVert A \rVert_p 
\leq \lVert A \rVert_q
\leq (d_1)^{1/q - 1/p} \lVert A \rVert_p .
$

\subsubsection{GMC Conditions and Proposition \ref{prop:gmc_conditions_on_map}}

\begin{lemma}\label{lemma:finite_Lr_norm_process}
    Assume that $\{\epsilon_t\}_{t \in \Int}$, $\epsilon_t \in \mathcal{E} \subseteq \Real^{d_Z}$ are i.i.d., and $\{Z_t\}_{t \in \Int}$ is generated according to $Z_t = G(Z_{t-1}, \epsilon_t)$, where $Z_t \in \mathcal{Z} \subseteq \Real^{d_Z}$ and $G$ is a measurable function. If either
    \begin{itemize}
        \item[(a)] Contractivity conditions \eqref{eq_sec3:contract_y}-\eqref{eq_sec3:contract_eps} hold, $\sup_{t \in \Int} \lVert \epsilon_t \rVert_{L^r} < \infty$ and $\lVert G(\overline{z}, \overline{\epsilon}) \rVert < \infty$ for some $(\overline{z}, \overline{\epsilon}) \in \mathcal{Z} \times \mathcal{E}$;
        \item[(b)] Stability conditions \eqref{eq_sec3:stable_y}-\eqref{eq_sec3:stable_eps} hold, $\sup_{t \in \Int} \lVert \epsilon_t \rVert_{L^r} < \infty$ and $\lVert \partial G / \partial Z \rVert \leq M_Z < \infty$;
    \end{itemize}
    then $\sup_t \lVert Z_t \rVert_{L^r} < \infty \quad \text{w.p.} 1$.
\end{lemma}

\begin{proof}
~\\[-2em]
\begin{itemize}
    \item[(a)] In a first step, we show that, given event $\omega \in \Omega$, realization $Z_t(\omega)$ is unique with probability one. To do this, introduce initial condition $z_\circ$ for $\ell > 1$ such that $z_\circ \in \mathcal{Z}$ and $\lVert z_\circ \rVert < \infty$. Define
    $
        Z^{(-\ell)}_{t}(\omega) = G^{(\ell)}(y_\circ, \epsilon_{t-\ell+1:t}(\omega))
    $.
    Further, let $Z'^{(-\ell)}_{t}$ be the realization with initial condition $z'_\circ \not= z_\circ$ and innovation realizations $\epsilon_{t-\ell+1:t}(\omega)$. Note that
    $
        \lVert Z^{(-\ell)}_{t}(\omega) - Z'^{(-\ell)}_{t}(\omega) \rVert 
        \leq C_Z^\ell \left\lVert z_\circ - z'_\circ \right\rVert
    $,
    which goes to zero as $\ell \to \infty$. Therefore, if we set $Z_t(\omega) := \lim_{\ell \to \infty} Z^{(-\ell)}_{t}(\omega)$, $Z_t(\omega)$ is unique with respect to the choice of $z_\circ$ w.p.1. 
    A similar recursion shows that 
    \begin{align*}
        \left\lVert Z^{(-\ell)}_{t}(\omega) \right\rVert 
        & \leq C_Z^\ell \left\lVert z_\circ \right\rVert + \sum_{k=0}^{\ell - 1} C_Z^k C_\epsilon \left\lVert \epsilon_{t-k}(\omega) \right\rVert .
    \end{align*}
    By norm equivalence, this implies 
    \begin{align*}
        \left\lVert Z^{(-\ell)}_{t} \right\rVert_{L^r} 
        & \leq C_Z^\ell \left\lVert z_\circ \right\rVert_r + \sum_{k=0}^{\ell - 1} C_Z^k C_\epsilon \left\lVert \epsilon_{t-k} \right\rVert_{L^r} 
        \leq C_Z^\ell \left\lVert z_\circ \right\rVert_r + \frac{C_\epsilon}{1 - C_Z}\, \sup_{t \in \Int} \lVert \epsilon_t \rVert_{L^r} < \infty ,
    \end{align*}
    and taking the limit $\ell \to \infty$ proves the claim.
    \item[(b)] Consider again distinct initial conditions $z'_\circ \not= z_\circ$ and innovation realizations $\epsilon_{t-\ell+1:t}(\omega)$, yielding $Z'^{(-\ell)}_{t}(\omega)$ and $Z^{(-\ell)}_{t}(\omega)$, respectively. 
    We may use the contraction bound derived in the proof of Proposition \ref{prop:gmc_conditions_on_map} (b) below, that is, 
    $
        \lVert\, Z^{(-\ell)}_{t}(\omega) - Z'^{(-\ell)}_{t}(\omega) \,\rVert_r \leq C_Z^{\ell} C_2 \lVert z_\circ - z'_\circ \rVert_r 
    $,
    where $C_2 > 0$ is a constant. With trivial adjustments, the uniqueness and limit arguments used for (a) above apply here too.

\end{itemize}
\end{proof}

\begin{proof}[Proof of Proposition \ref{prop:gmc_conditions_on_map}]
~\\[-2em]
\begin{itemize}
    \item[(a)] By assumption for all $(z, z') \in \mathcal{Z} \times \mathcal{Z}$ and $(e, e') \in \mathcal{E} \times \mathcal{E}$ it holds that
    $
        \lVert G(z, \epsilon) - G(z', \epsilon') \rVert \leq C_Z \lVert z - z' \rVert + C_\epsilon \lVert e - e' \rVert
    $,
    where $0 \leq C_Z < 1$ and $0 \leq C_\epsilon < \infty$. The equivalence of norms directly generalizes this inequality to any $r$-norm for $r > 2$. We study $\lVert Z_{t+h} - Z'_{t+h} \rVert_r$ where $Z'_{t+h}$ is constructed with a time-$t$ perturbation of the history of $Z_{t+h}$.
    Therefore, for any given $t$ and $h \leq 1$ it holds that
    \begin{align*}
        \left\lVert\, Z_{t+h} - G^{(h)}(Z'_t, \epsilon_{t+1:t+h}) \,\right\rVert_r & \leq C_Z \lVert G^{(h-1)}(Z_t, \epsilon_{t+1:t+h-1}) - G^{(h-1)}(Z'_t, \epsilon_{t+1:t+h-1}) \rVert_r \\
        & \leq C_Z^h \lVert Z_t - Z'_t \rVert_r ,
    \end{align*}
    since sequence $\epsilon_{t+1:t+h}$ is common between $Z_{t+h}$ and $Z'_{t+h}$. Clearly then
    \begin{equation*}
        \left\lVert\, Z_{t+h} - G^{(h)}(Z'_t, \epsilon_{t+1:t+h}) \,\right\rVert_r \leq 2 \lVert Z_t \rVert_r \exp( - \gamma h)
    \end{equation*}
    for $\gamma = - \log(C_Z)$. 
    Letting $a = 2 \lVert Z_t \rVert_r$ and shifting time index $t$ backward by $h$, since $\sup_t \lVert Z_t \rVert_{L^r} < \infty$ w.p.1 from Lemma \ref{lemma:finite_Lr_norm_process} the result for $L^r$ follows with $\tau = 1$.
    \item[(b)] Proceed similar to (a), but notice that now we must handle cases of steps $1 \leq h < h^*$. Consider iterate $h^* + 1$, for which
    \begin{align*}
        & \left\lVert\, Z_{t+h+1} - G^{(h+1)}(Z'_t, \epsilon_{t+1:t+h+1}) \,\right\rVert_r \\
        &\qquad\qquad \leq C_Z \lVert G^{(h)}(G(Z_t, \epsilon_{t+1}), \epsilon_{t+2:t+h}) - G^{(h)}(G(Z'_t, \epsilon_{t+1}), \epsilon_{t+2:t+h}) \rVert_r \\
        &\qquad\qquad \leq C_Z^h \lVert G(Z_t, \epsilon_{t+1}) - G(Z'_t, \epsilon_{t+1}) \rVert_r \\
        &\qquad\qquad \leq C_Z^h M_Z \lVert Z_t - Z'_t \rVert_r
    \end{align*}
    by the mean value theorem. Here we may assume that $M_Z \geq 1$ otherwise we would fall under case (a), so that $M_Z \leq M_Z^2 \leq \ldots \leq M_Z^{h^*-1}$. More generally,
    \begin{equation*}
        \left\lVert\, Z_{t+h+1} - G^{(h+1)}(Z'_t, \epsilon_{t+1:t+h+1}) \,\right\rVert_r \leq C_Z^{j(h)} \max\{M_Z^{h^*-1}, 1\} \lVert Z_t - Z'_t \rVert_r
    \end{equation*}
    for $j(h) := \lfloor h / h^* \rfloor$. Result (b) then follows by noting that $j(h) \geq h / h^* - 1$ and then proceeding as in (a) to derive GMC coefficients.
\end{itemize}
\end{proof}

\begin{remark}\label{remark:W2_regs_gmc}
    The assumption of GMC for a process translates naturally to vectors that are composed of stacked lags of realizations. This, for example, is important in the discussion of Section \ref{section:estimation}, since one needs that regressors $\{W_{2t}\}_{t \in\Int}$ be geometric moment contracting.

    Recall that $W_{2t} = (X_t, X_{t-1}, \ldots, X_{t-p}, Y_{t-1}, \ldots, Y_{t-p}, \epsilon_{1t})$. Here we shall reorder this vector slightly to actually be
    $
        W_{2t} = (X_t, X_{t-1}, Y_{t-1}, \ldots, X_{t-p}, Y_{t-p}, \epsilon_{1t}) .
    $
    For $h > 0$ and $1 \leq l \leq h$, let $Z_{t+j}' := \Phi^{(l)}(Z'_t, \ldots, Z'_{t-p}; \epsilon_{t+1:t+j})$ be the a perturbed version of $Z_t$, where $Z'_t, \ldots, Z'_{t-p}$ are taken from an independent copy of $\{Z_t\}_{t \in \Int}$. Define 
    $
        W'_{2t} = (X'_t, X'_{t-1}, Y'_{t-1}, \ldots, X'_{t-p}, Y'_{t-p}, \epsilon_{1t}) .
    $
    Using Minkowski's inequality 
    \begin{align*}
        \lVert W_{2t+h} - W'_{2t+h} \rVert_{L^r} 
        & \leq \lVert X_{t+h} - X'_{t+h} \rVert_{L^r} + \sum_{j=1}^p \lVert Z_{t+h-j} - Z'_{t+h-j} \rVert_{L^r} \\
        & \leq \sum_{j=0}^p \lVert Z_{t+h-j} - Z'_{t+h-j} \rVert_{L^r} ,
    \end{align*}
    thus, since $p > 0$ is fixed finite,
    \begin{equation*}
        \sup_t \lVert W_{2t+h} - W'_{2t+h} \rVert_{L^r} 
        \leq \sum_{j=0}^p \Delta_r(h-j)
        \leq (p + 1) \, a_{1 Z} \exp(-a_{2 Z} h) .
    \end{equation*}
    Above, $a_{1 Z}$ and $a_{2 Z}$ are the GMC coefficients of $\{Z_t\}_{t \in \Int}$.
\end{remark}

\subsubsection{Lemma \ref{lemma:gmc_sample_gram_matrix_op1} and Matrix Inequalities under Dependence}

In order to prove Lemma \ref{lemma:gmc_sample_gram_matrix_op1}, we modify the approach of \cite{chenOptimalUniformConvergence2015}, which relies on Berbee's Lemma and an interlaced coupling, to handle variables with physical dependence. This is somewhat similar to the proof strategies used in \cite{chenSelfnormalizedCramertypeModerate2016}.

\paragraph*{}
First of all, we recall below a Bernstein-type inequality for {independent} random matrices of \cite{troppUserFriendlyTailBounds2012}.

\begin{theorem}
	Let $\{\Xi_{i}\}_{i = 1}^n$ be a finite sequence of independent random matrices with dimensions $d_1 \times d_2$. Assume $\E[\Xi_i] = 0$ for each $i$ and $\max_{1 \leq i \leq n} \lVert \Xi_i \rVert \leq R_n$ and define
	$$ \varsigma^2_n := \max\left\{ \left\lVert \sum_{i=1}^n \E\left[ \Xi_{i,n} \Xi_{j,n}' \right] \right\rVert, \left\lVert \sum_{i=1}^n \E\left[ \Xi_{i,n}' \Xi_{j,n} \right] \right\rVert \right\} . $$
	Then for all $z \geq 0$,
	\begin{equation*}
		\P\left( \left\lVert \sum_{i=1}^n \Xi_{i} \right\rVert \geq z \right) 
		\leq 
		(d_1 + d_2) \exp\left( \frac{-z^2/2}{n q \varsigma^2_n + q R_n z/3} \right) .
	\end{equation*}
\end{theorem}

The main exponential matrix inequality due to \cite{chenOptimalUniformConvergence2015}, Theorem 4.2 is as follows.

\begin{theorem}\label{theorem:cc15_exp_matrix_ineq}
	Let $\{X_i\}_{i \in \Int}$ where $X_i \in \mathcal{X}$ be a $\beta$-mixing sequence and let $\Xi_{i,n} = \Xi_n(X_i)$ for each $i$ where $\Xi_n : \mathcal{X} \to \Real^{d_1 \times d_2}$ be a sequence of measurable $d_1 \times d_2$ matrix-valued functions. Assume that $\E[\Xi_{i,n}] = 0$ and $\lVert \Xi_{i,n} \rVert \leq R_n$ for each $i$ and define 
	$$ S^2_n := \max\left\{ \E\left[ \lVert \Xi_{i,n} \Xi_{j,n}' \rVert \right], \E\left[ \lVert \Xi_{i,n}' \Xi_{j,n} \rVert \right] \right\} . $$
	Let $1 \leq q \leq n/2$ be an integer and let $I_\bullet = q\lfloor n/q \rfloor, \ldots, n$ when $q \lfloor n/q \rfloor < n$ and $I_\bullet = \emptyset$ otherwise. 
	Then, for all $z \geq 0$,
	\begin{equation*}
		\P\left( \left\lVert \sum_{i=1}^n \Xi_{i,n} \right\rVert \geq 6 z \right) 
		\leq 
		\frac{n}{q} \beta(q) 
		+ \P\left( \left\lVert \sum_{i \in I_\bullet} \Xi_{i,n} \right\rVert \geq z \right) 
		+ 2 (d_1 + d_2) \exp\left( \frac{-z^2/2}{n q S^2_n + q R_n z/3} \right) ,
	\end{equation*}
	where $\lVert \sum_{i \in I_\bullet} \Xi_{i,n} \rVert := 0$ whenever $I_\bullet = \emptyset$.
\end{theorem}

To fully extend Theorem \ref{theorem:cc15_exp_matrix_ineq} to physical dependence, we will proceed in steps. First, we derive a similar matrix inequality by directly assuming that random matrices $\Xi_{i,n}$ have physical dependence coefficient $\Delta^{\Xi}_r(h)$.
In the derivations we will use that
\begin{equation*}
	\frac{1}{(d_2)^{1/2 - 1/r}} \lVert A \rVert_r 
	\leq \lVert A \rVert_2
	\leq (d_1)^{1/2 - 1/r} \lVert A \rVert_r .
\end{equation*}
for $r \geq 2$.

\begin{theorem}\label{theorem:phys_dep_exp_matrix_ineq}
	Let $\{\epsilon_j\}_{j \in \Int}$ be a sequence of i.i.d. variables and let $\{\Xi_{i,n}\}_{i = 1}^n$,
	$$ \Xi_{i,n} = G^{\Xi}_n(\ldots, \epsilon_{i-1}, \epsilon_i) $$ 
	for each $i$, where $\Xi_n : \mathcal{X} \to \Real^{d_1 \times d_2}$, be a sequence of measurable $d_1 \times d_2$ matrix-valued functions.
	Assume that $\E[\Xi_{i,n}] = 0$ and $\lVert \Xi_{i,n} \rVert \leq R_n$ for each $i$ and define 
	$$ S^2_n := \max\left\{ \E\left[ \lVert \Xi_{i,n} \Xi_{j,n}' \rVert \right], \E\left[ \lVert \Xi_{i,n}' \Xi_{j,n} \rVert \right] \right\} . $$
	Additionally assume that $\lVert \Xi_{i,n} \rVert_{L^r} < \infty$ for $r > 2$ and define the matrix physical dependence measure $\Delta^{\Xi}_r(h)$ as
	\begin{equation*}
		\Delta^{\Xi}_r(h) := \max_{ 1 \leq i \leq n } \left\lVert\, \Xi_{i,n} - \Xi^{h*}_{i,n} \,\right\rVert_{L^r} ,
	\end{equation*}
	where $\Xi^{h*}_{i,n} := G^{\Xi}_n(\ldots, \epsilon^*_{i-h-1}, \epsilon^*_{i-h}, \epsilon_{i-h+1}, \ldots, \epsilon_{i-1}, \epsilon_i)$ for independent copy $\{\epsilon^*_j\}_{j \in \Int}$.
	Let $1 \leq q \leq n/2$ be an integer and let $I_\bullet = \{ q\lfloor n/q \rfloor, \ldots, n \}$ when $q \lfloor n/q \rfloor < n$ and $I_\bullet = \emptyset$ otherwise. 
	Then, for all $z \geq 0$,
	\begin{align*}
		\P\left( \left\lVert \sum_{i=1}^n \Xi_{i,n} \right\rVert \geq 6 z \right) 
		& \leq 
		\frac{n^{r+1}}{q^{r} (d_2)^{r/2 - 1} z^r} \, \Delta^\Xi_r(q) 
		+ \P\left( \left\lVert \sum_{i \in I_\bullet} \Xi_{i,n} \right\rVert \geq z \right) + \\
		&\qquad\qquad\qquad 2 (d_1 + d_2) \exp\left( \frac{-z^2/2}{n q S^2_n + q R_n z/3} \right) ,
	\end{align*}
	where $\lVert \sum_{i \in I_\bullet} \Xi_{i,n} \rVert := 0$ whenever $I_\bullet = \emptyset$.
\end{theorem}

\begin{proof}
	To control dependence, we can adapt the interlacing block approach outlined by \cite{chenSelfnormalizedCramertypeModerate2016}. To interlace the sum, split it into
	\begin{equation*}
		\sum_{i=1}^n \Xi_{i,n} =  \sum_{j \in K_e} J_k + \sum_{j \in J_o} W_k + \sum_{i \in I_\bullet} \Xi_{i,n} ,
	\end{equation*}
	where $W_j := \sum_{i = q(j-1)+1}^{q j} \Xi_{i,n}$ for $j = 1, \ldots, \lfloor n / q \rfloor$ are the blocks, $I_\bullet := \{q \lfloor n / q \rfloor + 1, \ldots, n\}$ if $q \lfloor n / q \rfloor < n$ and $J_e$ and $J_o$ are the subsets of even and odd numbers of $\{ 1, \ldots, \lfloor n / q \rfloor \}$, respectively. For simplicity define $J = J_e \cup J_o$ as the set of block indices and let
	\begin{equation*}
		W^\dagger_j := \E\big[\, W_j \,\vert\, \epsilon_\ell, \, q(j-2)+1 \leq \ell \leq q j \,\big] .
	\end{equation*}
	%
	Note that by construction $\{W^\dagger_j\}_{j \in J_e}$ are independent and also $\{W^\dagger_j\}_{j \in J_o}$ are independent.
	Using the triangle inequality we find
	\begin{align*}
		\P\left( \left\lVert \sum_{i=1}^n \Xi_{i,n} \right\rVert \geq 6 z \right) 
		& \leq \P\left( \left\lVert \sum_{j \in J} ( W_j - W_j^\dagger ) \right\rVert + \left\lVert \sum_{j \in J} W_j^\dagger \right\rVert + \left\lVert \sum_{i \in I_\bullet} \Xi_{i,n} \right\rVert \geq 6 z \right) \\
		& \leq \P\left( \left\lVert \sum_{j \in J} ( W_j - W_j^\dagger ) \right\rVert \geq z \right) +  \P\left( \left\lVert \sum_{j \in J_e} W_j^\dagger \right\rVert \geq z \right) \\
		& \qquad + \P\left( \left\lVert \sum_{j \in J_o} W_j^\dagger \right\rVert \geq z \right) + \P\left( \left\lVert \sum_{i \in I_\bullet} \Xi_{i,n} \right\rVert \geq z \right) \\[3pt]
		& = I + II + III + IV .
	\end{align*}
	We keep term $IV$ as is. As in the proof of \cite{chenOptimalUniformConvergence2015}, terms $II$ and $III$ consist of sums of independent matrices, where each $W_j^\dagger$ satisfies $\lVert W_j^\dagger \rVert \leq q R_n$ and
	\begin{equation*}
		\max\left\{ \E\left[ \lVert W_j^\dagger \, W_j^{\dagger\prime} \rVert \right], \E\left[ \lVert W_j^{\dagger\prime} \, W_j^\dagger \rVert \right] \right\} \leq q S^2_n .
	\end{equation*}
	Then, using the exponential matrix inequality of \cite{troppUserFriendlyTailBounds2012},
	\begin{equation*}
		\P\left( \left\lVert \sum_{j \in J_e} W_k^\dagger \right\rVert \geq z \right) \leq (d_1 + d_2) \exp\left( \frac{-z^2/2}{n q S^2_n + q R_n z/3} \right) .
	\end{equation*}
	The same holds for the sum over $J_o$.
	Finally, we use the physical dependence measure $\Delta^\Xi_r$ to bound $I$. Start with the union bound to find
	\begin{align*}
		\P\left( \left\lVert \sum_{j \in J} ( W_j - W_j^\dagger ) \right\rVert \geq z \right) 
		& \leq \P\left( \sum_{j \in J} \left\lVert W_j - W_j^\dagger \right\rVert \geq z \right) \\
		& \leq \frac{n}{q} \, \P\left( \left\lVert W_j - W_j^\dagger \right\rVert \geq \frac{q}{n}\, z \right) ,
	\end{align*}
	where we have used that $\lfloor n / q \rfloor \leq n / q$.
	Since $W_j$ and $W_j^\dagger$ differ only over a $\sigma$-algebra that is $q$ steps in the past, by assumption 
	\begin{equation*}
		\left\lVert W_j - W_j^\dagger \right\rVert_{L^r} \leq q\, \Delta^\Xi_r(q) ,
	\end{equation*}
	which implies, by means of the $r$th moment inequality,
	\begin{equation*}
		\P\left( \left\lVert W_j - W_j^\dagger \right\rVert \geq \frac{q}{n}\, z \right) 
		\leq \P\left( (d_2)^{1/r-1/2} \left\lVert W_j - W_j^\dagger \right\rVert_r \geq \frac{q}{n}\, z \right) 
		\leq  \frac{n^r}{q^{r-1} (d_2)^{r/2 - 1} z^r} \, \Delta^\Xi_r(q) . 
	\end{equation*}
	where $(d_2)^{1/r-1/2}$ is the operator norm equivalence constant such that $\lVert \cdot \rVert \geq (d_2)^{1/r-1/2} \lVert \cdot \rVert_r$ \citep{fengEquivalenceConstantsCertain2003}.
	Therefore, 
	\begin{equation*}
		\P\left( \bigg\lVert \sum_{j \in J} ( W_j - W_j^\dagger ) \bigg\rVert \geq z \right)
		\leq
		\frac{n^{r+1}}{q^{r} (d_2)^{r/2 - 1} z^r} \, \Delta^\Xi_r(q) 
	\end{equation*}
	as claimed.
\end{proof}

Notice that the first term in the bound is weaker than that derived by \cite{chenOptimalUniformConvergence2015}. The $\beta$-mixing assumption and Berbee's Lemma give strong control over the probability $\P( \lVert \sum_{j \in J} ( W_j - W_j^\dagger ) \rVert \geq z )$. In contrast, assuming physical dependence means we have to explicitly handle a moment condition. One might think of sharpening Theorem \ref{theorem:phys_dep_exp_matrix_ineq} by sidestepping the $r$th moment inequality (cf. avoiding Chebyshev's inequality in concentration results), but we do not explore this approach here.

The second step is to map the physical dependence of a generic vector time series $\{X_i\}_{i \in \Int}$ to matrix functions.

\begin{proposition}\label{proposition:phys_dep_series_to_matrix}
	Let $\{X_i\}_{i \in \Int}$ where $X_i = G(\ldots, \epsilon_{i-1}, \epsilon_i) \in \mathcal{X}$ for $\{\epsilon_j\}_{j \in \Int}$ i.i.d. be a sequence with finite $r$th moment, where $r > 0$, and functional physical dependence coefficients
	\begin{equation*}
		\Delta_r(h) = \sup_i \left\lVert\, X_{i+h} - G^{(h)}(X^*_i, \epsilon_{i+1:i+h}) \,\right\rVert_{L^r}
	\end{equation*}
	for $h \geq 1$. Let $\Xi_{i,n} = \Xi_n(X_i)$ for each $i$ where $\Xi_n : \mathcal{X} \to \Real^{d_1 \times d_2}$ be a sequence of measurable $d_1 \times d_2$ matrix-valued functions such that $\Xi_n = (v_1, \ldots,  v_{d_2})$ for $v_\ell \in \Real^{d_1}$. 
	If $\lVert \Xi_{i,n} \rVert_{L^r} < \infty$ and 
	$$ C_{\Xi,\ell} := \sup_{x \in \mathcal{X}} \lVert \nabla v_\ell(x) \rVert \leq C_\Xi < \infty , $$
	then matrices $\Xi_{i,n}$ have physical dependence coefficients
	\begin{equation*}
		\Delta^{\Xi}_r(h) = \sup_{i} \left\lVert\, \Xi_{i,n} - \Xi^{h*}_{i,n} \,\right\rVert_{L^r} 
		\leq
		\sqrt{d_1} \left( \frac{d_2}{d_1} \right)^{1/r} C_\Xi \, \Delta_r(h) ,
	\end{equation*}
	where $\Xi^{h*}_{i,n} = \Xi_n(G^{(h)}(X'_i, \epsilon_{i+1:i+h}))$.
\end{proposition}

\begin{proof}
	To derive the bound, we use $\Xi_n(X_i)$ and $\Xi_n(X^{h*}_i)$ in place of $\Xi_{i,n} $ and $ \Xi^{h*}_{i,n}$, respectively, where $X^{h*}_i = G^{(h)}(X^*_i, \epsilon_{i+1:i+h})$.
	First we move from studying the operator $r$-norm (recall, $r > 2$) to the Frobenius norm,
	\begin{equation*}
		\left\lVert \Xi_n(X_i) - \Xi_n(X^{h*}_i) \right\rVert_r 
		\leq (d_2)^{1/2 - 1/r} \left\lVert \Xi_n(X_i) - \Xi_n(X^{h*}_i) \right\rVert_F .
	\end{equation*}
	where as intermediate step we use the 2-norm.
	Let $\Xi_n = (v_1, \ldots,  v_{d_2})$ for $v_\ell \in \Real^{d_1}$ and $\ell \in 1, \ldots, d_2$, so that
	\begin{equation*}
		\left\lVert \Xi_n \right\rVert_F = \sqrt{ \sum_{\ell = 1}^{d_2} \lVert v_\ell \rVert^2 }
	\end{equation*}
	where $v_\ell = (v_{\ell 1}, \ldots, v_{\ell d_1})'$. Since $v_\ell : \mathcal{X} \to \Real^{d_1}$ are vector functions, the mean value theorem gives that
	\begin{equation*}
		\left\lVert \Xi_n(X_i) - \Xi_n(X^{h*}_i) \right\rVert_F 
		\leq \sqrt{ \sum_{\ell = 1}^{d_2}  C_{\Xi,\ell}^2 \, \lVert X_i - X^{h*}_i \rVert^2 } 
		\leq \sqrt{d_2} \, C_\Xi \, \lVert X_i - X^{h*}_i \rVert .
	\end{equation*}
	Combining results and moving from the vector $r$-norm to the 2-norm yields
	\begin{equation*}
		\left\lVert \Xi_n(X_i) - \Xi_n(X^{h*}_i) \right\rVert_r 
		\leq 
		(d_2)^{1 - 1/r} (d_1)^{1/2 - 1/r} \, C_\Xi \, \lVert X_i - X^{h*}_i \rVert_r .
	\end{equation*}
	The claim involving the $L^r$ norm follows immediately.
\end{proof}

The following Corollary, which specifically handles matrix functions defined as outer products of vector functions, is immediate and covers the setups of series estimation.

\begin{corollary}\label{corollary:phys_dep_series_to_matrix_square_lowrank}
	Under the conditions of Proposition \ref{proposition:phys_dep_series_to_matrix}, if 
	\begin{equation*}
		\Xi_n(X_i) = \xi_n(X_i) \xi_n(X_i)' + Q_n
	\end{equation*}
	where $\xi_n : \mathcal{X} \to \Real^{d}$ is a vector function and $Q_n \in \Real^{d \times d}$ is nonrandom matrix, then 
	\begin{equation*}
		\Delta^{\Xi}_r(h) \leq d^{\,3/2 - 2/r} \, C_\xi \, \Delta_r(h) ,
	\end{equation*}
	where $ C_{\xi} := \sup_{x \in \mathcal{X}} \lVert \nabla \xi_n(x) \rVert < \infty$.
\end{corollary}

\begin{proof}
	Matrix $Q_n$ cancels out since it is nonrandom and appears in both $\Xi_n(X_i)$ and $\Xi_n(X^{h*}_i)$. Since $\Xi_n(X_i)$ is square, the ratio of row to column dimensions simplifies.
\end{proof}

The following Corollaries to Theorem \ref{theorem:phys_dep_exp_matrix_ineq} can now be derived in a straightforward manner.

\begin{corollary}\label{corollary:phys_dep_exp_matrix_ineq_from_series}
	Under the conditions of Theorem \ref{theorem:phys_dep_exp_matrix_ineq} and Proposition \ref{proposition:phys_dep_series_to_matrix}, for all $z \geq 0$
	\begin{align*}
		\P\left( \left\lVert \sum_{i=1}^n \Xi_{i,n} \right\rVert \geq 6 z \right) 
		& \leq 
		\frac{n^{r+1}}{q^{r} z^r} (d_2)^{2 - (r/2 + 1/r)} (d_1)^{1/2 - 1/r} C_\Xi \, \Delta_r(q)
		+ \P\left( \left\lVert \sum_{i \in I_\bullet} \Xi_{i,n} \right\rVert \geq z \right) \\
		& \qquad\qquad\qquad\qquad + 2 (d_1 + d_2) \exp\left( \frac{-z^2/2}{n q S^2_n + q R_n z/3} \right) .
	\end{align*}
	where $\Delta_r(\cdot)$ if the functional physical dependence coefficient of $X_i$.
\end{corollary}

\begin{corollary}\label{corollary:Op_order_norm_sum_xi}
	Under the conditions of Theorem \ref{theorem:phys_dep_exp_matrix_ineq} and Proposition \ref{proposition:phys_dep_series_to_matrix}, if $q = q(n)$ is chosen such that 
	\begin{equation*}
		\frac{n^{r+1}}{q^{r}} (d_2)^{2 - (r/2 + 1/r)} (d_1)^{1/2 - 1/r} C_\Xi \, \Delta_r(q) = o(1)
	\end{equation*}
	and $R_n \sqrt{q  \log(d_1 + d_2)} = o(S_n \sqrt{n})$, then
	\begin{equation*}
		\left\lVert \sum_{i=1}^n \Xi_{i,n} \right\rVert = O_P \left( S_n \sqrt{n q \log(d_1 + d_2)} \right) .
	\end{equation*}
\end{corollary}

This result is almost identical to Corollary 4.2 in \cite{chenOptimalUniformConvergence2015}, with the only adaptation of using Theorem \ref{theorem:phys_dep_exp_matrix_ineq} as a starting point. 
Condition $R_n \sqrt{q  \log(d_1 + d_2)} = o(S_n \sqrt{n})$ is simple to verify in our setting by assuming, e.g., $q = o(n / \log(n))$, since $\log(d_1 + d_2) \lesssim \log(K)$ and $K = o(n)$.

\begin{proof}[Proof of Corollary~\ref{corollary:Op_order_norm_sum_xi}]
    Start by setting $z = C S_n \sqrt{n q  \log(d_1 + d_2)}$ for sufficiently large $C > 0$. 
    The first term in the bound of 
    Theorem~\ref{theorem:phys_dep_exp_matrix_ineq} is $o(1)$ by assumption. The second term is also $o(1)$, since
    \begin{align*}
        \P\left( \left\lVert \sum_{i \in I_\bullet} \Xi_{i,n} \right\rVert \geq z \right) 
        & \leq 
        \P\left( \sum_{i \in I_\bullet} \left\lVert  \Xi_{i,n} \right\rVert \geq z \right) \\
        & = 
        \P\left( \sum_{i \in q \lfloor n/q \rfloor}^n \left\lVert  \Xi_{i,n} \right\rVert \geq C S_n \sqrt{n q  \log(d_1 + d_2)} \right) \\
        & \leq 
        \P\left( R_n \sqrt{q} \geq C S_n \sqrt{n  \log(d_1 + d_2)} \right) ,
    \end{align*}
    where we have used that $\lVert \Xi_{i,n} \rVert \leq R_n$ for all $i$. As $1 \lesssim \sqrt{\log(d_1 + d_2)}$, we have
    \begin{equation*}
        R_n \sqrt{q} 
        = o ( S_n \sqrt{n} )
        = o \big( S_n \sqrt{n \log(d_1 + d_2)} \big) ,
    \end{equation*}
    meaning that $\P\left( \left\lVert \sum_{i \in I_\bullet} \Xi_{i,n} \right\rVert \geq z \right)$ is asymptotically negligible.
    The final order of probability follows directly from the last term obtained in Theorem~\ref{theorem:phys_dep_exp_matrix_ineq} under the assumption that $R_n \sqrt{q  \log(d_1 + d_2)} = o(S_n \sqrt{n})$.
\end{proof}

\paragraph*{}
Note that when $d_1 = d_2 \equiv K$, which is the case of interest in the series regression setup, the first condition in Corollary \ref{corollary:Op_order_norm_sum_xi} reduces to
\begin{equation*}
	K^{5/2 - (r/2 + 2/r)} \, C_\Xi \, \Delta_r(q) = o(1) ,
\end{equation*}
which also agrees with the rate of Corollary \ref{corollary:phys_dep_series_to_matrix_square_lowrank}.
Assumption \ref{assumption:sieve_regularity}(i) and a compact domain further allow to explicitly bound factor $C_\Xi$ by
\begin{equation*}
	C_\Xi \lesssim K^{\omega_2} ,
\end{equation*}
so that the required rate becomes
\begin{equation*}
	K^{\rho} \, \Delta_r(q) = o(1), \quad \text{where} \quad \rho := \frac{3}{2} - \frac{r}{2} + \omega_2 .
\end{equation*}

\begin{proof}[Proof of Lemma \ref{lemma:gmc_sample_gram_matrix_op1}]
	The proof follows from Corollary \ref{corollary:Op_order_norm_sum_xi} by the same steps of the proof of Lemma 2.2 in \cite{chenOptimalUniformConvergence2015}. Simply take
	\begin{equation*}
	    \Xi_{i,n} = n^{-1}\big( \widetilde{b}^K_{\pi}(X_i) \widetilde{b}^K_{\pi}(X_i)' - I_K \big)
	\end{equation*}
	and note that $R_n \leq n^{-1}(1 + \zeta_{K,n}^2 \lambda_{K,n}^2)$ and $S_n \leq n^{-2}(1 + \zeta_{K,n}^2 \lambda_{K,n}^2)$.
\end{proof}

For Lemma \ref{lemma:gmc_sample_gram_matrix_op1} to hold under GMC assumptions a valid choice for $q(n)$ is
$$ q(n) = \gamma^{-1} \log(K^\rho n^{r+1}) $$
where $\gamma$ as in Proposition \ref{prop:gmc_conditions_on_map}. This is due to
\begin{align*}
	\left( \frac{n}{q} \right)^{r+1} q K^\rho \Delta_r(q)
	& \lesssim \frac{n^{r+1}}{q^r} K^\rho \exp(- \gamma q) \\
	& \lesssim \frac{n^{r+1} K^\rho }{\log(K^\rho n^{r+1})^r} (K^\rho n^{r+1})^{-1}  \\
	& = \frac{1}{\log(K^\rho n^{r+1})^r} = o(1) .
\end{align*}
Note then that, if $\lambda_{K,n} \lesssim 1$ and $\zeta_{K,n} \lesssim \sqrt{K}$, since
\begin{equation*}
	\zeta_{K,n} \lambda_{K,n} \sqrt{\frac{q \log K}{n}}
	\lesssim
	\sqrt{ \frac{ K \log(K^\rho n^{r + 1}) \log(K) }{n} }
	\lesssim
	\sqrt{ \frac{ K \log(n^{\rho + r + 2}) \log(n) }{n} }
	\lesssim
	\sqrt{ \frac{ K \log(n)^2 }{n} } ,
\end{equation*}
to satisfy Assumption \ref{assumption:series_gram_matrix_convergence} we may assume $\sqrt{ { K \log(n)^2 }/{n} } = o(1)$ as in Remark 2.3 of \cite{chenOptimalUniformConvergence2015} for the case of exponential $\beta$-mixing regressors.

\section{Simulation Details}\label{appendix:sim_details}

\subsection{Benchmark Bivariate Design}

The first simulation setup involves a bivariate DGP where the structural shock does not directly affect other observables. This is a simple environment to check that, indeed, the two-step estimator recovers the nonlinear component of the model and impulse responses are consistently estimated, and that the MSE does not worsen excessively.

I consider three bivariate data generation processes. DGP 1 sets $X_t$ to be a fully exogenous innovation process,
\begin{equation}\label{sim_eg:DGP_1}
	\begin{split}
		X_t & = \epsilon_{1t} , \\
		Y_t & = 0.5 Y_{t-1} + 0.5 X_{t} + 0.3 X_{t-1} - 0.4 \max(0, X_t) + 0.3 \max(0, X_{t-1}) + \epsilon_{2t} .
	\end{split}
\end{equation}
DGP 2 adds an autoregressive component to $X_t$, but maintains exogeneity, 
\begin{equation}\label{sim_eg:DGP_2}
	\begin{split}
		X_t & = 0.5 X_{t-1} + \epsilon_{1t} , \\
		Y_t & = 0.5 Y_{t-1} + 0.5 X_{t} + 0.3 X_{t-1} - 0.4 \max(0, X_t) + 0.3 \max(0, X_{t-1}) + \epsilon_{2t} .
	\end{split}
\end{equation}
Finally, DGP 3 add an endogenous effect of $Y_{t-1}$ on the structural variable by setting
\begin{equation}\label{sim_eg:DGP_3}
	\begin{split}
		X_t & = 0.5 X_{t-1} + 0.2 Y_{t-1} + \epsilon_{1t} , \\
		Y_t & = 0.5 Y_{t-1} + 0.5 X_{t} + 0.3 X_{t-1} - 0.4 \max(0, X_t) + 0.3 \max(0, X_{t-1}) + \epsilon_{2t} .
	\end{split}
\end{equation}
Following Assumption~\ref{assumption:structural_model}, innovations are mutually independent. To accommodate Assumption~\ref{assumption:compactness}, both $\epsilon_{1t}$ and $\epsilon_{2t}$ are drawn from a truncated standard Gaussian distribution over $[-3, 3]$.\footnote{Let $e_{it} \sim \mathcal{N}(0,1)$ for $i = 1, 2$, then the truncated Gaussian innovations used in simulation are set to be $\epsilon_{it} = \min(\max(-3, e_{it}), 3)$. The resulting r.v.s have a non-continuous density with two mass points at -3 and 3. However, in practice, since these masses are negligible, for the moderate sample sizes used, this choice does not create issues.} 
All DGPs are centered to have a zero intercept in population.

We evaluate bias and MSE plots using $10\,000$ Monte Carlo simulation. For a chosen horizon $H$, the impact of a relaxed shock on $\epsilon_{1t}$ is evaluated on $Y_{t+h}$ for $h = 1, \ldots, H$. To compute the population IRF, we employ a direct simulation strategy that replicates the shock's propagation through the model, and we use $10^5$ replications. To evaluate the estimated IRF, the two-step procedure is implemented: a sample of length $n$ is drawn, the linear least squares and the semiparametric series estimators of the model are used to estimate the model, and the relaxed IRF is computed following Proposition \ref{prop:irf_iterate_algorithm}. For the sake of brevity, we discuss the case of $\delta = 1$, and we set the shock relaxation function to be
\begin{equation*}
	\rho(z) = \mathbb{I}\{x \leq 3\} \exp\left( 1 + \left( \left\lvert \frac{z}{3} \right\rvert^4 - 1 \right)^{-1} \right)
\end{equation*}
It can be easily checked that this choice of $\rho$ is compatible with shocks of size $0 \leq |\delta| \leq 1$. Choices of $\delta = -1$ and $\delta = \pm 0.5$ yield similar results in simulations, so we do not discuss them here.

Figure \ref{fig:mse_bias_DGP_1_2_3} contains the results for sample size $n = 240$. This choice is motivated by considering the average sample sizes found in most macroeconometric settings, which is equivalent to 20 years of monthly data or 60 years of quarterly data \citep{goncalvesImpulseResponseAnalysis2021}. The benchmark method is an OLS regression that relies on a priori knowledge of the underlying DGP specification. Given the moderate sample size, to construct the cubic spline sieve estimator of the nonlinear component of the model, we use a single knot, located at 0. The simulations in Figure \ref{fig:mse_bias_DGP_1_2_3} show that while the MSE is slightly higher for the sieve model, the bias is comparable across methods. Note that for DGP 3, due to the dependence of the structural variable on non-structural series lags, the MSE and bias increase significantly, and there is no meaningful difference in performance between the two estimation approaches.

\subsection{Structural Partial Identification Design}

To showcase the validity of the proposed sieve estimator under the type of partial structural identification discussed in the paper, we again rely on the simulation design proposed by \cite{goncalvesImpulseResponseAnalysis2021}. All specifications are block-recursive and require estimating the contemporaneous effects of a structural shock on non-structural variables, unlike in the previous section.

The form of the DGPs is
\begin{equation*}
	B_0  Z_t = B_1 Z_{t-1} + C_0 f(X_t) + C_1 f(X_{t-1}) + \epsilon_t ,
\end{equation*}
where in all variations of the model 
\begin{equation*}
	B_0 = \begin{bmatrix}
		1 & 0 & 0 \\
		-0.45 & 1 & -0.3 \\
		-0.05 & 0.1 & 1
	\end{bmatrix} ,
	\quad
	C_0 = \begin{bmatrix}
		0 \\
		-0.2 \\
		0.08
	\end{bmatrix} ,
	\quad \text{and} \quad 
	C_1 = \begin{bmatrix}
		0 \\
		-0.1 \\
		0.2
	\end{bmatrix} .
\end{equation*}
I focus on the case $f(x) = \max(0, x)$, since this type of nonlinearity is simpler to study. DGP 4 treats $X_t$ as an exogenous shock by setting
\begin{equation*}
	B_1 = \begin{bmatrix}
		0 & 0 & 0 \\
		0.15 & 0.17 & -0.18 \\
		-0.08 & 0.03 & 0.6
	\end{bmatrix} ;
\end{equation*}
DGP 5 adds serial correlation to $X_t$,
\begin{equation*}
	B_1 = \begin{bmatrix}
		-0.13 & 0 & 0 \\
		0.15 & 0.17 & -0.18 \\
		-0.08 & 0.03 & 0.6
	\end{bmatrix} ;
\end{equation*}
and DGP 6 includes dependence on $Y_{t-1}$,
\begin{equation*}
	B_1 = \begin{bmatrix}
		-0.13 & 0.05 & -0.01 \\
		0.15 & 0.17 & -0.18 \\
		-0.08 & 0.03 & 0.6
	\end{bmatrix} .
\end{equation*}
For these data-generating processes, we employ the same setup of simulations with DGPs 1-3, including the number of replications as well as the type of relaxed shock, as well as the sieve grid. Here, too, we evaluate MSE and bias of both the sieve and the correct specification OLS estimators with a sample size of $n = 240$ observations. The results in Figure \ref{fig:mse_bias_DGP_4_5_6} show again that there is little difference in terms of performance between the semiparametric sieve approach and a correctly-specified OLS regression.

\subsection{Model Misspecification}

The results from benchmark simulations support the use of the sieve IRF estimator in a sample of moderate size, since it performs comparably to a regression performed with a priori knowledge of the underlying DGP. We now show that the semiparametric approach is also robust to model misspecification compared to simpler specifications involving fixed choices for nonlinear transformations.

To this end, we modify DGP 2 to use a smooth nonlinear transformation to define the effect of structural variable $X_t$ on $Y_t$. That is, there is no compounding of linear and nonlinear effects. The autoregressive coefficient in the equation for $X_t$ is also increased to make the shock more persistent. The new data-generating process, DGP 7, is, thus, given by 
\begin{equation}
	\begin{split}
		X_t & = 0.8 X_{t-1} + \epsilon_{1t} , \\
		Y_t & = 0.5 Y_{t-1} + 0.9 \varphi(X_t) + 0.5 \varphi(X_{t-1}) + \epsilon_{2t} .
	\end{split}
\end{equation}
where $\varphi(x) := (x - 1)(0.5 + \tanh(x - 1)/2)$. 


To emphasize the difference in estimated IRFs, in this setup we focus on $\delta = \pm 2$, which requires adapting the choice of innovations and shock relaxation function. In simulations of DGP 7, $\epsilon_{1t}$ and $\epsilon_{2t}$ are both drawn from a truncated standard Gaussian distribution over $[-5, 5]$. The shock relaxation function of this setup is given by
\begin{equation*}
	\rho(z) = \mathbb{I}\{x \leq 5\} \exp\left( 1 + \left( \left\lvert \frac{z}{5} \right\rvert^{3.9} - 1 \right)^{-1} \right) .
\end{equation*}
This form of $\rho$ is adapted to choices of $\delta$ such that $0 < |\delta| \leq 2$. The sieve grid now consists of 4 equidistant knots within $(-5, 5)$. We use the same number of replications as in the previous simulations. Finally, the regression design is identical to that used for DGP 2 under correct specification.

The results obtained with sample size $n =2400$ are collected in Figure \ref{fig:mse_bias_DGP_2'}. We choose this larger sample size to clearly showcase the inconsistency of impulse responses under misspecification: as it can be observed, the simple OLS estimator involving the negative-censoring transform produces IRF estimates with consistently worse MSE and bias than those of the sieve estimator at almost all horizons. Similar results are also obtained for more moderate shocks $\delta = \pm 1$, but the differences are less pronounced. These simulations suggest that the semiparametric sieve estimator can produce substantially better IRF estimates in large samples than methods involving nonlinear transformations selected a priori.

In this setup, it is also important to highlight the fact that the poor performance of OLS IRF estimates does not come from $\varphi(x)$ being ``complex'', and, thus, hard to approximate by combinations of simple functions. In fact, if in DGP 7 function $\varphi$ is replaced by $\widetilde{\varphi}(x) := \varphi(x + 1)$, the differences between sieve and OLS impulse response estimates become minimal in simulations, with the bias of the latter decreasing by approximately an order of magnitude, see Figure \ref{fig:mse_bias_DGP_2'_alt}. This is simply due to the fact that $\widetilde{\varphi}(x)$ is well approximated by $\max(0, x)$ directly. However, one then requires either prior knowledge or sheer luck when constructing the nonlinear transforms of $X_t$ for an OLS regression. The proposed series estimator, instead, just requires an appropriate choice of sieve. Many data-driven procedures to select sieves in applications have been proposed, see for example the discussion in \cite{kangINFERENCENONPARAMETRICSERIES2021}.

\newpage 
\section{Additional Plots}\label{appendix:additional_plots}

\subsection{Primary Simulation Designs}

\begin{figure}[h!]
	\centering
	\begin{subfigure}[b]{0.9\textwidth}
		\centering
		\includegraphics[width=\textwidth]{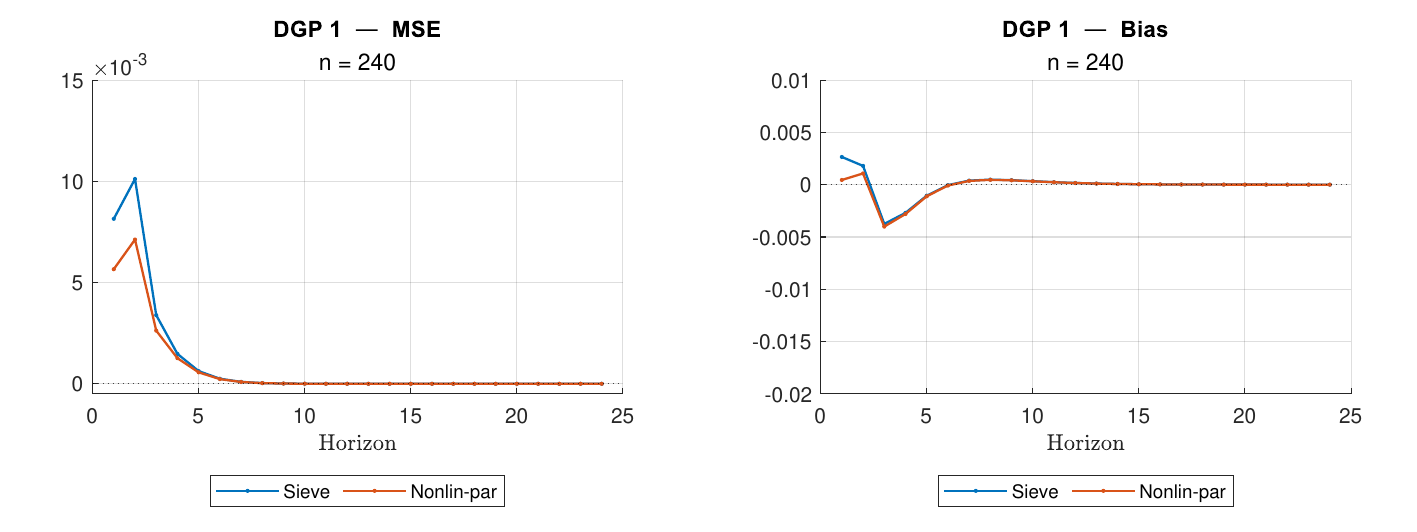}
	\end{subfigure}
	\\[15pt]
	\begin{subfigure}[b]{0.9\textwidth}
		\centering
		\includegraphics[width=\textwidth]{figures/plot_mse_bias_DGP2__n=240_deg=3_B=100000_M=10000.pdf}
	\end{subfigure}
	\\[15pt]
	\begin{subfigure}[b]{0.9\textwidth}
		\centering
		\includegraphics[width=\textwidth]{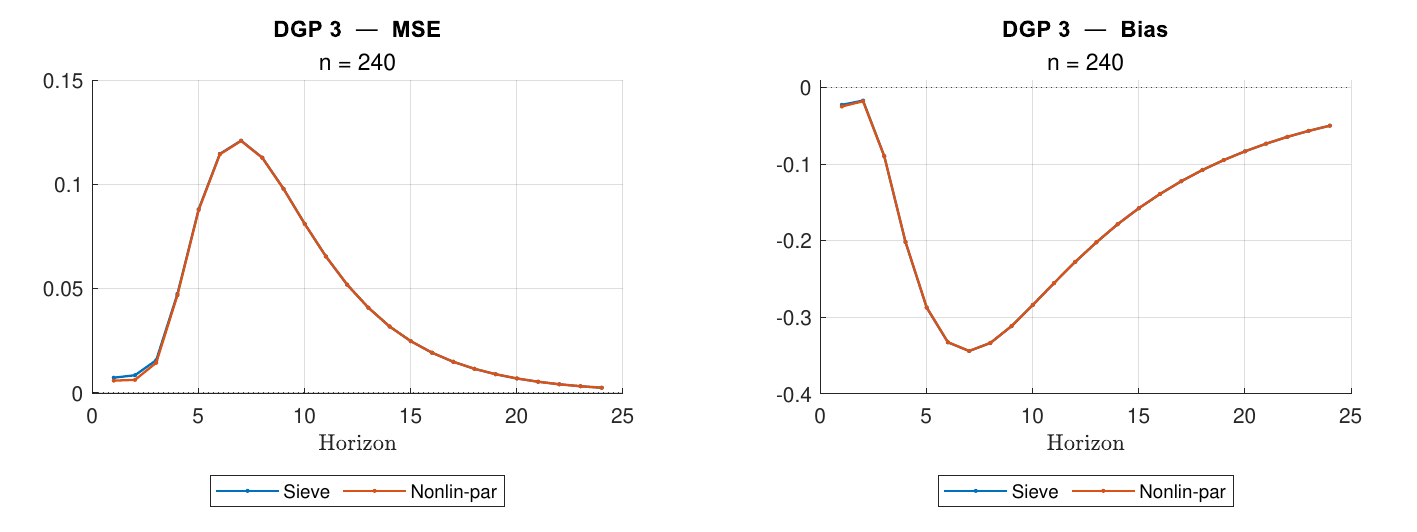}
	\end{subfigure}
	\\[15pt]
	\caption{Simulation results for DGPs 1-3.}
	\label{fig:mse_bias_DGP_1_2_3}
\end{figure}

\newpage
\begin{figure}[H]
	\centering
	\begin{subfigure}[b]{\textwidth}
		\centering
		\includegraphics[width=\textwidth]{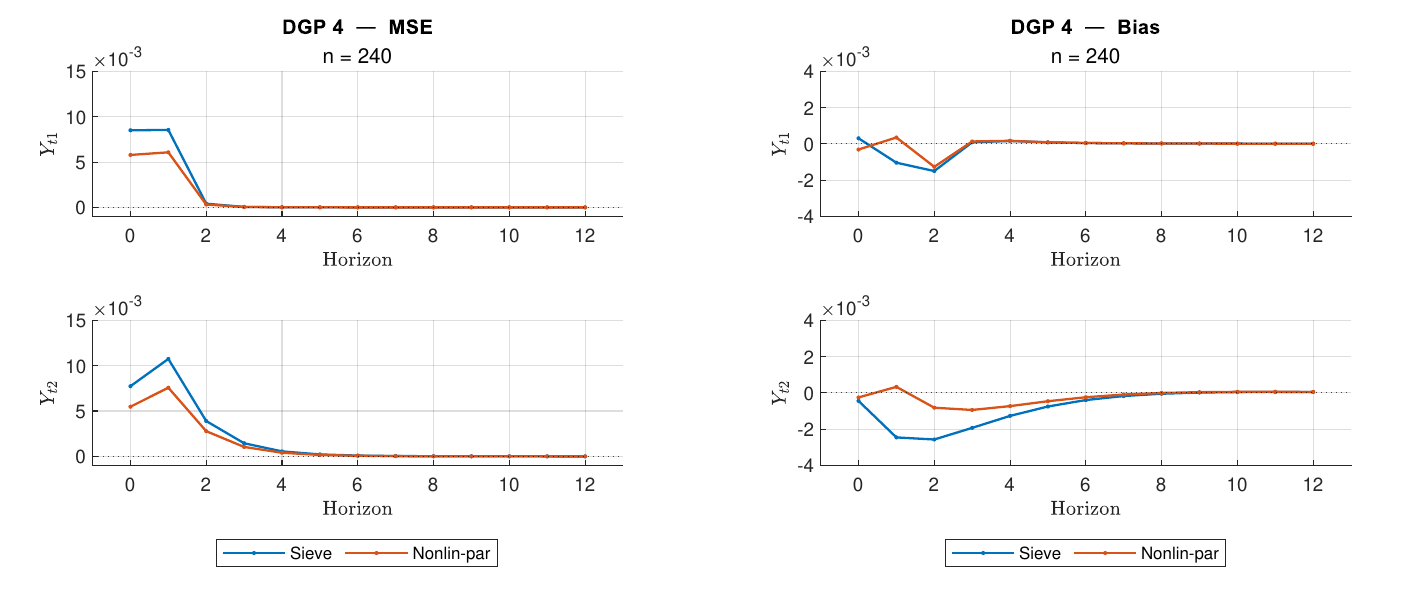}
	\end{subfigure}
	\\[5pt]
	\begin{subfigure}[b]{\textwidth}
		\centering
		\includegraphics[width=\textwidth]{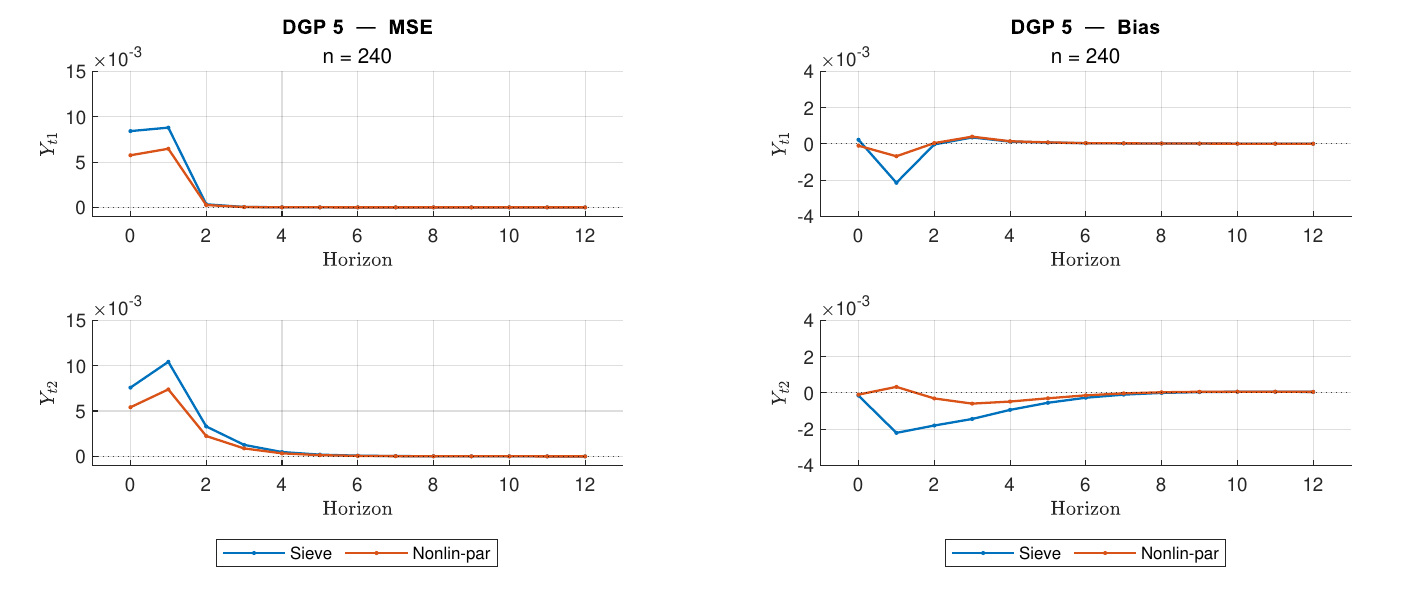}
	\end{subfigure}
	\\[5pt]
	\begin{subfigure}[b]{\textwidth}
		\centering
		\includegraphics[width=\textwidth]{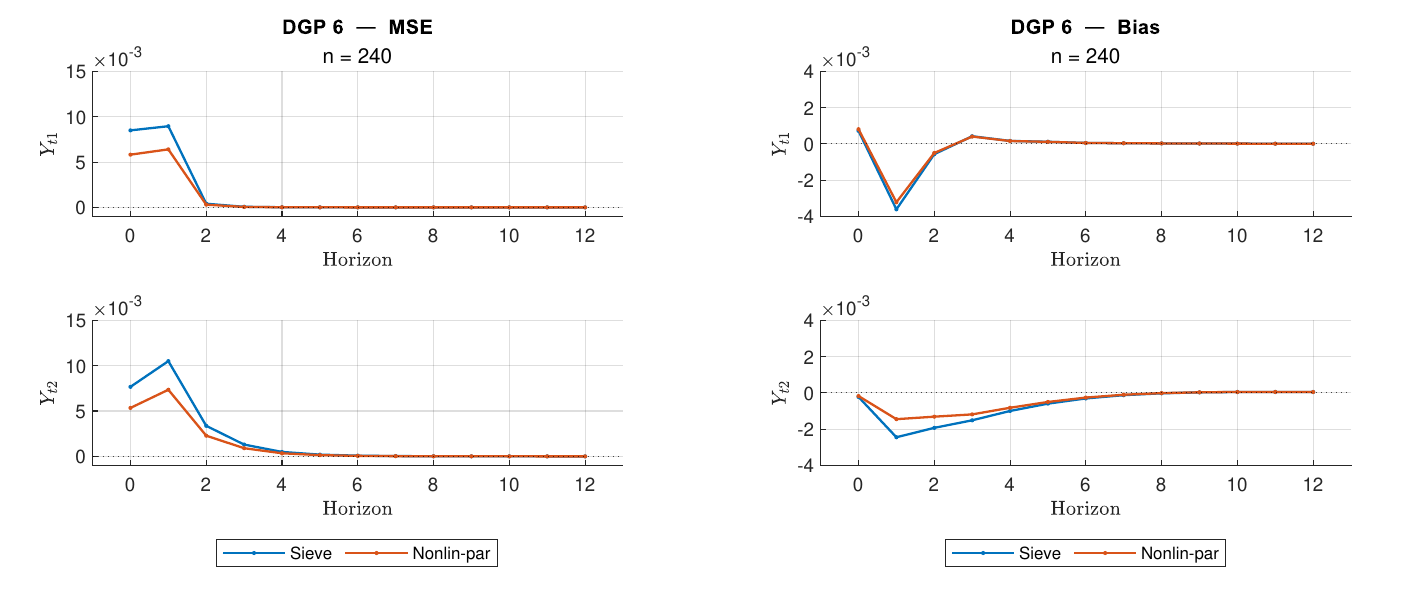}
	\end{subfigure}
	\\[5pt]
	\caption{Simulation results for DGPs 4-6.}
	\label{fig:mse_bias_DGP_4_5_6}
\end{figure}

\newpage

\subsection{Simulations under Misspecification}

\begin{figure}[H]
	\centering
	\begin{subfigure}[b]{\textwidth}
		\centering
		\includegraphics[width=\textwidth]{figures/plot_mse_bias_DGP2-smooth__n=2400_deg=3_B=100000_M=10000_delta=2.pdf}
		\caption{$\delta = +2$}
	\end{subfigure}
	\begin{subfigure}[b]{\textwidth}
		\centering
		\includegraphics[width=\textwidth]{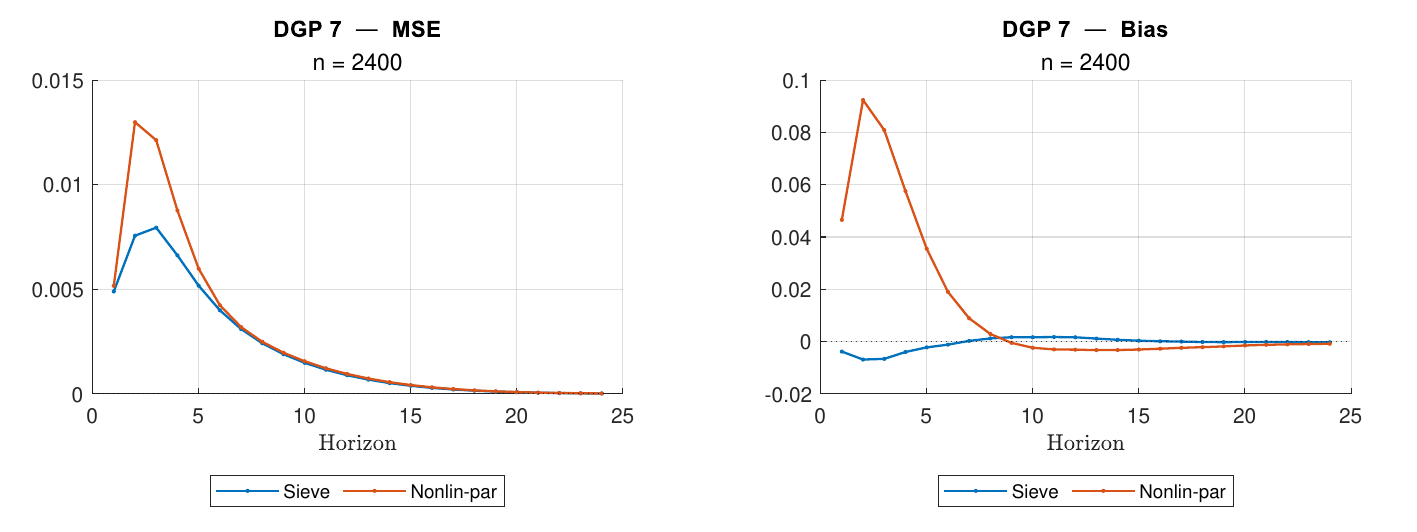}
		\caption{$\delta = -2$}
	\end{subfigure}
	\caption{Simulations results for DGP 7.}
	\label{fig:mse_bias_DGP_2'}
\end{figure}

\newpage
\begin{figure}[H]
    \centering
    \begin{subfigure}[b]{\textwidth}
		\centering
		\includegraphics[width=\textwidth]{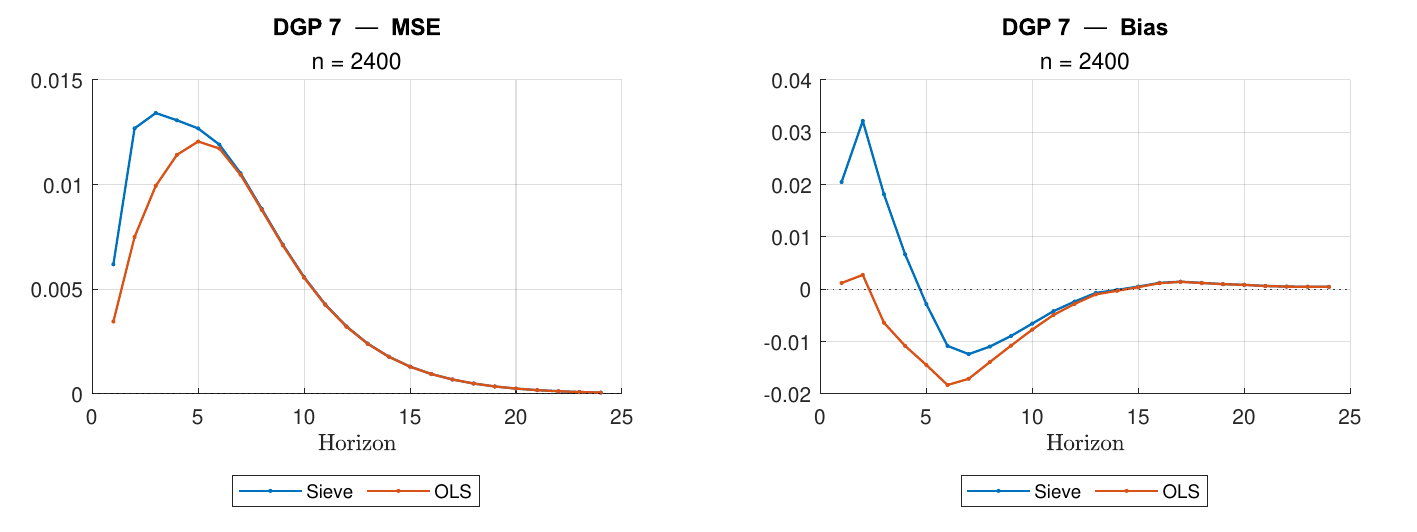}
		\caption{$\delta = +2$}
	\end{subfigure}
    \begin{subfigure}[b]{\textwidth}
		\centering
		\includegraphics[width=\textwidth]{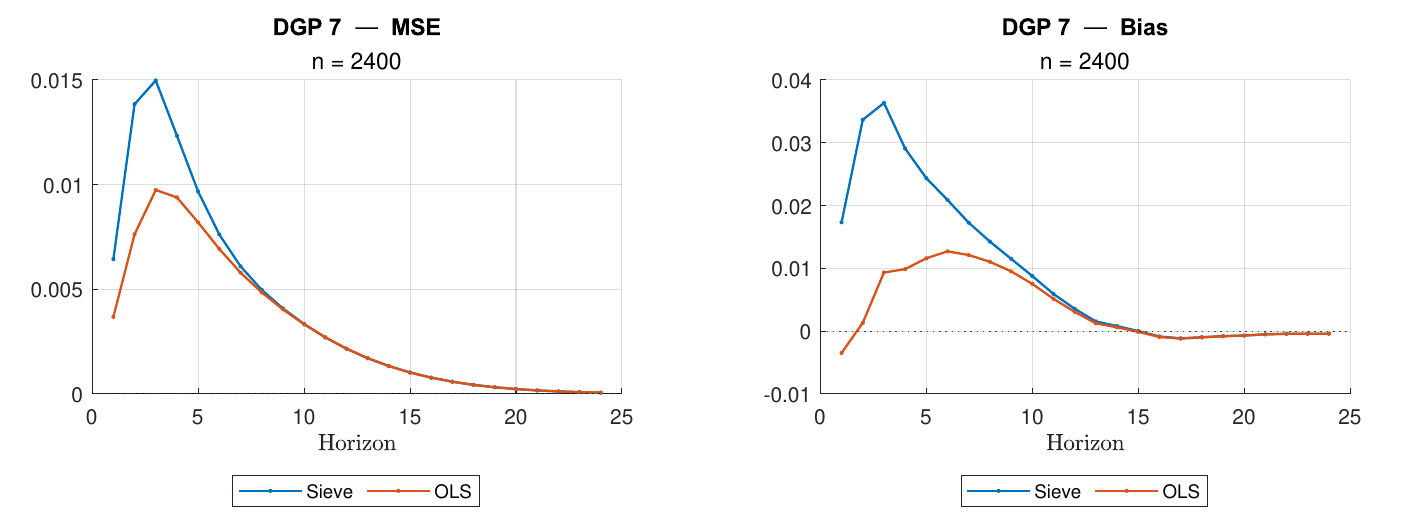}
		\caption{$\delta = -2$}
	\end{subfigure}
	\\[15pt]
    \caption{Simulation results for DGP 7 when considering $\widetilde{\varphi}$ in place of $\varphi$.}
    \label{fig:mse_bias_DGP_2'_alt}
\end{figure}

\newpage

\subsection{Empirical Applications}

\begin{figure}[h!]
    \centering
    \includegraphics[width=0.7\textwidth]{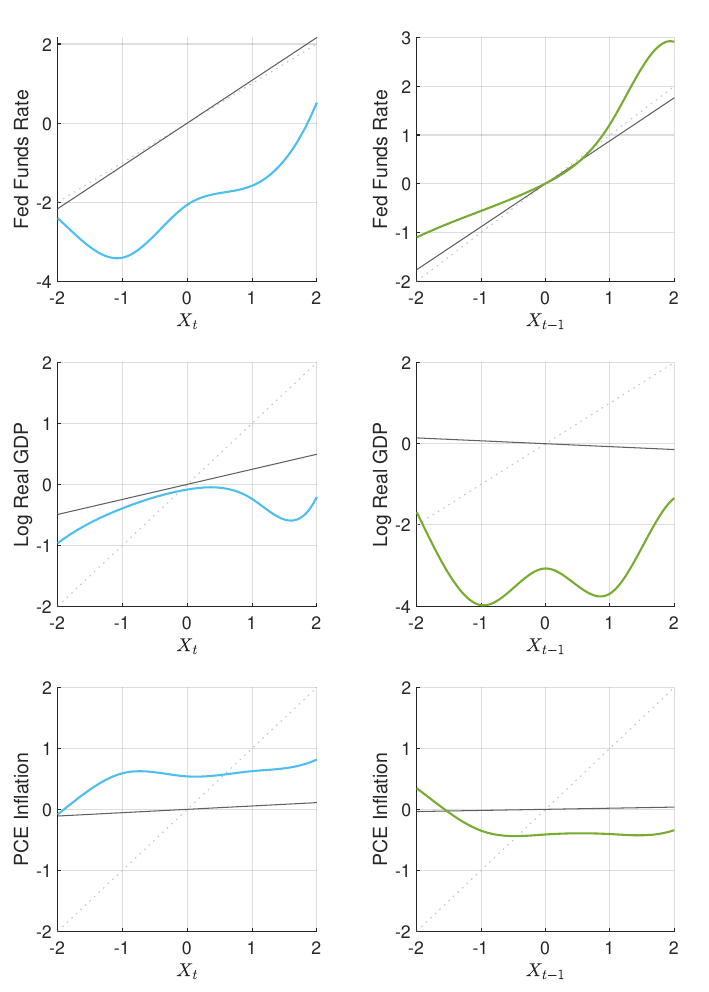}
    \caption{Estimated nonlinear regression functions for the narrative U.S. monetary policy variable. Contemporaneous (left side) and one-period lag (right side) effects are shown, linear and nonlinear functions. For comparison, linear VAR coefficients (dark gray) and the identity map (light gray, dashed) are shown as lines.}
    \label{fig:plot_app_gonc2021_regfuns}
\end{figure}

\newpage
\begin{figure}[H]
    \centering
    \begin{subfigure}[b]{0.85\textwidth}
		\centering
		\includegraphics[width=\textwidth]{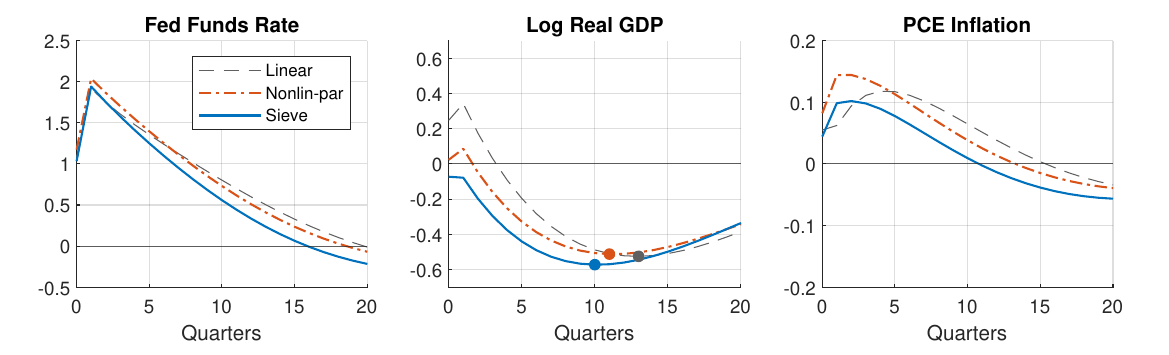}
		\caption{$\delta = +1$, knots at $\{-1, 1\}$}
	\end{subfigure}
    \\
    \begin{subfigure}[b]{0.85\textwidth}
		\centering
		\includegraphics[width=\textwidth]{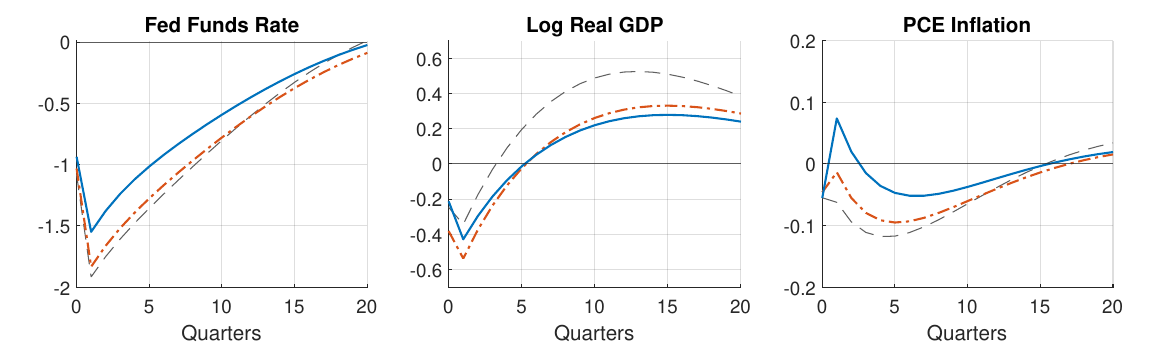}
		\caption{$\delta = -1$, knots at $\{-1, 1\}$}
	\end{subfigure}
    \\[10pt]
    \begin{subfigure}[b]{0.85\textwidth}
		\centering
		\includegraphics[width=\textwidth]{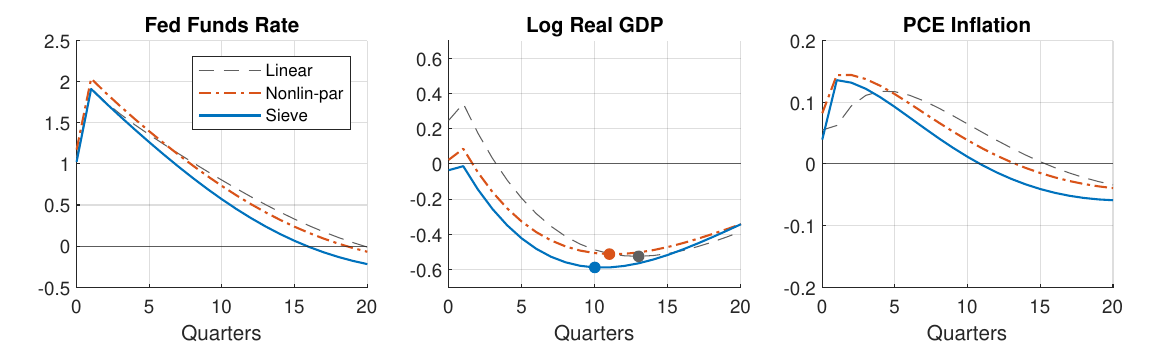}
		\caption{$\delta = +1$, knot at $\{0\}$}
	\end{subfigure}
    \\
    \begin{subfigure}[b]{0.85\textwidth}
		\centering
		\includegraphics[width=\textwidth]{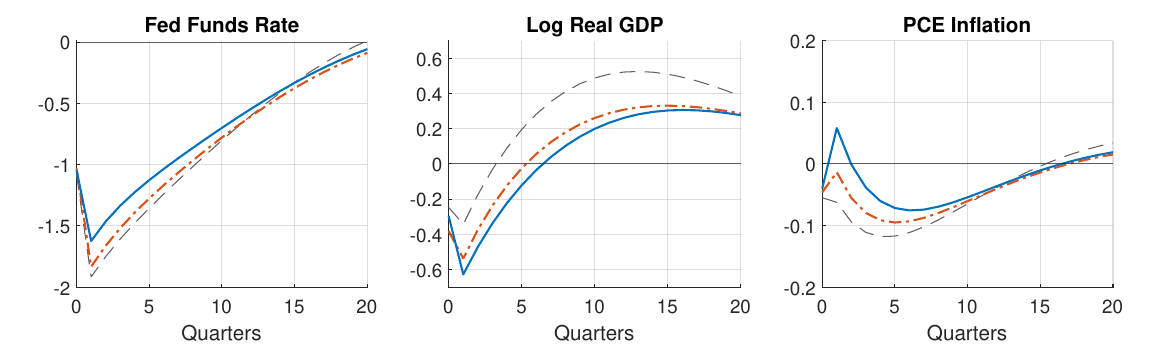}
		\caption{$\delta = -1$, knot at $\{0\}$}
	\end{subfigure}
	\\
    \caption{Robustness plots for U.S. monetary policy shock when changing knots compared to those used in Figure \ref{fig:app_gonc2021_irfs}. Note that linear and parametric nonlinear responses do not change.}
    \label{fig:plot_app_gonc2021_robustness}
\end{figure}

\newpage
\begin{figure}[H]
    \centering
    \textbf{GDP}\\
    \includegraphics[width=0.85\textwidth]{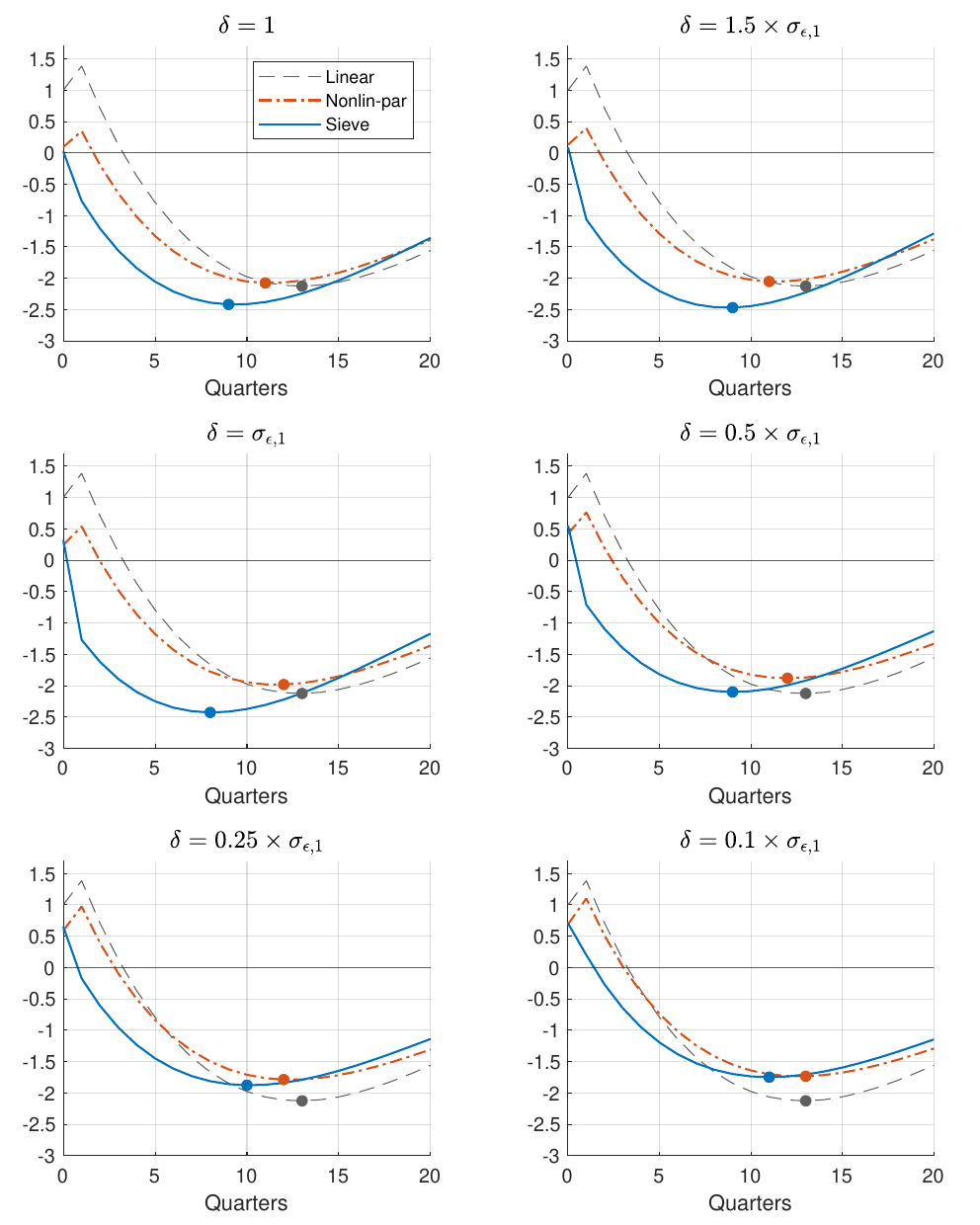} 
    \\[15pt]
    \caption{Relative changes in the GDP impulse responses function when the size of the shock is reduced from that used in Figure \ref{fig:app_gonc2021_irfs}. The standard deviation of $X_t \equiv \epsilon_{1t}$ is $\sigma_{\epsilon,1} \approx 0.5972$, therefore $\delta = 1 \approx 1.7 \times \sigma_{\epsilon,1}$. Linear IRFs are rescaled such that, for all values of $\delta$, the linear response at $h = 0$ is one in absolute value. Nonlinear IRFs are rescaled by $\delta$ times the linear response scaling factor.}
    \label{fig:plot_app_gonc2021_scale}
\end{figure}

\newpage
\begin{figure}[H]
    \centering
    \includegraphics[width=0.75\textwidth]{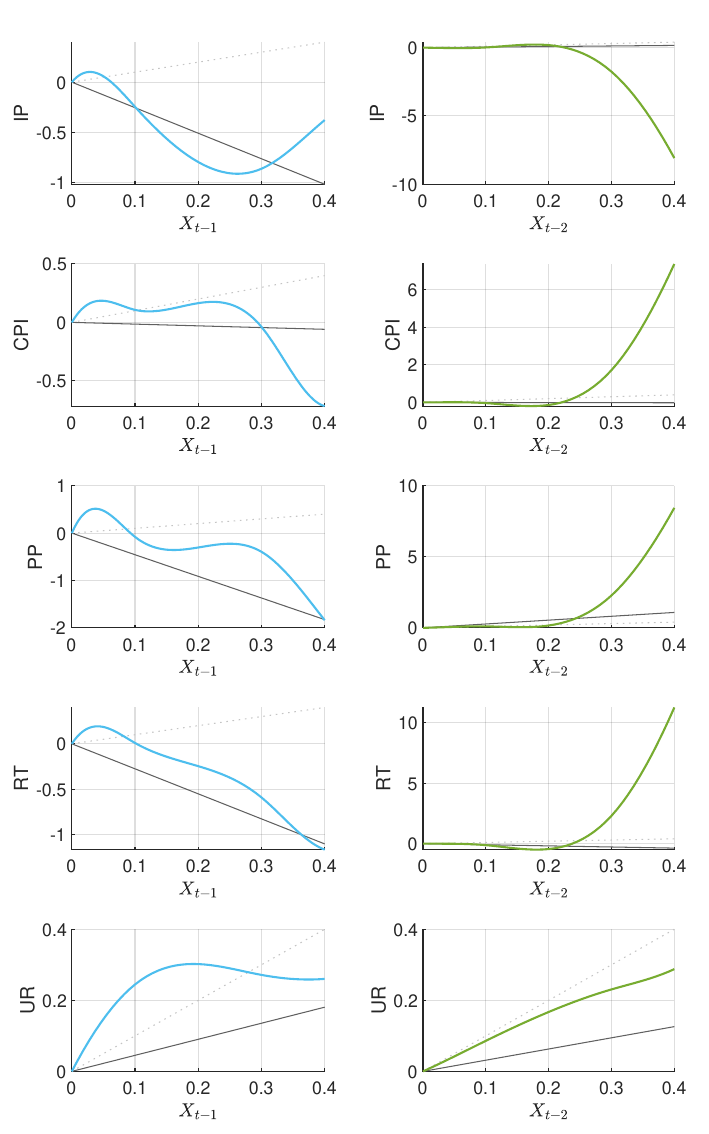}
    \caption{Estimated nonlinear regression functions for the 3M3M subjective interest rate uncertainty measure. One-period (left side) and two-period lag (right side) effects are shown, combining linear and nonlinear functions. For comparison, linear VAR coefficients (dark gray) and the identity map (light gray, dashed) are shown as lines.}
    \label{fig:plot_app_istrefi2018_regfuns}
\end{figure}

\newpage

\newpage
\begin{figure}[H]
	\centering
	\begin{subfigure}[b]{0.49\textwidth}
		\centering
		\includegraphics[width=\textwidth]{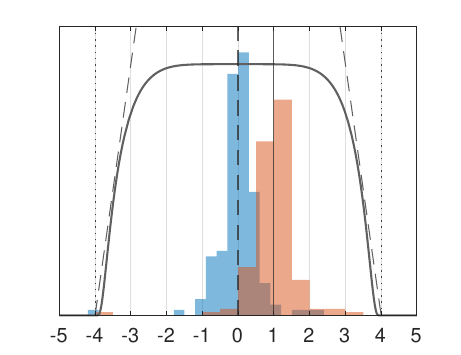}
		\caption{$\delta = +1$}
	\end{subfigure}
	\begin{subfigure}[b]{0.49\textwidth}
		\centering
		\includegraphics[width=\textwidth]{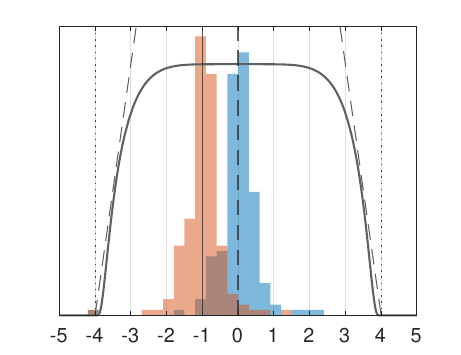}
		\caption{$\delta = -1$}
	\end{subfigure}
	\\[15pt]
	\caption{Comparison of histograms and shock relaxation function for a positive (left) and negative (right) shock in monetary policy. Original (blue) versus shocked (orange) distribution of the sample realization of $\epsilon_{1t}$. The dashed vertical line is the mean of the original distribution, while the solid vertical line is the mean after the shock.}
	\label{fig:plot_app_gonc2021_meanshift}
\end{figure}
\begin{figure}[H]
    \centering
    \begin{subfigure}[b]{0.49\textwidth}
		\centering
		\includegraphics[width=\textwidth]{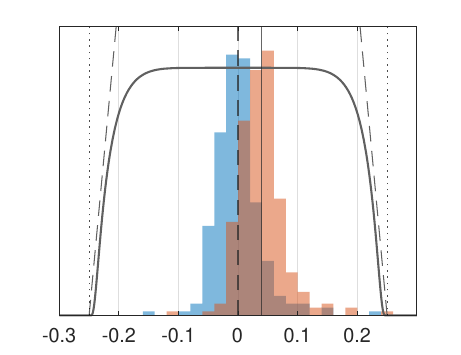}
		\caption{$\delta = \sigma_\epsilon$}
	\end{subfigure}
    \begin{subfigure}[b]{0.49\textwidth}
		\centering
		\includegraphics[width=\textwidth]{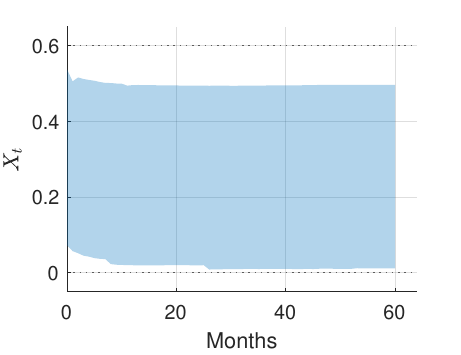}
		\caption{Envelope}
	\end{subfigure}
    \\[15pt]
    \caption{Left: Histograms and shock relaxation function for a one-standard-deviation shock in interest rate uncertainty. Original (blue) versus shocked (orange) distribution of the sample realization of $\epsilon_{1t}$. The dashed vertical line is the mean of the original distribution, while the solid vertical line is the mean after the shock. Right: Envelope (min-max) of shocked paths for one-standard-deviation impulse response.}
    \label{fig:plot_app_istrefi2018_meanshift}
\end{figure}

\newpage
\begin{figure}[H]
    \centering
    \textbf{IP}\\
    \includegraphics[width=0.85\textwidth]{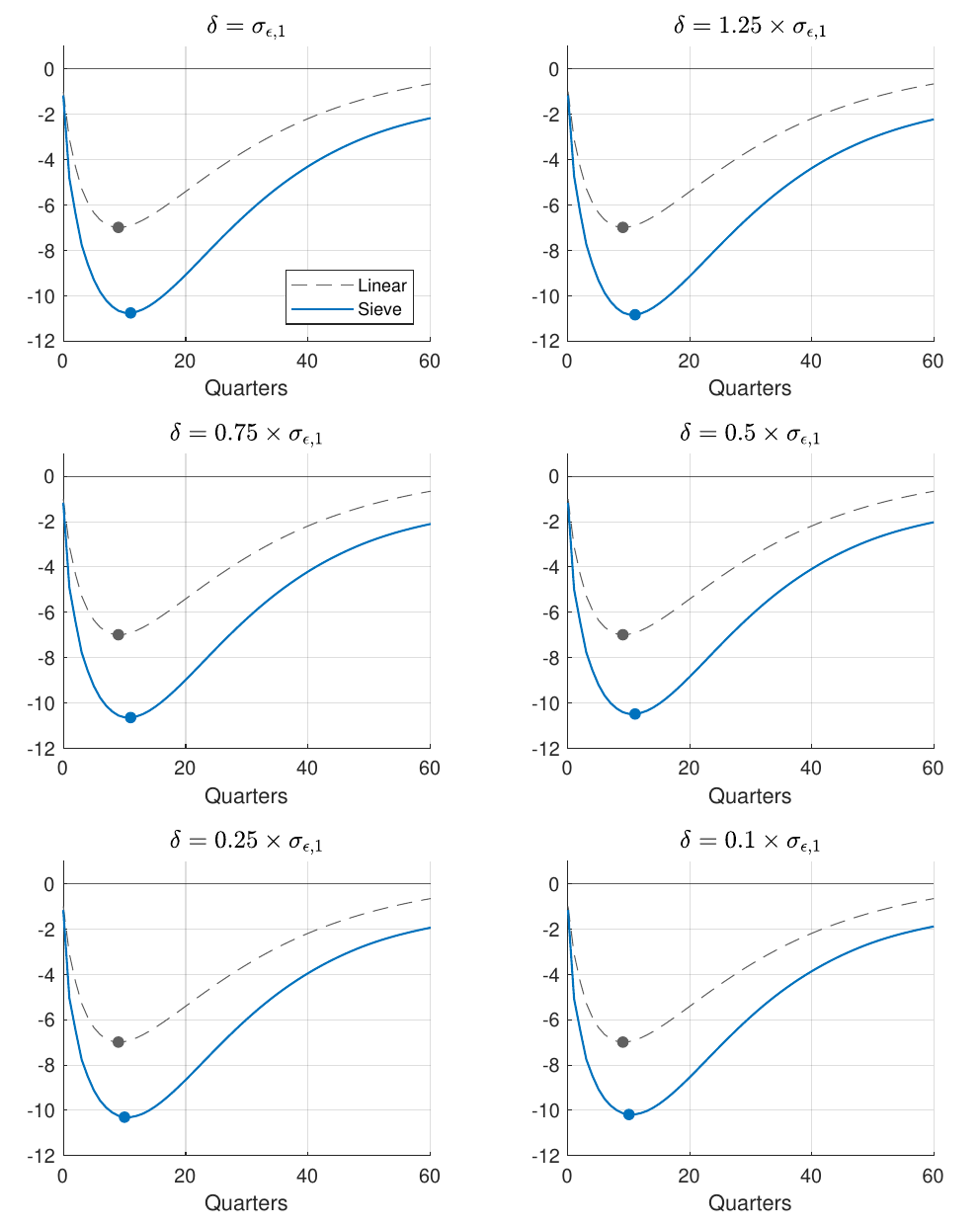} 
    \\[15pt]
    \caption{Relative changes in the industrial production impulse responses function when the size of the shock is reduced from that used in Figure \ref{fig:app_istrefi2018_irfs}. The standard deviation of $\equiv \epsilon_{1t}$ is $\sigma_{\epsilon,1} \approx 0.0389$. Linear IRFs are rescaled such that, for all values of $\delta$, the linear response at $h = 0$ is one in absolute value. Nonlinear IRFs are rescaled by $\delta$ times the linear response scaling factor.}
    \label{fig:plot_app_istrefi2018_scale_IP}
\end{figure}

\newpage
\begin{figure}[H]
    \centering
    \textbf{CPI}\\
    \includegraphics[width=0.85\textwidth]{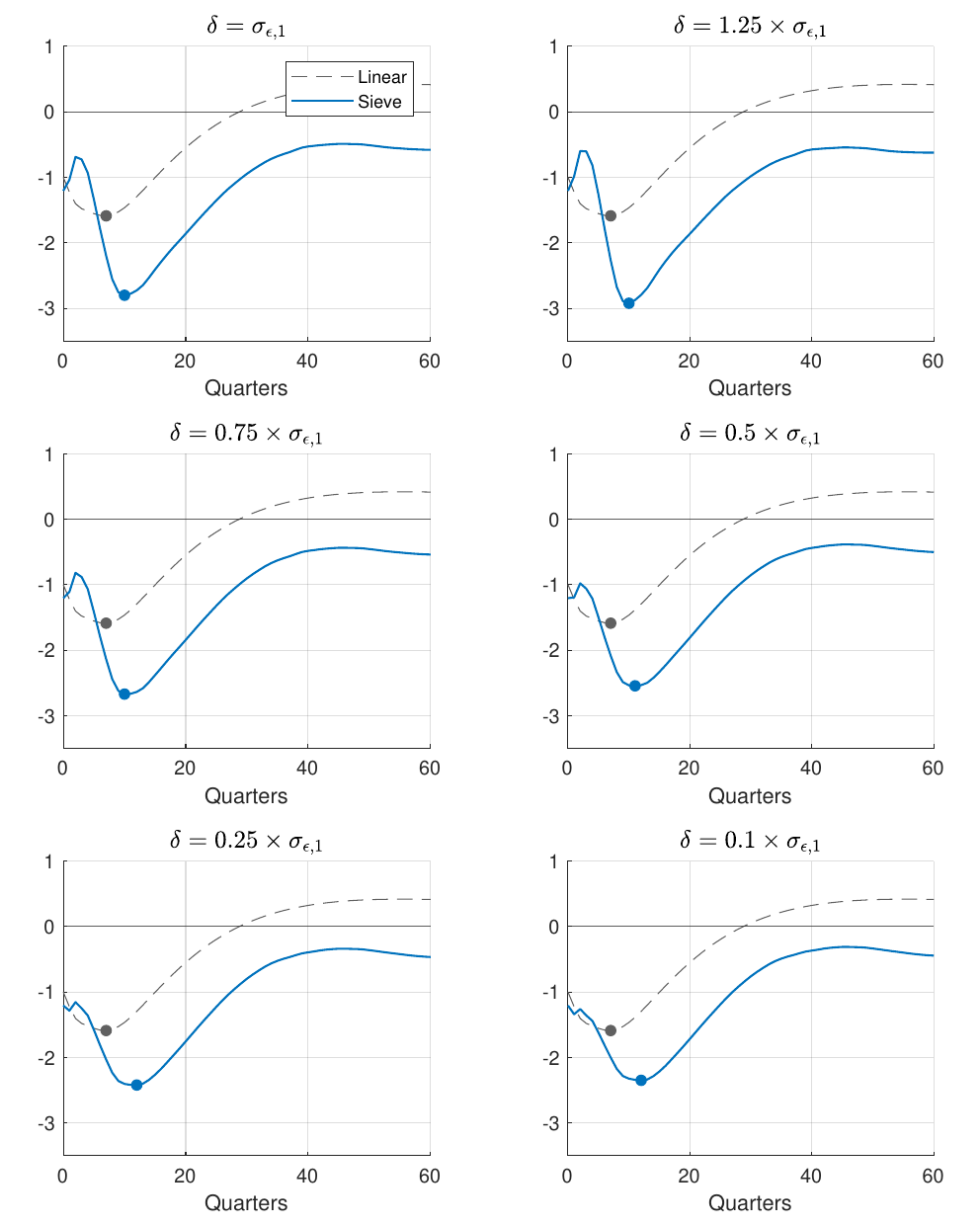} 
    \\[15pt]
    \caption{Relative changes in the CPI impulse responses function when the size of the shock is reduced from that used in Figure \ref{fig:app_istrefi2018_irfs}. The standard deviation of $\equiv \epsilon_{1t}$ is $\sigma_{\epsilon,1} \approx 0.0389$. Linear IRFs are rescaled such that, for all values of $\delta$, the linear response at $h = 0$ is one in absolute value. Nonlinear IRFs are rescaled by $\delta$ times the linear response scaling factor.}
    \label{fig:plot_app_istrefi2018_scale_CPI}
\end{figure}

\section{Robustness to Relaxation}\label{appendix:robustness}

Since shock relaxation implies that the researcher chooses an additional functional hyperparameter ($\rho$) in the construction of impulse responses, one should wonder whether there is much difference in implementing relaxation or not. In this appendix, we provide additional plots that speak to the robustness of our main simulation and application results by including semiparametric nonlinear IRFs computed using non-relaxed shocks. 

\paragraph{Simulation Checks.}
In Figures~\ref{fig:mse_bias_DGP_1_2_3__norelax} and \ref{fig:mse_bias_DGP_4_5_6__norelax} we plot parametric and semiparametric IRFs for the simulation designs discussed in both Section~\ref{section:simulations} and Appendix~\ref{appendix:sim_details}, with the important change that now the population response we target is \textit{without shock relaxation}. That is, we follow Section~\ref{section:nonlin_irfs} directly to define the target responses of interest. We compared both relaxed and non-relaxed sieve IRFs, although for the latter we cannot provide a formal theory of consistency in our framework. As we can see from both figures, relaxed sieve IRFs generally suffer from a significant increase in bias at short horizons. However, in terms of MSE, the ranking reversed: non-relaxed impulse responses have generally higher mean squared error than their relaxed counterparts, while parametric nonlinear IRFs have both the smallest MSE and bias.
These results seem to agree with the natural intuition that the semiparametric estimates are more inaccurate at the edges of the regressor's domain, and thus may induce higher variation in the obtained IRFs.

\paragraph{Empirical Applications Checks.}
By additionally including impulse responses obtained with a non-relaxed shock, we can conclude that our discussion in Section~\ref{section:applications} remains broadly valid.
With regards to our MP application, Figure~\ref{fig:app_gonc2021_irfs__norelax} and \ref{fig:plot_app_gonc2021_meanshift__norelax} show that, with a positive shock, there are negligible differences between relaxed and non-relaxed semiparametric IRFs. This is also true as we scaled down the intensity of the shock, cf. Figure~\ref{fig:plot_app_gonc2021_scale___norelax}. However, with a negative shock, we can see more noticeable disagreements. This is likely due to the negative shift of a small mass of residuals as the edge of the support, as revealed by Figure~\ref{fig:plot_app_gonc2021_meanshift__norelax}(b), having high leverage on the overall sample average used in computing IRFs.

Lastly, when checking the differences between relaxed and non-relaxed sieve IRFs in our second application, Figures~\ref{fig:plot_app_istrefi2018_meanshift__norelax} and \ref{fig:plot_app_istrefi2018_scale_IP__norelax} also indicate that the main analysis is robust to doing away with relaxation. We can see some differences in the shape of responses, mainly for CPI inflation and retail IRFs. One noticeable fact is that, when the shock size is reduced, there is practically no difference between the two types of semiparametric responses we construct.

\pagebreak
\begin{figure}[H]
	\centering
	\begin{subfigure}[b]{0.9\textwidth}
		\centering
		\includegraphics[width=\textwidth]{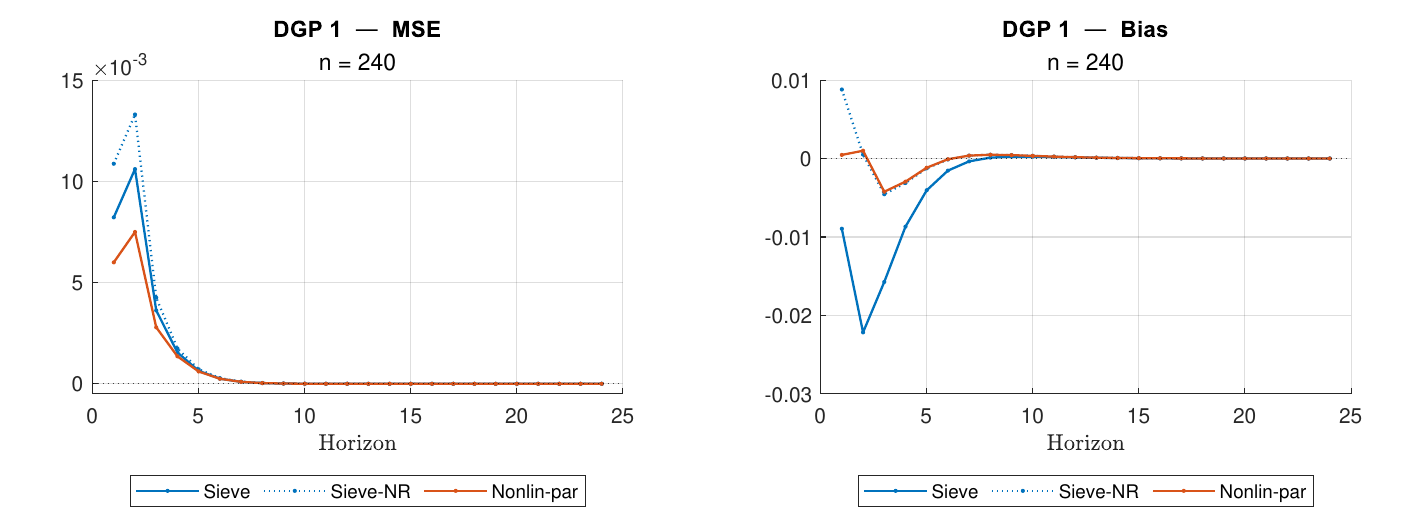}
	\end{subfigure}
	\\[15pt]
	\begin{subfigure}[b]{0.9\textwidth}
		\centering
		\includegraphics[width=\textwidth]{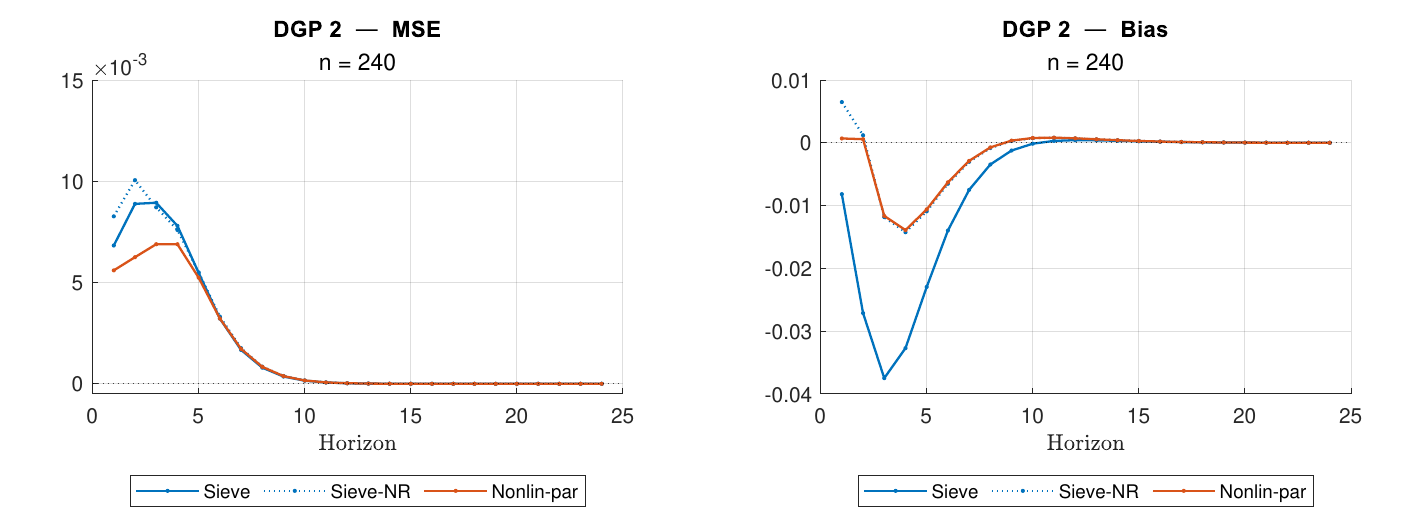}
	\end{subfigure}
	\\[15pt]
	\begin{subfigure}[b]{0.9\textwidth}
		\centering
		\includegraphics[width=\textwidth]{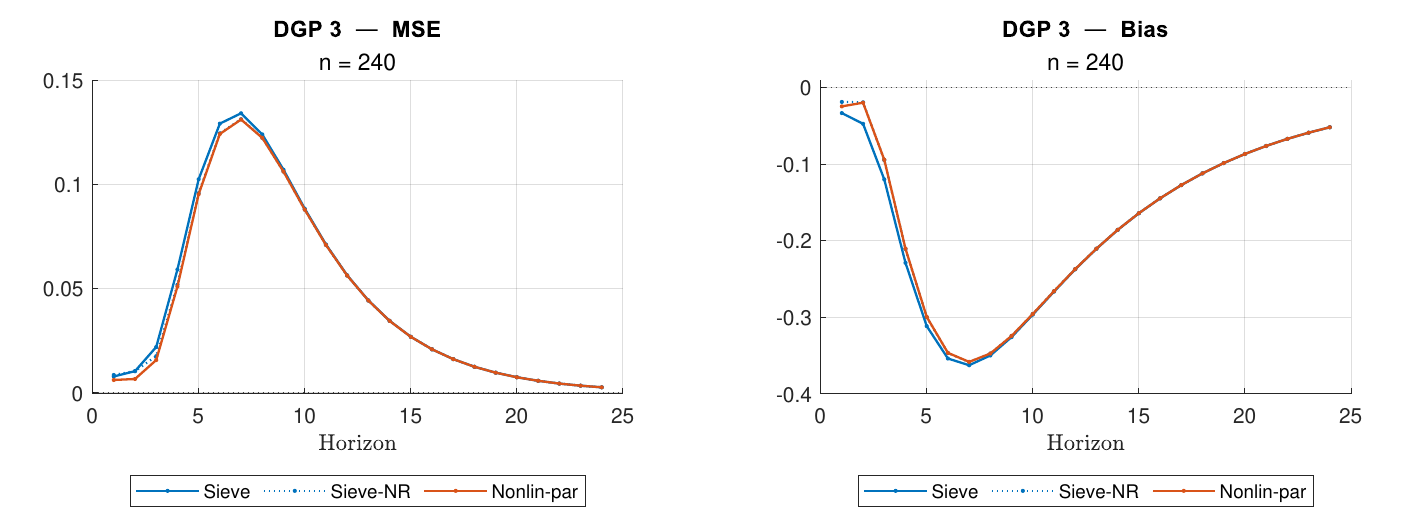}
	\end{subfigure}
	\\[15pt]
	\caption{Simulations results for DGPs 1-3 targeting population IRFs \textit{without relaxation} of the shock.}
	\label{fig:mse_bias_DGP_1_2_3__norelax}
\end{figure}

\newpage
\begin{figure}[H]
	\centering
	\begin{subfigure}[b]{\textwidth}
		\centering
		\includegraphics[width=\textwidth]{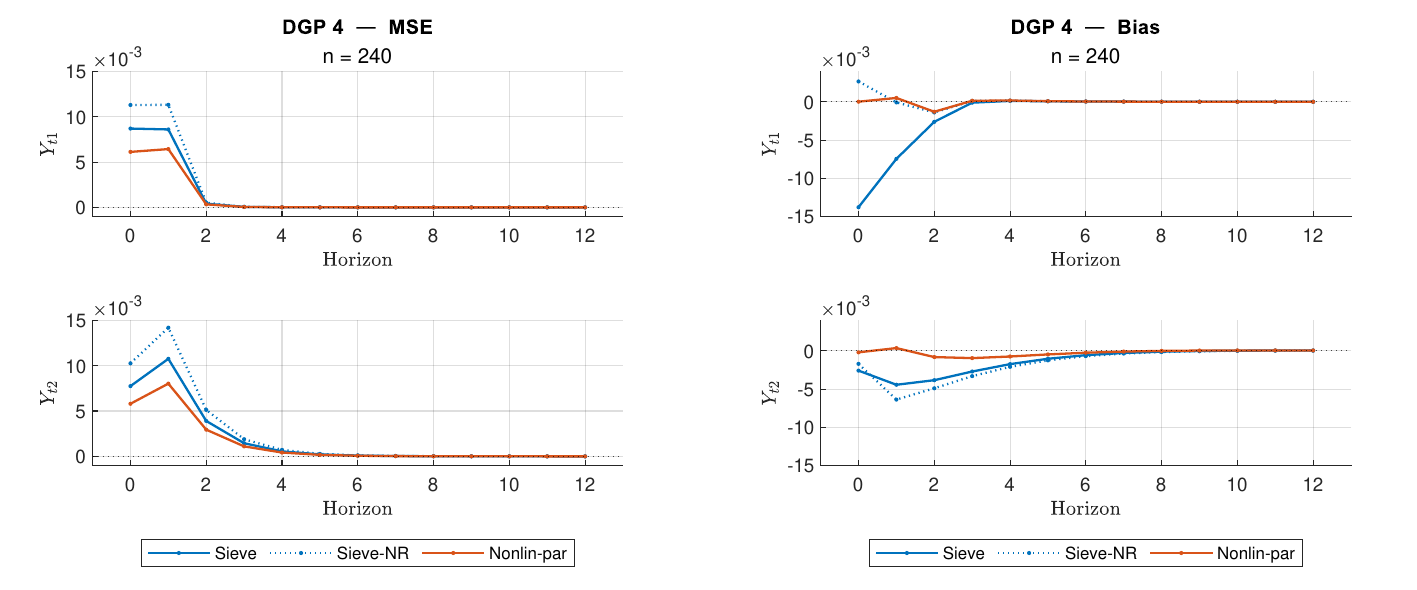}
	\end{subfigure}
	\\[5pt]
	\begin{subfigure}[b]{\textwidth}
		\centering
		\includegraphics[width=\textwidth]{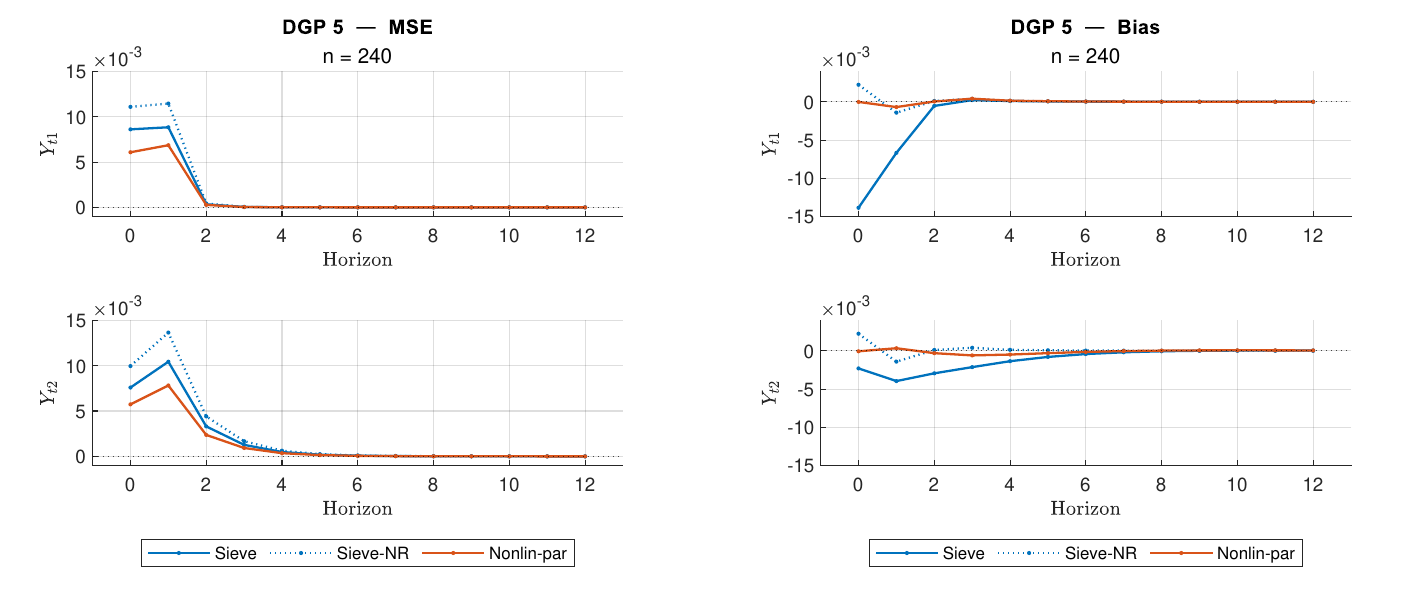}
	\end{subfigure}
	\\[5pt]
	\begin{subfigure}[b]{\textwidth}
		\centering
		\includegraphics[width=\textwidth]{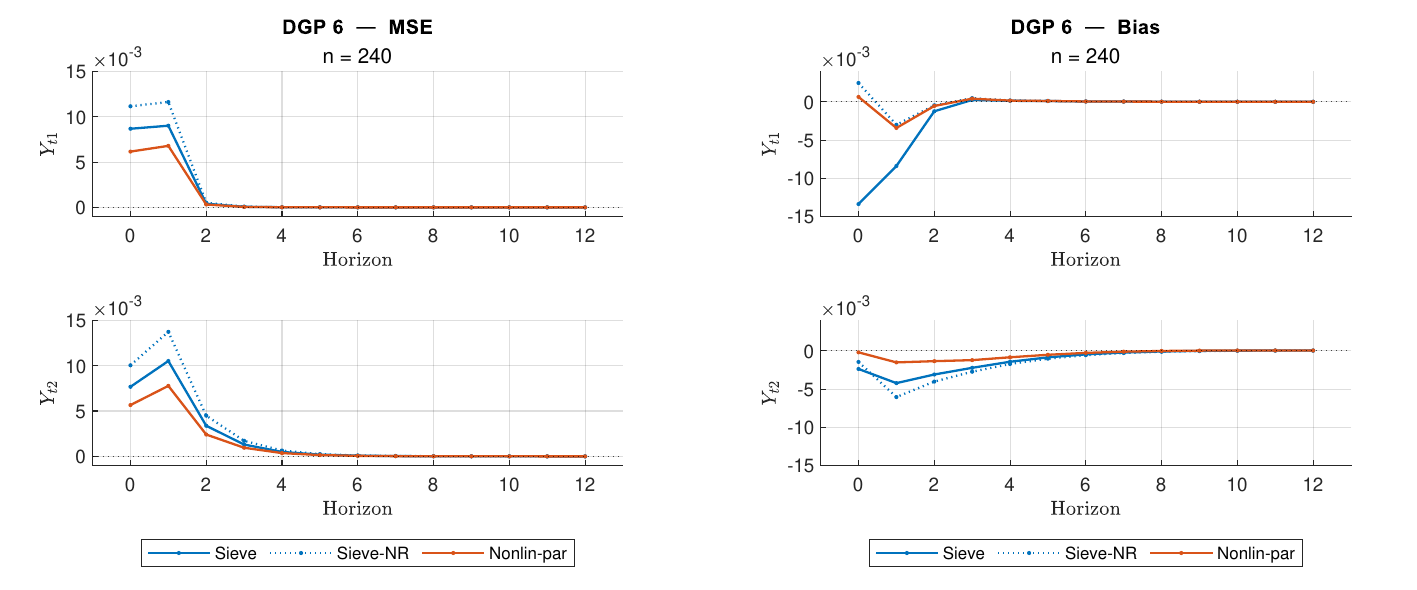}
	\end{subfigure}
	\\[5pt]
	\caption{Simulations results for DGPs 4-6 targeting population IRFs \textit{without relaxation} of the shock.}
	\label{fig:mse_bias_DGP_4_5_6__norelax}
\end{figure}

\newpage
\begin{figure}[H]
    \centering
    \begin{subfigure}[b]{\textwidth}
		\centering
		\includegraphics[width=0.85\textwidth]{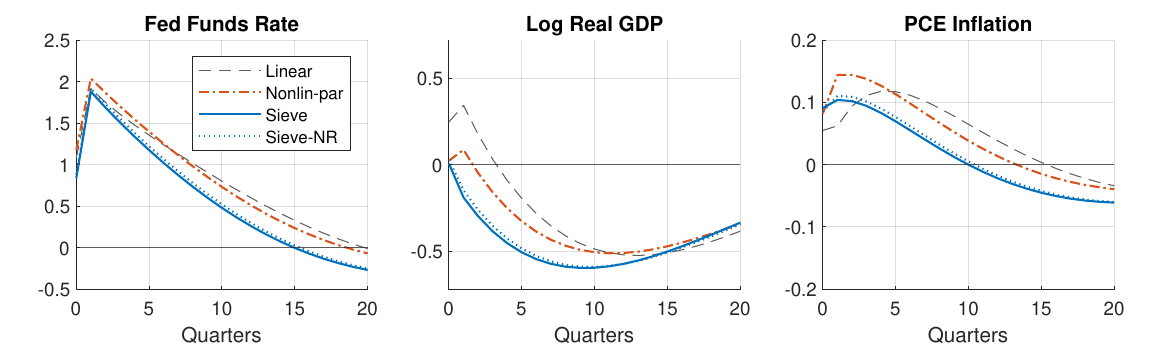}
		\caption{$\delta = +1$}
	\end{subfigure}
	\\[15pt]
    \begin{subfigure}[b]{\textwidth}
		\centering
		\includegraphics[width=0.85\textwidth]{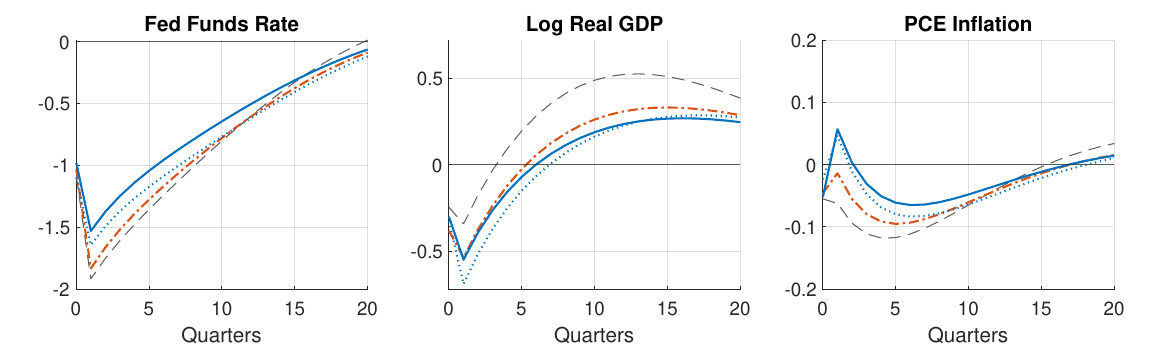}
		\caption{$\delta = -1$}
	\end{subfigure}
	\\[5pt]
    \caption{Effect of an unexpected U.S. monetary policy shock on federal funds rate, GDP, and inflation. This figure is a variation of Figure~\ref{fig:app_gonc2021_irfs} with the addition of semiparametric sieve IRFs obtained \textit{without relaxation} of the shock (blue, dotted).}
    \label{fig:app_gonc2021_irfs__norelax}
\end{figure}
\begin{figure}[H]
	\centering
	\begin{subfigure}[b]{0.45\textwidth}
		\centering
		\includegraphics[width=\textwidth]{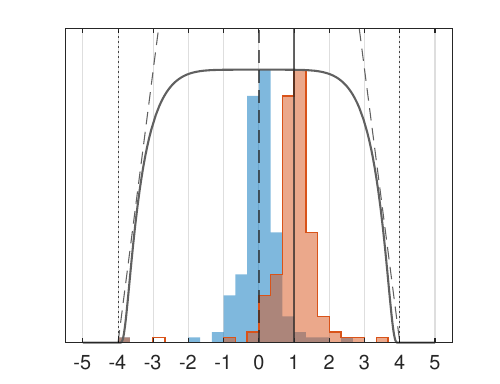}
		\caption{$\delta = +1$}
	\end{subfigure}
	\begin{subfigure}[b]{0.45\textwidth}
		\centering
		\includegraphics[width=\textwidth]{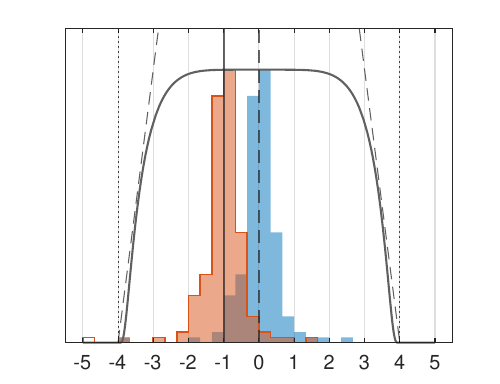}
		\caption{$\delta = -1$}
	\end{subfigure}
	\\[5pt]
	\caption{Comparison of histograms and shock relaxation function for a positive (left) and negative (right) shock in monetary policy. This figure is a variation of Figure~\ref{fig:plot_app_gonc2021_meanshift} with the addition of the histogram of shocked innovations \textit{without relaxation} (red outline).}
	\label{fig:plot_app_gonc2021_meanshift__norelax}
\end{figure}

\newpage
\begin{figure}[H]
    \centering
    \textbf{GDP}\\
    \includegraphics[width=0.85\textwidth]{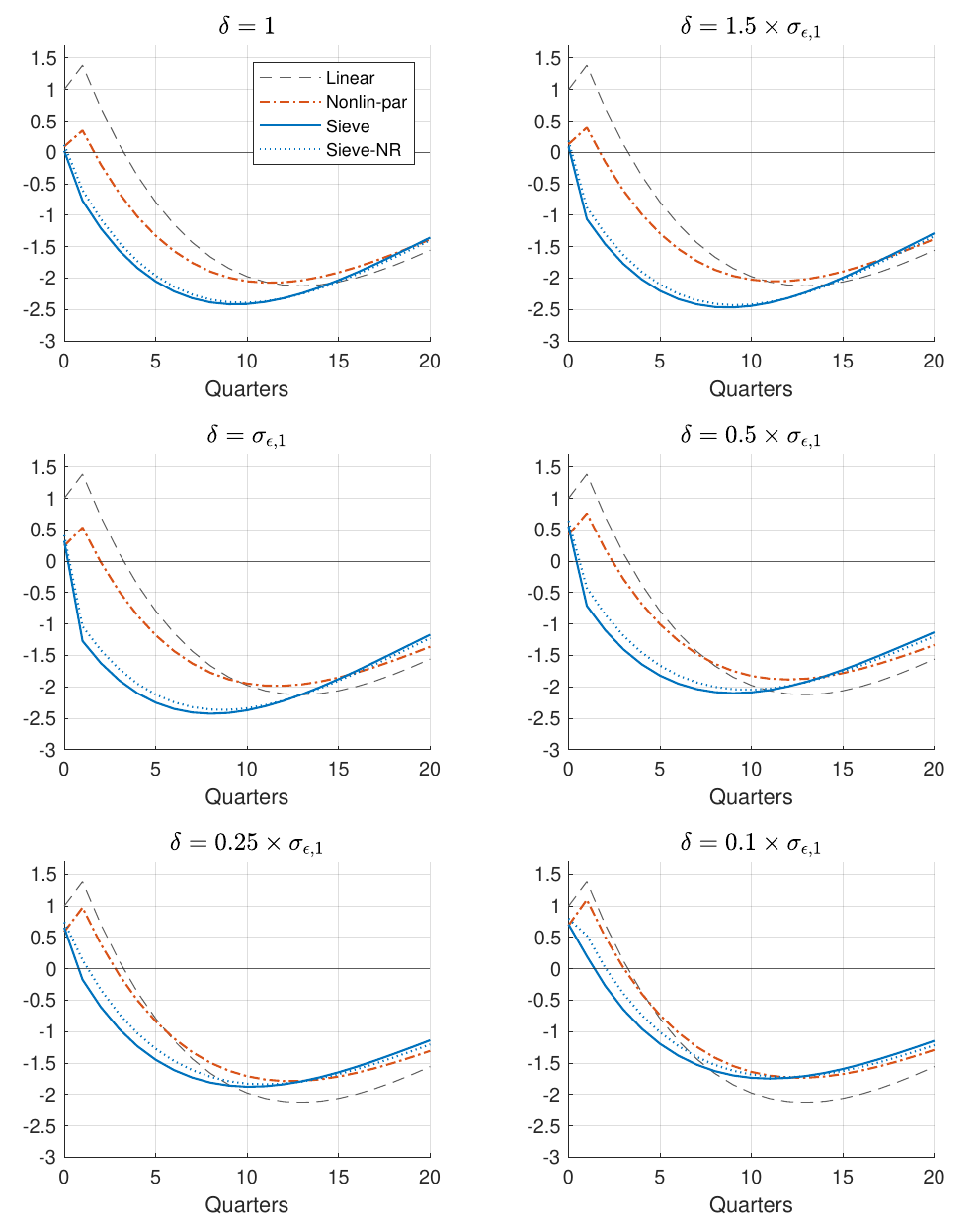} 
    \\[15pt]
    \caption{Relative changes in the GDP impulse responses function when the size of the shock is reduced from that used in Figure \ref{fig:app_gonc2021_irfs__norelax}. This figure is a variation of Figure~\ref{fig:plot_app_gonc2021_scale} with the addition of semiparametric sieve IRFs obtained \textit{without relaxation} of the shock (blue, dotted).}
    \label{fig:plot_app_gonc2021_scale___norelax}
\end{figure}

\newpage
\begin{figure}[H]
    \centering
	\includegraphics[width=\textwidth]{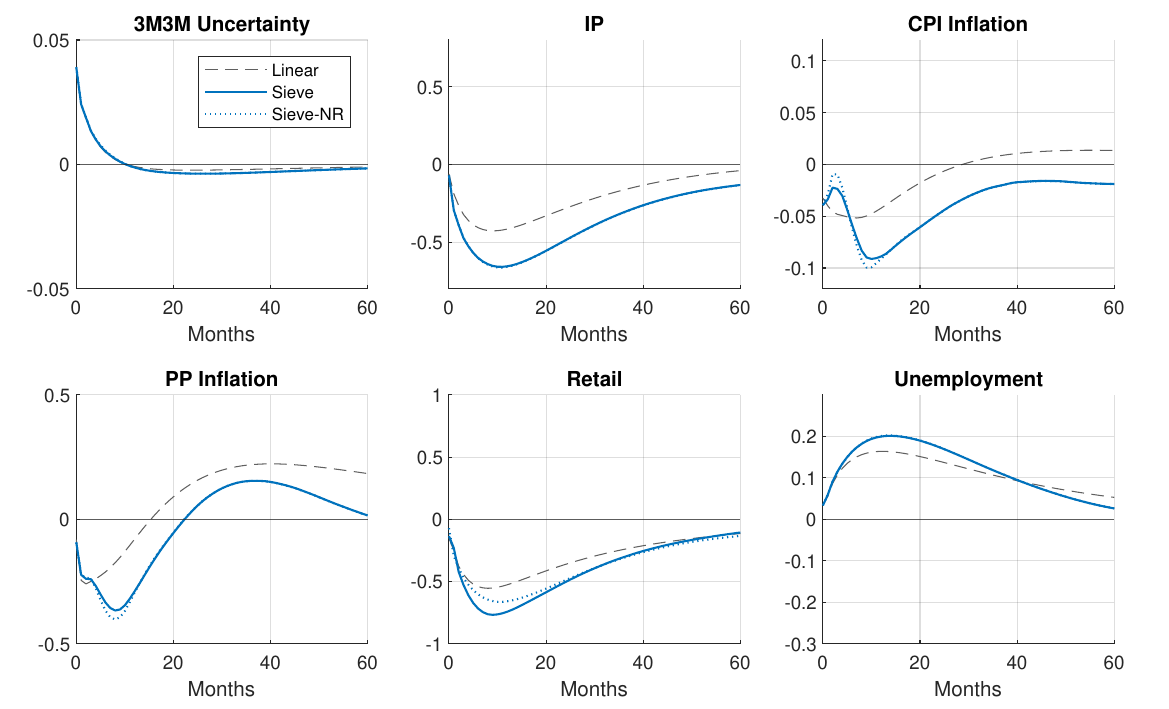}
	\\[5pt]
    \caption{Effect of an unexpected, one-standard-deviation uncertainty shock to U.S. macroeconomic variables. This figure is a variation of Figure~\ref{fig:app_istrefi2018_irfs} with the addition of semiparametric sieve IRFs obtained \textit{without relaxation} of the shock (blue, dotted).}
    \label{fig:app_istrefi2018_irfs__norelax}
\end{figure}
\begin{figure}[H]
    \centering
    \begin{subfigure}[b]{0.45\textwidth}
		\centering
		\includegraphics[width=\textwidth]{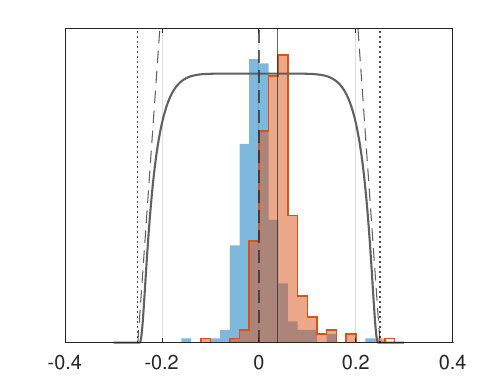}
		\caption*{$\delta = \sigma_\epsilon$}
	\end{subfigure}
    \\[5pt]
    \caption{Comparison of histograms and shock relaxation function for a positive, one-standard-deviation shock in interest rate uncertainty. This figure is a variation of Figure~\ref{fig:plot_app_istrefi2018_meanshift} with the addition of the histogram of shocked innovations \textit{without relaxation} (red outline).}
    \label{fig:plot_app_istrefi2018_meanshift__norelax}
\end{figure}

\newpage
\begin{figure}[H]
    \centering
    \textbf{IP}\\
    \includegraphics[width=0.85\textwidth]{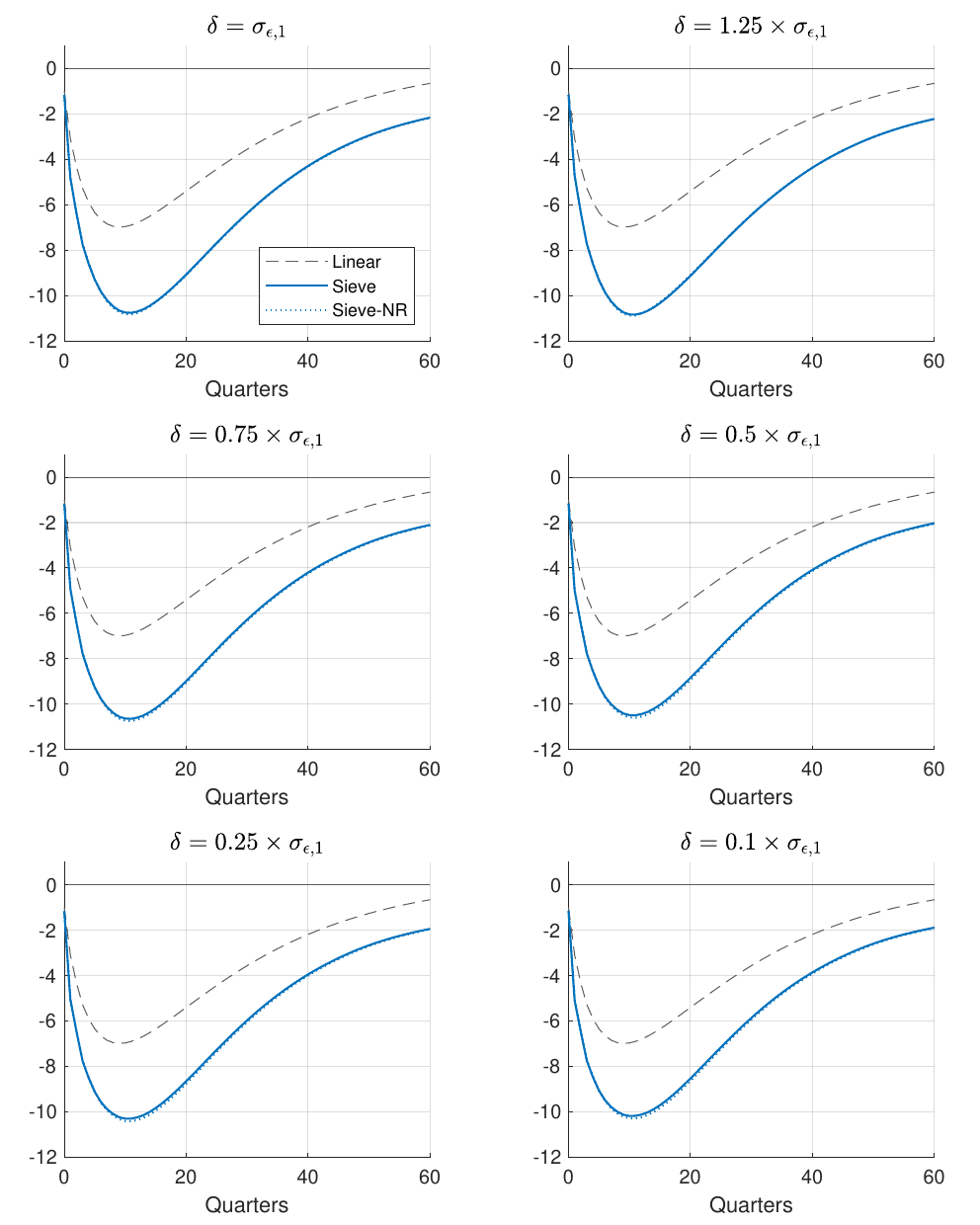} 
    \\[15pt]
    \caption{Relative changes in the industrial production impulse responses function when the size of the shock is reduced from that used in Figure \ref{fig:app_istrefi2018_irfs__norelax}. This figure is a variation of Figure~\ref{fig:plot_app_istrefi2018_scale_IP} with the addition of semiparametric sieve IRFs obtained \textit{without relaxation} of the shock (blue, dotted).}
    \label{fig:plot_app_istrefi2018_scale_IP__norelax}
\end{figure}

\end{document}